\def\VersionLong{}
\ifdefined\VersionLong%
	\newcommand{\LongVersion}[1]{#1}
	\newcommand{\ShortVersion}[1]{}
\else
	\newcommand{\LongVersion}[1]{}
	\newcommand{\ShortVersion}[1]{#1}
\fi

\LongVersion{%
\documentclass[twocolumn]{paper}
}
\ShortVersion{%
\documentclass[lettersize,journal]{IEEEtran}
}
\usepackage{amsmath,amsfonts}
\usepackage{array}
\usepackage{textcomp}
\usepackage{stfloats}
\usepackage{url}
\usepackage{verbatim}
\usepackage{graphicx}
\ShortVersion{\usepackage{cite}}

\usepackage[utf8]{inputenc}
\usepackage[english]{babel} %
\usepackage{csquotes}
\usepackage{setspace}

\usepackage[ruled,lined,linesnumbered,noend]{algorithm2e}
\SetKwInOut{Input}{input}
\SetKwInOut{Output}{output}
\SetKw{KwPush}{push}
\SetKw{KwPop}{pop}
\SetKw{KwFrom}{from}
\SetKw{KwCompute}{compute}
\SetKw{KwOr}{or}
\SetKw{KwAnd}{and}

\usepackage{subcaption}
\usepackage{paralist} %
\usepackage{xspace}

\newenvironment{ienumeration}
	{\ifdefined\VersionLong\begin{enumerate}\else\begin{inparaenum}[\itshape i\upshape)]\fi}
	{\ifdefined\VersionLong\end{enumerate}\else\end{inparaenum}\fi}

\newenvironment{oneenumeration}
	{\ifdefined\VersionLong\begin{enumerate}\else\begin{inparaenum}[1)]\fi}
	{\ifdefined\VersionLong\end{enumerate}\else\end{inparaenum}\fi}

\usepackage{amsthm} %
\usepackage{amssymb} %
\usepackage{mathtools} %
\usepackage{multirow}

\usepackage{listings}
\usepackage{color}

	\definecolor{mygreen}{rgb}{0,0.6,0}
	\definecolor{mygray}{rgb}{0.5,0.5,0.5}
	\definecolor{mymauve}{rgb}{0.58,0,0.82}
	\definecolor{weborange}{RGB}{255,165,0}

\lstdefinestyle{log}{
	backgroundcolor=\color{white},   %
	basicstyle=\scriptsize,        %
	breakatwhitespace=false,         %
	breaklines=true,                 %
	captionpos=b,                    %
	commentstyle=\color{mygreen},    %
	deletekeywords={...},            %
	escapeinside={\%*}{*)},          %
	extendedchars=true,              %
	frame=single,	                   %
	keepspaces=true,                 %
	keywordstyle=\color{red!70!black}\bfseries,       %
	morekeywords={@, open, close, update},            %
	numbers=left,                    %
	numbersep=5pt,                   %
	numberstyle=\tiny\color{mygray}, %
	rulecolor=\color{black},         %
	showspaces=false,                %
	showstringspaces=false,          %
	showtabs=false,                  %
	stepnumber=1,                    %
	stringstyle=\color{mymauve},     %
	tabsize=2,	                   %
	classoffset=1, %
	otherkeywords={@},
	morekeywords={@},
	keywordstyle=\color{weborange},
	classoffset=0,
}

\LongVersion{%
	\usepackage[backend=biber,backref=true,style=alphabetic,url=false,doi=true,defernumbers=true,sorting=anyt,maxnames=99]{biblatex} %
	\addbibresource{HyperPTCTL.bib}

	\DeclareSourcemap{
		\maps[datatype=bibtex]{
			\map[overwrite]{
				\step[fieldsource=toappear, match={true},
					fieldset=note, fieldvalue={To appear.~{}}, append]
			}
		}
	}

	\renewbibmacro*{doi+eprint+url}{%
		\iftoggle{bbx:doi}
			{\color{black!40}\footnotesize\printfield{doi}}
			{}%
		\newunit\newblock
		\iftoggle{bbx:eprint}
			{\usebibmacro{eprint}}
			{}%
		\newunit\newblock
		\iftoggle{bbx:url}
			{\usebibmacro{url+urldate}}
			{}%
	}

}
\usepackage[svgnames,table]{xcolor}
\definecolor{USPNcobalt}{HTML}{293358}
\definecolor{USPNocre}{HTML}{8b7d6d}
\definecolor{USPNblanc}{HTML}{ffffff}
\definecolor{USPNceruleen}{HTML}{354878}
\definecolor{USPNsable}{HTML}{ad947e}

\usepackage[
		pdfauthor={Masaki Waga and Etienne Andre},%
		pdftitle={Hyper parametric timed CTL},
		breaklinks  = true,
		colorlinks  = true,
	\ifdefined \VersionWithComments
	\fi
		citecolor   = USPNsable,
		linkcolor   = USPNocre,
		urlcolor    = USPNceruleen,
	]{hyperref}

\usepackage[capitalise,english,nameinlink]{cleveref} %
\crefname{line}{\text{line}}{\text{lines}} %
\crefname{assumption}{\text{Assumption}}{\text{Assumptions}} %

\usepackage{wrapfig}
\usepackage{tikz}
\usetikzlibrary{arrows,automata,positioning,math}
\tikzstyle{every node}=[initial text=]
	\tikzstyle{location}=[circle, minimum size=12pt, draw=black, fill=blue!10, inner sep=1pt] %
\tikzstyle{final}=[double]
\tikzstyle{accepting}=[final]
\tikzstyle{transition}=[->]
\tikzstyle{PTPMOPT}=[,dashed,color=red,semithick]

\tikzstyle{pta}=[auto, ->, >=stealth']
\tikzstyle{invariant}=[draw=black, dashed, inner sep=1pt]
\tikzstyle{urgent}=[dotted, draw=red, very thick]

\newcommand{\styleact}[1]{\ensuremath{\textcolor{coloract!80!black}{\mathrm{#1}}}}
\newcommand{\styleclock}[1]{\ensuremath{\textcolor{colorclock!80!black}{#1}}}
\newcommand{\styleloc}[1]{\ensuremath{\mathrm{#1}}}
\newcommand{\styleparam}[1]{\ensuremath{\textcolor{colorparam!80!black}{#1}}} %

\newcommand{\stylepathvar}[1]{\ensuremath{\textcolor{colorvar!80!black}{#1}}}
	\definecolor{coloract}{rgb}{0.50, 0.70, 0.30}
	\definecolor{colorclock}{rgb}{0.4, 0.4, 1}
	\definecolor{colorconst}{rgb}{0.50, 0.20, 0.00}
	\definecolor{colordisc}{rgb}{1, 0, 1}
	\definecolor{colorloc}{rgb}{0.4, 0.4, 0.65}
	\definecolor{colorparam}{rgb}{1, 0.6, 0.0}
	\definecolor{colorvar}{rgb}{0.6, 0.7, 1}
	\definecolor{colorlvar}{rgb}{0.4, 0.4, .5}
	\definecolor{colordparam}{rgb}{.9, 0.8, 0.0}
\newif\iftikzgnuplot
\tikzgnuplottrue 
\usepackage{gnuplot-lua-tikz}
\usepackage{pgfplotstable}
\pgfplotsset{compat=1.12}
\usepackage{booktabs}

\newcommand{\init}{_0}

\newcommand{\set}[1]{\ensuremath{\left\{#1\right\}}}

\newcommand{\A}{\ensuremath{\mathcal{A}}}
\newcommand{\Actions}{\Sigma}
\newcommand{\action}{\ensuremath{\sigma}}

\newcommand{\C}{C}
\newcommand{\Clock}{\mathbb{C}} %
\newcommand{\ClockCard}{H} %
\newcommand{\clock}{c} %

\newcommand{\clockval}{\nu} %
\newcommand{\ClocksZero}{\vec{0}_{\Clock}}
\newcommand{\compOp}{\bowtie}
\newcommand{\productOp}{\times}
\newcommand{\composeOp}{\mathbin{||}}

\newcommand{\domain}{\mathbf{dom}}
\newcommand{\duration}{\mathit{Dur}}
\newcommand{\edge}{e}
\newcommand{\Edges}{E}
\newcommand{\longuefleche}[1]{\stackrel{#1}{\longrightarrow}}
\newcommand{\longueflecheRel}[1]{\stackrel{#1}{\mapsto}}

\newcommand{\flecheRel}{{\rightarrow}}

\newcommand{\guard}{g}
\newcommand{\Init}{\mathit{Init}}
\newcommand{\invariant}{I}
\newcommand{\K}{K}
\newcommand{\Label}{\Lambda} %

\newcommand{\loc}{\ell} %
\newcommand{\locinit}{\loc\init}
\newcommand{\Loc}{L} %
\newcommand{\LocInit}{\Loc\init}

\newcommand{\lterm}{\mathit{lt}}
\newcommand{\ltermNN}{\mathit{lt}_{\ge 0}}
\newcommand{\ltermCnt}{\mathit{cnt}}
\newcommand{\ltermCntNN}{\mathit{cnt}_{\ge 0}}
\newcommand{\Param}{\mathbb{P}} %
\newcommand{\param}{p} %
\newcommand{\paramOrInt}{\ensuremath{\gamma}}
\newcommand{\parami}[1]{p_{#1}} %
\newcommand{\ParamCard}{M} %
\newcommand{\PVal}[1][\Param]{({\setQnn})^{#1}}
\newcommand{\pval}{v} %
\newcommand{\R}{{\mathbb{R}}}
\newcommand{\reduce}{\mathit{reduce}}
\newcommand{\Rgeqzero}{\R_{\geq 0}}

\newcommand{\Time}{\mathbb{T}} %
\newcommand{\Traces}{\mathit{Paths}}
\newcommand{\emptyruns}{\runs_{\emptyset}}
\newcommand{\props}{\ensuremath{\mathit{pr}}}
 \newcommand{\removedEdges}[2]{{\mathcal{D}^{#2}_{#1}}}
\newcommand{\runs}{\Pi}
\newcommand{\vrun}{\pi}
\newcommand{\pathOrder}[1][\runs]{\mathbin{\trianglelefteq_{#1}}}
\newcommand{\npathOrder}[1][\runs]{\mathbin{\ntrianglelefteq_{#1}}}
\newcommand{\pathOrderToN}{\mathcal{N}_{\pathOrder}}
\newcommand{\PathVar}{\mathcal{V}}
\newcommand{\trace}{\rho}
\newcommand{\sinit}{s\init} %
\newcommand{\state}{\ensuremath{s}} %
\newcommand{\States}{S} %
\newcommand{\SInit}{S\init} %
\newcommand{\suffixEqOp}{\succeq}
\newcommand{\suffixOp}{\succ}
\newcommand{\TTS}[1][\A]{T_{#1}} %
\newcommand{\TTSLabel}{\Lambda} %

\newcommand{\setN}{{\mathbb N}}
\newcommand{\setNpos}{\setN_{>0}}
\newcommand{\setQ}{{\mathbb Q}}
\newcommand{\setQnn}{\setQ_{\geq 0}}
\newcommand{\setR}{{\mathbb{R}}}
\newcommand{\setRnn}{\ensuremath{\setR_{\geq 0}}}
\newcommand{\setZ}{{\mathbb Z}}

\newcommand{\partfun}{\nrightarrow}
\newcommand{\pluseq}{\mathrel{{+}{=}}}

\newcommand{\resets}{R}
\newcommand{\project}[2]{\ensuremath{#1{|_{#2}}}}
\newcommand{\reset}[2]{\ensuremath{[#1]_{#2}}}
\newcommand{\UpdateCountval}[3]{[{#1}]_{{#2} - {#3}}}
\newcommand{\UpdateRecordval}[3]{[{#1}]_{{#2} - {#3}}}
\newcommand{\valuate}[2]{\ensuremath{#2(#1)}}

\newcommand{\tuple}[1]{\ensuremath{\langle #1 \rangle}}

\newcommand{\disjointUnion}{\sqcup}
\newcommand{\powerset}[1]{\mathcal{P}(#1)}

\newcommand{\CTLA}{\ensuremath{\forall}}
\newcommand{\CTLE}{\ensuremath{\exists }}
\newcommand{\CTLF}{\ensuremath{\Diamond}}
\newcommand{\CTLG}{\ensuremath{\square}}
\newcommand{\CTLU}{\ensuremath{\mathcal{U}}}
\newcommand{\CTLEF}{\ensuremath{\CTLE \CTLF}}
\newcommand{\CTLAF}{\ensuremath{\CTLA \CTLF}}

\newcommand{\fml}{\varphi}
\newcommand{\fullFml}{\psi}
\newcommand{\Boolean}{\mathcal{B}}

\newcommand{\LAST}{\mathit{LAST}}
\newcommand{\COUNT}{\mathit{COUNT}}
\newcommand{\LASTExpr}{\textsc{LASTExpr}}
\newcommand{\COUNTNNExpr}{\ensuremath{\textsc{CountExpr}_{\ge 0}}}
\newcommand{\COUNTModExpr}{\ensuremath{\textsc{CountExpr}_{\mathrm{mod}}}}
\newcommand{\CountActions}{\Actions_{\COUNT}}
\newcommand{\RISING}{\mathit{Rising}}
\newcommand{\Vars}{\mathit{Vars}}

\newcommand{\EUntilOp}[1]{\mathbin{\exists\CTLU_{#1}}}
\newcommand{\AUntilOp}[1]{\mathbin{\forall\CTLU_{#1}}}
\newcommand{\UntilOp}[1]{\mathbin{\CTLU_{#1}}}
\newcommand{\HEUntilFml}[4][\changed{\ge 0}]{\exists{#2}.\, \allowbreak{#3} \allowbreak \UntilOp{#1}\allowbreak {#4}} %
\newcommand{\HAUntilFml}[4][\changed{\ge 0}]{\forall{#2}.\, {#3} \UntilOp{#1} {#4}} %
\newcommand{\Release}[1]{\mathbin{\mathcal{R}_{#1}}}
\newcommand{\HEReleaseFml}[4][\changed{\ge 0}]{\exists{#2}.\, {#3} \Release{#1} {#4}} %
\newcommand{\HAReleaseFml}[4][\changed{\ge 0}]{\forall{#2}.\, {#3} \Release{#1} {#4}} %
\newcommand{\WeakUntilOp}[1]{\mathbin{\mathcal{W}_{#1}}}
\newcommand{\HEWeakUntilFml}[4][\changed{\ge 0}]{\exists{#2}.\, {#3} \WeakUntilOp{#1} {#4}} %
\newcommand{\HAWeakUntilFml}[4][\changed{\ge 0}]{\forall{#2}.\, {#3} \WeakUntilOp{#1} {#4}} %
\newcommand{\DiaOp}[1]{\CTLF_{#1}}

\newcommand{\HEDiaFml}[3][\changed{\ge 0}]{\exists{#2}.\, \DiaOp{#1} {#3}} %
\newcommand{\HADiaFml}[3][\changed{\ge 0}]{\forall{#2}.\, \DiaOp{#1} {#3}} %
\newcommand{\BoxOp}[1]{\CTLG_{#1}}

\newcommand{\HEBoxFml}[3][\changed{\ge 0}]{\exists{#2}.\, \BoxOp{#1} {#3}} %
\newcommand{\HABoxFml}[3][\changed{\ge 0}]{\forall{#2}.\, \BoxOp{#1} {#3}} %
\newcommand{\countval}{\eta}
\newcommand{\approxcountval}{\tilde{\eta}}
\newcommand{\CountZero}{\vec{0}_{\mathrm{cnt}}}
\newcommand{\recordval}{\theta}

\newcommand{\RecordZero}{\vec{0}_{\mathrm{rec}}}

\newcommand{\PExists}{\styleMacro{\ensuremath{\tilde{\exists}}}}

\ifdefined\VersionWithComments
	\newcommand{\styleMacro}[1]{\textcolor{USPNocre!60!black}{#1}}
\else
	\newcommand{\styleMacro}[1]{#1}
\fi

\newcommand{\PTCTL}{\styleMacro{PTCTL}}

\newcommand{\ExistsPTCTL}{\styleMacro{$\PExists$\PTCTL{}}}

\newcommand{\HyperPTCTL}{\styleMacro{Hyper\PTCTL{}}}
\newcommand{\ExtHyperPTCTL}{\styleMacro{Ext-\HyperPTCTL{}}}
\newcommand{\ExtPTCTL}{\styleMacro{Ext-\PTCTL{}}}
\newcommand{\ExtOrNotHyperPTCTL}{\styleMacro{(Ext-)\HyperPTCTL{}}}

\newcommand{\NFHyperPTCTL}{\styleMacro{Nest-Free \HyperPTCTL{}}}
\newcommand{\NFExtHyperPTCTL}{\styleMacro{Nest-Free \ExtHyperPTCTL{}}}
\newcommand{\NFExtPTCTL}{\styleMacro{Nest-Free \ExtPTCTL{}}}

\newcommand{\NFExtHyperPTCTLBDR}{\styleMacro{Nest-Free RP\ExtHyperPTCTL{}}}
\newcommand{\NFExistsExtHyperPTCTLBDR}{\styleMacro{Nest-Free $\PExists$RP\ExtHyperPTCTL{}}}

\newcommand{\ExistsHyperPTCTL}{\styleMacro{$\PExists$\HyperPTCTL{}}}

\newcommand{\NFExistsHyperPTCTL}{\styleMacro{Nest-Free \ExistsHyperPTCTL{}}}

\newcommand{\NFExistsExtReachHyperPTCTL}{\styleMacro{Nest-Free $\PExists$-$\exists\Diamond$\ExtHyperPTCTL{}}}

\newcommand{\styleComplexity}[1]{\textsc{#1}}
\newcommand{\EXPTIME}{\styleComplexity{ExpTime}}
\newcommand{\NEXPTIME}[1]{\styleComplexity{#1\EXPTIME}}
\newcommand{\EXPSPACE}{\styleComplexity{ExpSpace}}
\newcommand{\PSPACE}{\styleComplexity{PSpace}}

\newcommand{\observerof}[1]{\mathcal{O}_{#1}}
\newcommand{\rmextfrom}[1]{{#1}_{\mathrm{noext}}}
\newcommand{\extof}[1]{\mathrm{Ext}({#1})}
\usepackage{pifont}%

\newcommand{\arbitrarilyMany}{\ensuremath{*}}

\newcommand{\recallResult}[2]
{%
	\smallskip

	\noindent\fcolorbox{black}{green!15}{
		\begin{minipage}{.95\columnwidth}
			\noindent\textbf{\cref{#1} (recalled).}
			{\em{}#2}
		\end{minipage}
	}

	\smallskip
}

\definecolor{vertfonce}{rgb}{0.0, 0.5, 0.0}
\definecolor{rougefonce}{rgb}{1, 0.0, 0.0}

\theoremstyle{plain}
\newtheorem{lemma}{Lemma}
\newtheorem{proposition}[lemma]{Proposition}
  \newtheorem{theorem}[lemma]{Theorem}

\theoremstyle{definition}
\newtheorem{definition}[lemma]{Definition}
\ShortVersion{%
\newtheorem{example}[lemma]{Example}
}
\theoremstyle{remark}
\newtheorem{remark}[lemma]{Remark}

\newcommand{\changed}[1]{#1}

\usepackage{verbatim} %

\newcommand{\ourTool}{\textsf{HyPTCTLchecker}} %
\newcommand{\imitator}{\textsf{IMITATOR}}

\newcommand{\Running}{\textsf{ClkGen}}
\newcommand{\Coffee}{\textsf{Coffee}}
\newcommand{\STAC}{\textsf{STAC}}
\newcommand{\WFAS}[2][]{\ensuremath{\textsf{WFAS}^{#1}_{#2}}}
\newcommand{\ATM}{\textsf{ATM}}

\newcommand{\FIFO}{\textsf{FIFO}}
\newcommand{\Priority}{\textsf{Priority}}
\newcommand{\RoundRobin}{\textsf{R.R.}}
\newcommand{\ClockDeviation}{\textsf{Deviation}}
\newcommand{\Opacity}{\textsf{Opacity}}
\newcommand{\Unfairness}{\textsf{Unfair}}
\newcommand{\RobustObservationalNonDeterminism}{\textsf{RobOND}}
\newcommand{\EF}[1]{\ensuremath{\textsf{EF}_{#1}}}
\usepackage{colortbl}

\newcommand{\TIMEOUT}{{\color{red}\textbf{T.O.}}}

\newcommand{\eg}{e.g.,\xspace}

\newcommand{\ie}{i.e.,\xspace}
\newcommand{\st}{s.t.\xspace}

\newcommand{\resp}{resp.\xspace}

\tikzstyle{rqanswer} = [
 draw=black,
 fill=gray!30,
 text=black,
 line width=0.5pt,
  text width = \linewidth - 1.6 ex - 1pt,
  inner sep = 0.8 ex,
  rounded corners=4pt]
\newcommand{\rqanswer}[2]{%
	\smallskip

	\noindent%
	\begin{tikzpicture}%
	\draw node[rqanswer]{\textbf{Answer to {#1}}:{ #2}};%
	\end{tikzpicture}%
}

\newcommand{\defProblem}[3]
{%
	\smallskip

	\noindent%
	\begin{tikzpicture}%
	\draw node[rqanswer]{
		\small%
		\textbf{#1 problem:}\\
		\textsc{Input}: #2\\
		\textsc{Problem}: #3
	};%
	\end{tikzpicture}%
	\smallskip
}

\usepackage{version}	%

\begin{document}

\title{Hyper parametric timed CTL}

\author{Masaki~Waga and
Étienne André
\thanks{This work is \LongVersion{partially }supported by JST PRESTO Grant No.\ JPMJPR22CA, JST CREST Grant No.\ JPMJCR2012, JSPS KAKENHI Grant No.\ 22K17873, ANR ProMiS (ANR-19-CE25-0015), and ANR BisoUS (ANR-22-CE48-0012).}%
\thanks{M.\ Waga is with Kyoto University, Kyoto, Japan. \'E.\ Andr\'e is with Université Sorbonne Paris Nord, LIPN, CNRS UMR 7030, F-93430 Villetaneuse, France.}%
\LongVersion{%
\thanks{This is the author (and extended) version of the manuscript of the same name published in IEEE Transactions on Computer-Aided Design of Integrated Circuits and Systems (TCAD).
The final version is available at \url{https://ieeexplore.ieee.org/}.}}
}

\LongVersion{%
  \date{}
}

\ShortVersion{%
\markboth{Journal of \LaTeX\ Class Files,~Vol.~14, No.~8, August~2021}%
{Shell \MakeLowercase{\textit{et al.}}: A Sample Article Using IEEEtran.cls for IEEE Journals}
}

\maketitle

\setcounter{footnote}{0}

\LongVersion{
	\thispagestyle{plain}
}

\ifdefined \VersionWithComments
	\textcolor{red}{\textbf{This is the version with comments. To disable comments, comment out line~3 in the \LaTeX{} source.}}
	\textcolor{orange}{\LongVersion{This is the long version.}
	\ShortVersion{This is the short version.}
	}
\fi

\begin{abstract}
	Hyperproperties enable simultaneous reasoning about multiple execution traces of a system and are useful to reason about non-interference, opacity, robustness, fairness, observational determinism, etc.
	We introduce \emph{hyper parametric timed computation tree logic (HyperPTCTL)}, extending hyperlogics with timing reasoning and, notably, parameters to express unknown values.
	We mainly consider its nest-free fragment, where temporal operators cannot be nested.
	However, we allow extensions that enable counting actions and comparing the duration since the most recent occurrence of specific actions.
	We show that our nest-free fragment with this extension is sufficiently expressive to encode properties, \eg{} opacity, (un)fairness, or robust observational (non-)determinism.
	We propose semi-algorithms for model checking and synthesis of parametric timed automata (an extension of timed automata with timing parameters) against this nest-free fragment with the extension via reduction to PTCTL model checking and synthesis.
	While the general model checking (and thus synthesis) problem is undecidable, we show that a large part of our extended (yet nest-free) fragment is decidable, provided the parameters only appear in the property, not in the model.
	We also exhibit additional decidable fragments where parameters within the model are allowed.
	We implemented our semi-algorithms on top of the \imitator{} model checker, and performed experiments.
	Our implementation supports most of the nest-free fragments (beyond the decidable classes).
	The experimental results highlight our method's practical relevance.
\end{abstract}

\section{Introduction}\label{section:introduction}
Parametric timed automata (PTAs)~\cite{AHV93} is an extension of finite-state automata for modeling and verification of real-time systems, where the timing constraints are not fixed but parameterized.
\changed{PTAs} extend the concept of timed automata (TAs)~\cite{AD94} by introducing parameters into time bounds, allowing for analyzing a system \LongVersion{behavior }across a range of timing scenarios.

Hyperproperties enable reasoning simultaneously about multiple execution traces of a system and turn useful to reason about non-interference, opacity, fairness, robustness, observational determinism, etc.
We introduce \emph{hyper parametric timed computation tree logic} (\HyperPTCTL{}), extending hyperlogics with not only timing reasoning but also timing parameters able to express unknown values.
\HyperPTCTL{} can be used typically to reason about multiple traces on PTAs.

After defining the syntax and semantics of general \HyperPTCTL{}, we mainly consider the nest-free fragment, where temporal operators cannot be nested.
However, we extend \HyperPTCTL{} with additional predicates that enable counting actions and comparing the duration since the most recent occurrence of specific actions using diagonal constraints of the form $\LAST(\action_{\vrun_1}) - \LAST(\action_{\vrun_2})$\changed{, where $\action$ is a proposition and $\vrun_1$ $\vrun_2$ represent two paths}.
Even without the nesting of temporal operators, we demonstrate that this extension enables encoding classical properties such as opacity, (un)fairness, or observational (non-) determinism---in a timed and parametric setting.
For example, we can use a simple \HyperPTCTL{} formula to encode a \emph{robust observational non-determinism}: ``By giving the same sequence of inputs at the same timing to the system, it is possible to get the same sequence of outputs but with large time difference''.
A timing parameter in the formula is used to leave the time difference unspecified, and, for example, the feasible values can be synthesized (by our semi-algorithm).
\LongVersion{The concrete formula is shown in \cref{example:observational-non-determinism}.}
We denote this nest-free but extended fragment by \NFExtHyperPTCTL{}.

We consider two problems over parametric formulas and/or models:
\begin{itemize}
	\item The \emph{model checking} problem asks whether there exists a valuation for which the model satisfies the formula; and %
	\item The \emph{synthesis} problem asks for the exact valuations set for which the model satisfies the formula. %
		Ideally, this representation should be given symbolically, \eg{} in a decidable logical formalism. %
\end{itemize}
\LongVersion{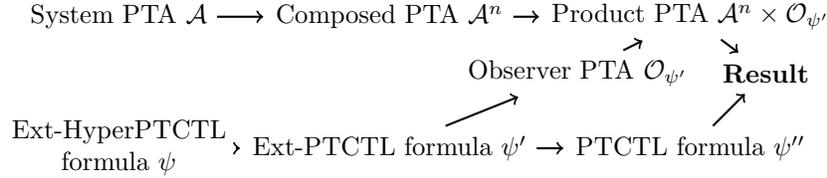
\begin{figure*}[tp]}%
\ShortVersion{\begin{figure}[tp]}%
\centering
\scalebox{\ShortVersion{.8}\LongVersion{1}}{%
	\begin{tikzpicture}[shorten >=1pt,every node/.style={transform shape}, thick]
	\node[align=center] (system) at (0.4,1.75) {System PTA $\A$};
	\node[align=center] (spec) at (0.4,0.0) {\ExtHyperPTCTL{}\\ formula $\fullFml$};

  \node[align=center] (composed system) at (4.0,1.75) {Composed PTA $\A^n$};
  \node[align=center] (extptctl) at (4.0,0.0) {\ExtPTCTL{} formula $\fullFml'$};

  \node[align=center] (observer) at (6.5,1.0) {Observer PTA $\observerof{\fullFml'}$};

  \node[align=center] (product system) at (8.0,1.75) {Product PTA $\A^n \productOp \observerof{\fullFml'}$};

  \node[align=center] (ptctl) at (8.0,0.0) {\PTCTL{} formula $\fullFml''$};

  \node[align=center] (result) at (9.0,1.0) {\textbf{Result}};

  \path[->] (extptctl) edge node{} (observer)
            (spec) edge node{} (extptctl)
            (system) edge node{} (composed system)
            (observer) edge node{} (product system)
            (composed system) edge node{} (product system)
            (extptctl) edge node{} (ptctl)
            (ptctl) edge node{} (result)
            (product system) edge node{} (result)
            ;
\end{tikzpicture}
}

\caption{Our reduction: \NFExtHyperPTCTL{} synthesis (\resp{} model checking) is reduced to \ExtPTCTL{} synthesis (\resp{} model checking) via self-composition; the extended predicates in \ExtPTCTL{} are evaluated by an observer PTA, which is composed with the PTA $\A^n$.}%
\label{figure:workflow}
\LongVersion{\end{figure*}}
\ShortVersion{\end{figure}}
We\LongVersion{ constructively} show that \NFExtHyperPTCTL{} model checking (\resp{} synthesis) of PTAs is reducible to PTCTL model checking (\resp{} synthesis) of PTAs.
\cref{figure:workflow} outlines our reduction. We show a more concrete working example later in \cref{section:example}.
First, we reduce \NFExtHyperPTCTL{} model checking (\resp{} synthesis) to \NFExtPTCTL{} model checking (\resp{} synthesis) by taking the self-composition $\A^n$ of the system PTA $\A$, where $n$ is the number of quantified path variables (\ie{} the number of simultaneously reasoned execution traces) in the \ExtHyperPTCTL{} formula $\fullFml$.
Then, we construct an \emph{observer} PTA $\observerof{\fullFml'}$~\cite{ABBL03} to evaluate the extended predicates in the given \NFExtPTCTL{} formula~$\fullFml'$.
We show that the result of \PTCTL{} model checking (\resp{} synthesis) for the product PTA $\A^n \times \observerof{\fullFml'}$ is the same as the result of the original problem.
Thus, the original problem is reduced to \PTCTL{} model checking (\resp{} synthesis).
By integrating this reduction with a semi-algorithm for PTCTL model checking (\resp{} synthesis), we derive a semi-algorithm for Nest-Free Ext-HyperPTCTL model checking (\resp{} synthesis).

While \NFExtHyperPTCTL{} model checking of PTAs is trivially undecidable due to the undecidability of reachability-emptiness of PTAs~\cite{AHV93},
we show that they are decidable for a large part of \NFExtHyperPTCTL{}, provided the parameters only appear in the property, not in the model.
We also exhibit additional decidable fragments where parameters in the model are allowed.

We implemented our approach on top of the existing \imitator{} parametric timed model checker~\cite{Andre21} and performed\LongVersion{ a set of} experiments. Our implementation \ourTool{} supports most of the nest-free fragment (beyond the decidable classes too---in which case at the risk of non-termination or approximated result). The experimental results show that our approach can handle various properties if the PTA has a moderate size.

\LongVersion{%
	\subsection{Contributions}
}%
	Our contributions are summarized as follows:
\begin{enumerate}
	\item We introduce \HyperPTCTL{} and its extension \ExtHyperPTCTL{} to count actions and to measure the time since their last occurrence\LongVersion{ and define their syntax and semantics} (\cref{section:hyper-ptcl});
	\item We propose semi-algorithms for \NFExtHyperPTCTL{} model checking (\resp{} synthesis) of PTAs\LongVersion{ via reduction to PTCTL model checking (\resp{} synthesis) of PTAs} (\cref{section:reduction});
	\item While \NFExtHyperPTCTL{} model checking and synthesis are trivially undecidable, we exhibit several decidable subclasses, with parameters either in the PTA or in the \NFExtHyperPTCTL{} formula (\cref{section:decidability});
	\item We implemented our approach and performed experiments. The experimental results suggest the practical relevance of our approach\LongVersion{, particularly for the PTAs with moderate size} (\cref{section:experiments}).
\end{enumerate}

To the best of our knowledge, our work is not only the first one extending TCTL into hyperlogics, but also the first one to allow for timing parameters in such a TCTL hyper-extension.

\section{Related works}\label{section:related}

\subsection{Model checking parametric timed formalisms}
First, model checking PTAs against the non-parametric nest-free fragment (without nested operators) of PTCTL is already undecidable, as reachability-emptiness (also called $\CTLEF{}$-emptiness, \ie{} the emptiness over the valuations set for which a given location can be reached) is undecidable for general PTAs over dense or discrete time~\cite{AHV93}\LongVersion{, including without any closed inequality~\cite{Doyen07}}.
Unavoidability(\CTLAF{})-emptiness is undecidable too~\cite{JLR15}.
\LongVersion{%
	As a consequence, model checking (P)TCTL against PTAs is undecidable in general too, and synthesis is intractable.
}

Reachability-emptiness over discrete time for PTAs with 2 parametric clocks\footnote{%
	A parametric clock is a clock compared to a parameter in at least one guard or invariant.
}, arbitrarily many non-parametric clocks and 1 parameter is \LongVersion{decidable~\cite{BO17} and }\EXPSPACE{}-complete~\cite{GH21}.

\LongVersion{%

	In~\cite{ET99}, a Parameterized Real-Time Computation Tree Logic (``PRTCTL'') is defined, with a decidable polynomial model checking algorithm for a fixed number of parameters; the timed model is however severely restricted when compared to~TAs. %
}

\LongVersion{%
	In~\cite{BR07}, model checking PTAs with a single parametric clock (and arbitrarily many non-parametric clocks, but this is equivalent over discrete time~\cite{AHV93}) against parametric TCTL is proved to be undecidable over discrete time.
	Also note that this negative result holds even when adding some further syntactic restrictions (see \cite[Remark~3.5]{BR07}).
	On the positive side, for the ``F-PTCTL'' fragment (in which equalities are forbidden in all operators, and inequalities with $\geq$ and~$>$ are forbidden in the $\forall U$ operators), model checking PTAs with a single parametric clock is proved to be decidable \cite[Corollary~3.11]{BR07}, and synthesis can be achieved \cite[Corollary~3.12]{BR07} thanks to a characterization of the valuations set using a Presburger formula.
	Emptiness, universality or finiteness (``is the valuations set finite?'')\ follow immediately. %
}

In~\cite{BDR08}, model checking non-parametric TAs against parametric TCTL (with integer-valued parameters) is considered over both discrete\LongVersion{ time} and dense time.
\ExistsPTCTL{} is defined as the existential fragment (over parameters) of PTCTL\footnote{%
	Of the form $\PExists \param_1, \cdots, \param_n : \fml$ with $\fml$ without quantifiers over parameters.
	Note that, \LongVersion{throughout}\ShortVersion{in} this paper, we use $\PExists$ to distinguish between existential quantification over parameters ($\PExists$) and over paths ($\exists$).
}.
Both the discrete time \cite[Corollary~7.3]{BDR08} and dense time \cite[Theorem~7.5]{BDR08} model checking problems are in \NEXPTIME{5} in the product of the model and the formula, and are in \NEXPTIME{3} for the \ExistsPTCTL{} fragment \cite[Propositions~7.4 and~7.6]{BDR08}.\LongVersion{%
	\footnote{%
		Exact complexity is not given in~\cite{BDR08}, but it is noted that discrete and dense time model checking \ExistsPTCTL{} is at least \textsc{PSpace-Hard}~\cite{ACD90,AL02}, while discrete and dense time model checking PTCTL is at least \textsc{2NExpTime-Hard}~\cite{FR74}.
	}
}

Model-checking subclasses of PTAs against TCTL (beyond reachability) is notably considered in~\cite{ALR18FORMATS}: on the one hand, even for the severely restricted \changed{class} of U-PTAs (\changed{a subclass of PTAs} in which parameters can only be compared to a clock as an upper bound~\cite{BlT09}), and even without invariants, the emptiness is undecidable for nested TCTL\LongVersion{ (the proof uses the ``$\exists \square \CTLAF_{=0}$'' formula)}.
\LongVersion{%
	This undecidability result comes with two flavors:
	\begin{ienumeration}%
		\item over unbounded (possibly integer-valued) parameters, the proof of which uses 5~clocks and 2~parameters; and %
		\item over bounded rational-valued parameters, the proof of which uses 4~clocks and 1~parameters. %
	\end{ienumeration}%
}%
On the other hand, it is then shown in~\cite{ALR18FORMATS} that nest-free TCTL is decidable for L/U-PTAs \changed{(a subclass of PTAs in which parameters are partitioned between lower-bound and upper-bound parameters~\cite{HRSV02})} without invariants.

\subsection{Hyperproperties}
Hyperproperties drew recent attention, and various hyperlogics have been introduced by extending conventional temporal logics (\eg{}~\cite{CFKMRS14,HZJ21,BPS22,GDABB23,BMNN23,BPS20}).

One of the closest works to our timed hyperlogics (without parameters) is HyperMITL~\cite{HZJ21}, a timed extension of HyperLTL~\cite{CFKMRS14}.
In general, the model checking problem is undecidable, even with very restricting timing constraints; it becomes decidable under certain conditions, notably absence of alternation.
For decidable subcases, \LongVersion{the authors}\ShortVersion{they} use a construction based on \emph{self-composition}, which we also use in \cref{section:self-composition}.
However, their construction is primarily for untimed models, while our reduction is for PTAs.

\changed{
Another closely related work is HyperMTL~\cite{BPS20}, another timed extension of HyperLTL.\@
If the time domain is discrete, \ie{} the timestamps are integers, HyperMTL model checking is decidable even with quantifier alternation.
Although their algorithm covers many interesting properties, it is limited to discrete-time and non-parametric settings.
}

Both amplitude and timing parameters are considered in~\cite{BMNN23} for HyperSTL, but the goal is \changed{requirement mining from traces} rather than model checking.
Quantifier alternation is allowed.

Parametric probabilities are considered in a setting~\cite{ABBD20} orthogonal to our timed setting.

\section{Preliminaries}\label{section:preliminaries}

For a set $X$, we denote its powerset by $\powerset{X}$.
\LongVersion{%
  For a finite set $X$, we denote its size by $|X|$.
}%
For sets $X, Y$,
we denote a partial function $f$ from $X$ to $Y$ by $f \colon X \partfun Y$ and denote its domain by $\domain(f) \subseteq X$.
\LongVersion{For $x \in \setZ$ and $N \in \setN$, we let $x \bmod N \in \{0,1,\dots, N - 1\}$ be such that $x \bmod N + y N = x$ for some $y \in \setZ$.}

\LongVersion{
\subsection{Clocks, clock guards, and parameters}
}

We let $\Time$ be the domain of the time, which will be either non-negative reals $\setRnn$\LongVersion{ (continuous-time semantics)} or naturals $\setN$\LongVersion{ (discrete-time semantics)}.
Let $\Clock = \{ \clock_1, \dots, \clock_\ClockCard \}$ be a set of \emph{clocks}, \ie{} variables that evolve at the same rate.
A clock valuation is a function
$\clockval : \Clock \rightarrow \Time$.
We write $\ClocksZero$ for the clock valuation assigning $0$ to all clocks.
Given $d \in \Time$, $\clockval + d$ 
denotes the valuation \st{} $(\clockval + d)(\clock) = \clockval(\clock) + d$, for all $\clock \in \Clock$.
Given $\resets \subseteq \Clock$, we define the \emph{reset} of a valuation~$\clockval$, denoted by $\reset{\clockval}{\resets}$, as follows: $\reset{\clockval}{\resets}(\clock) = 0$ if $\clock \in \resets$, and $\reset{\clockval}{\resets}(\clock)=\clockval(\clock)$ otherwise.

We assume a set~$\Param = \{ \param_1, \dots, \param_\ParamCard \} $ of \emph{parameters}, \ie{} unknown constants.
A parameter \emph{valuation} $\pval$ is  a function
$\pval\colon \Param \rightarrow \setQnn$.%
\footnote{%
	\changed{We choose $\setQnn$ by consistency with most of the PTA literature, but also because, for classical PTAs, choosing~$\Rgeqzero$ leads to undecidability~\cite{Miller00}.}
}
We assume ${\compOp} \in \{<, \leq, =, \geq, >\}$.
A \emph{(clock) guard}~$\guard$ is a constraint over $\Clock \cup \Param$ defined by a conjunction of inequalities of the form $\clock \compOp \paramOrInt$ with $\paramOrInt \in \Param \cup \setN$.
For simplicity, we often use intervals instead of a conjunction of inequalities.
Given~$\guard$, we write~$\clockval\models\pval(\guard)$ if %
the expression obtained by replacing each~$\clock$ with~$\clockval(\clock)$ and each~$\param$ with~$\pval(\param)$ in~$\guard$ evaluates to true.
\changed{For a finite set $X = \{x_1, x_2, \dots, x_N\}$ of size $N \in \setN$, a \emph{linear term} $\lterm$ (\resp{} \emph{non-negative linear term} $\ltermNN$)
is of the form $\sum_{1 \leq i \leq N} \alpha_i x_i + d$, with
	$\alpha_i, d \in \setZ$ (\resp{} $\alpha_i, d \in \setN$).}

\LongVersion{
\subsection{Parametric timed automata}
}

Parametric timed automata (PTAs)~\cite{AHV93} extend timed automata~\cite{AD94} with parameters within guards and invariants in place of integer constants.

\begin{definition}[PTA]%
 \label{definition:PTA}
 A \emph{parametric timed automaton (PTA)} $\A$ is an 8-tuple $\A = (\Actions, \Loc, \LocInit, %
 \Clock, \Param, \invariant, \Edges, \Label)$, where:
 \begin{itemize}
  \item $\Actions$ is a finite set of atomic propositions;
  \item $\Loc$ is a finite set of locations;
  \item $\LocInit \subseteq \Loc$ is the set of initial locations,
  \item $\Clock$ is a finite set of clocks;
  \item $\Param$ is a finite set of parameters;
  \item $\invariant$ is the invariant, assigning to every $\loc\in \Loc$ a clock guard $\invariant(\loc)$;
  \item $\Edges$ is a finite set of edges  $\edge = (\loc, \guard, \resets, \loc')$,
		where~$\loc, \loc' \in \Loc$ are the source and target locations, $\resets\subseteq \Clock$ is a set of clocks to be reset, and $\guard$ is the transition guard;
  \item $\Label\colon \Loc \to \powerset{\Actions}$ is the labeling function \changed{assigning the atomic propositions satisfied at each location}.
 \end{itemize}
\end{definition}
Given a parameter valuation~$\pval$, 
we denote by $\valuate{\A}{\pval}$ the non-parametric structure where all occurrences of a parameter~$\param_i$ have been replaced by~$\pval(\param_i)$.
We refer as a \emph{timed automaton (TA)} to any structure $\valuate{\A}{\pval}$, by assuming a rescaling of the constants: by multiplying all constants in $\valuate{\A}{\pval}$ by the least common multiple of their denominators, we obtain an equivalent (integer-valued) TA%
, as defined in~\cite{AD94}.

\begin{example}
	The PTA in \cref{figure:system_running_example} contains 1~clock $\styleclock{\clock}$ and one parameter~$\styleparam{\param_1}$.
	The invariant of $\loc_0$ is ``$\styleclock{\clock} \leq \styleparam{\param_1}$'' and the transition to~$\loc_1$ is guarded by ``$\changed{\styleparam{\param_1} - 1 < }\styleclock{\clock} < \styleparam{\param_1}$'', and resets~$\styleclock{\clock}$.
	Atomic propositions $\styleact{H}$ and $\styleact{L}$ are associated with $\loc_0$ and~$\loc_1$, respectively.
 This PTA models a clock generator with drift: the digital signal switches between high ($\styleact{H}$) and low ($\styleact{L}$) states in a near periodic manner but with some timing deviation, depending on the value of parameter $\styleparam{\param_1}$.
\end{example}
\begin{figure}[tbp]
 \centering
 \begin{tikzpicture}[pta,scale=0.65,every node/.style={transform shape,initial text=}]
  \node[initial,location] (a) at (0,0)  [align=center]{$\styleloc{\loc_0}$\\ $\{\styleact{H}\}$};
  \node[location] (b) at (7.5,0) [align=center]{$\styleloc{\loc_1}$\\ $\{\styleact{L}\}$};

  \node[invariant,yshift=-.6em] (a_invariant) at (0,1.0) {$\styleclock{c} \leq \styleparam{\param_1}$};
  \node[invariant,yshift=-.6em] (b_invariant) at (7.5,1.0) {$\styleclock{c} \leq 3$};

  \path[->]
  (a) edge [bend left=5] node[above] {$\changed{\styleparam{\param_1} - 1 < }\styleclock{c} < \styleparam{\param_1}/\styleclock{c} \Coloneq 0$} (b)
  (b) edge [bend left=5] node[below] {$\changed{2 < }\styleclock{c} < 3$} (a)
  ;
 \end{tikzpicture}
 \caption{Drifted clock generator example: PTA~$\A$}%
 \label{figure:system_running_example}
\end{figure}
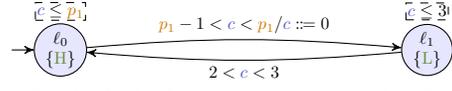

Let us now recall the concrete semantics of TAs.

\begin{definition}[Semantics of a TA]
	For a PTA $\A = (\Actions, \Loc, \LocInit, \Clock, \Param, \invariant, \Edges, \Label)$
	and a parameter valuation~\(\pval\),
	the semantics of the TA $\valuate{\A}{\pval}$ is given by the timed transition system (TTS) $\TTS = (\States, \SInit, \flecheRel)$, with
	\begin{itemize}
		\item $\States = \big\{ (\loc, \clockval) \in \Loc \times \Time^\ClockCard \mid \clockval \models \invariant(\loc) \big\}$,
		\item $\SInit = \{(\locinit, \ClocksZero) \mid \locinit \in \LocInit\} $,
		\item  $\flecheRel$ consists of the discrete and (continuous) delay transition relations:
		\begin{itemize}
			\item[discrete transitions] $(\loc, \clockval) \longueflecheRel{\edge} (\loc', \clockval')$, %
				if
				$(\loc, \clockval) , (\loc', \clockval') \in \States$, and
				there exists $\edge = (\loc, \guard, \resets, \loc') \in \Edges$, such that $\clockval' = \reset{\clockval}{\resets}$, and $\clockval \models \guard$.
			\item[delay transitions] $(\loc, \clockval) \longueflecheRel{d} (\loc, \clockval + d)$, with $d \in \Time$, if $\forall d' \in [0, d], (\loc, \clockval+d') \in \States$.
		\end{itemize}
	\end{itemize}
\end{definition}

Moreover, we write $(\loc, \clockval)\longuefleche{(d, \edge)} (\loc', \clockval')$ for a combination of delay and discrete transitions if
$\exists  \clockval'' :  (\loc, \clockval) \longueflecheRel{d} (\loc, \clockval'') \longueflecheRel{\edge} (\loc', \clockval')$.
We let $\TTSLabel((\loc, \clockval)) = \Label(\loc)$.

Given a TA~$\valuate{\A}{\pval}$ with concrete semantics $(\States, \SInit, \flecheRel)$, we refer to the states~$\States$ as the \emph{concrete states} of~$\valuate{\A}{\pval}$.
For $\state = (\loc, \clockval) \in \States$ and $d \in \Time$, we let $\state + d = (\loc, \clockval + d)$.
A \emph{path} of~$\valuate{\A}{\pval}$ from a concrete state $\state$ is an alternating \emph{infinite} sequence of concrete states of $\valuate{\A}{\pval}$ and pairs of edges and delays starting from $\state$ of the form
$\state_0 (= \state), (d_0, \edge_0), \state_1, \ldots$
with
$\sum_{i = 0}^{\infty} d_i = +\infty$, 
for each $i = 0, 1, \dots$, $d_i \in \Time$, $\edge_i \in \Edges$, and
$\state_i \longuefleche{(d_i, \edge_i)} \state_{i+1}$.
We denote the set of paths of $\valuate{\A}{\pval}$ from $\state$ by $\Traces(\valuate{\A}{\pval}, \state)$.
We let $\Traces(\valuate{\A}{\pval}) = \bigcup_{\state \in \States} \Traces(\valuate{\A}{\pval}, \state)$.
For a path $\state_0, (d_0, \edge_0), \state_1, \ldots$ of $\valuate{\A}{\pval}$,
a \emph{position} is a concrete state $\state$ satisfying $\state = \state_i + d$ for some $i \in \setN$ and $d \leq d_i$.
For a path $\trace$, we denote its initial position $\state_0$ by $\Init(\trace)$.
For a position $\state = \state_i + d$ of a path $\trace = \state_0, (d_0, \edge_0), \state_1, \ldots$,
the \emph{duration} $\duration_{\trace}(\state)$ is $\duration_{\trace}(\state) = d + \sum_{j = 0}^{i-1} d_j$.
If the path is clear from the context, we just write $\duration(\state)$.
For positions $\state = \state_i + d$ and $\state' = \state_j + d'$ of a path $\state_0, (d_0, \edge_0), \state_1, \ldots$,
we let $\state < \state'$ if we have $i < j$ or $\duration(\state) < \duration(\state')$.
We let $\state \leq \state'$ if we have $\state < \state'$ or $\state = \state'$.
For paths $\trace, \trace'$, we write $\trace \suffixEqOp \trace'$ if $\trace'$ is a suffix of $\trace$, \ie{} for $\trace = \state_0, (d_0, \edge_0), \state_1, \ldots$, $\trace'$ is such that $\state_i + d, (d_i - d, \edge_i), \state_{i+1}, \ldots$ for some $i \in \setN$ and $d \in [0, d_i]$.
We let $\removedEdges{\trace'}{\trace}$ be such $i$.
We let $\trace \suffixOp \trace'$ if we have $\trace \suffixEqOp \trace'$ and $\trace \neq \trace'$.
For paths $\trace, \trace'$ satisfying $\trace \suffixEqOp \trace'$, we let $\duration(\trace - \trace')$ be the duration of the initial position of $\trace'$ in $\trace$.

For PTAs $\A^1$ and $\A^2$, we define both parallel composition $\A^1 \composeOp \A^2$ and synchronized product $\A^1 \productOp \A^2$.
Intuitively, the parallel composition is to juxtapose two PTAs without synchronization, whereas
the synchronized product is to compose two PTAs synchronizing the edges with the propositions.
The parallel composition will be used when taking self-composition of the systems to handle multiple paths simultaneously, while
the synchronized product will be used when composing the systems with observers encoding extended predicates\LongVersion{ (\cref{ss:observers})}.

For PTAs $\A^1 = (\Actions^1, \Loc^1, \LocInit^1, \Clock^1, \Param^1, \invariant^1, \Edges^1, \Label^1)$
and
$\A^2 = (\Actions^2, \Loc^2, \LocInit^2, \Clock^2, \Param^2, \invariant^2, \Edges^2, \Label^2)$,
their \emph{parallel composition} is $\A^1 \composeOp \A^2 = \big(\Actions^1 \disjointUnion \Actions^2, \Loc^1 \times \Loc^2, \LocInit^1 \times \LocInit^2, \Clock^1 \disjointUnion \Clock^2, \Param^1 \cup \Param^2, \invariant, \Edges, \Label \big)$,
with
$\disjointUnion$ denoting disjoint union,
$\invariant\big((\loc^1, \loc^2) \big) = \invariant^1(\loc^1) \land \invariant^2(\loc^2)$,
$\Edges = \big\{ \big( (\loc^1, \loc^2), \guard, \resets, ({\loc^1}', \loc^2) \big) \mid (\loc^1, \guard, \resets, {\loc^1}') \in \Edges^1, \loc^2 \in \Loc^2 \big\} \cup  \big \{ \big((\loc^1, \loc^2), \guard, \resets, (\loc^2, {\loc^2}') \big) \mid (\loc^2, \guard, \resets, {\loc^2}') \in \Edges^2, \loc^1 \in \Loc^1 \big\} \cup \big\{ \big( (\loc^1, \loc^2), \guard^1 \land \guard^2, \resets^1\cup\resets^2, ({\loc^1}', {\loc^2}') \big) \mid (\loc^1, \guard^1, \resets^1, {\loc^1}') \in \Edges^1, (\loc^2, \guard^2, \resets^2, {\loc^2}') \in \Edges^2 \big\}$, and
$\Label \big((\loc^1, \loc^2)\big) = \Label^1(\loc^1) \disjointUnion \Label^2(\loc^2)$.
For TAs $\valuate{\A_1}{\pval_1}$ and $\valuate{\A_2}{\pval_2}$ satisfying $\pval_1(\param) = \pval_2(\param)$ for any $\param \in \Param_1 \cap \Param_2$, and
paths $\trace_1$ and $\trace_2$ of $\valuate{\A_1}{\pval_1}$ and $\valuate{\A_2}{\pval_2}$, respectively,
we let $\trace_1 \composeOp \trace_2$ be the path of $\valuate{\A_1 \composeOp \A_2}{(\pval_1 \cup \pval_2)}$ obtained by parallel composition of $\trace_1$ and $\trace_2$, where $\pval_1 \cup \pval_2$ is the parameter valuation such that $(\pval_1 \cup \pval_2)(\param) = \pval_1(\param)$ if $\param \in \pval_1$ and otherwise $(\pval_1 \cup \pval_2)(\param) = \pval_2(\param)$.
Conversely, for a path $\trace$ of $\valuate{\A_1 \composeOp \A_2}{(\pval_1 \cup \pval_2)}$ and $i \in \{1,2\}$,
we let $\project{\trace}{i}$ be the path of $\valuate{\A_i}{\pval_i}$ obtained by removing the locations, clock valuations, and edges from $\A_{3 - i}$.

For PTAs $\A^1 = (\Actions^1, \Loc^1, \LocInit^1, \Clock^1, \Param^1, \invariant^1, \Edges^1, \Label^1)$
and
$\A^2 = (\Actions^2, \Loc^2, \LocInit^2, \Clock^2, \Param^2, \invariant^2, \Edges^2, \Label^2)$,
their \emph{synchronized product} is $\A^1 \productOp \A^2 = \big(\Actions^1 \cup \Actions^2, \Loc^1 \times \Loc^2, \LocInit, \Clock^1 \disjointUnion \Clock^2, \Param^1 \cup \Param^2, \invariant, \Edges, \Label \big)$,
with
$\LocInit = \{(\locinit^1, \locinit^2) \in \LocInit^1 \times \LocInit^2 \mid \Label^1({\locinit^1}) \cap \Actions^2 = \Label^2({\locinit^2}) \cap \Actions^1 \}$,
$\invariant\big((\loc^1, \loc^2) \big) = \invariant^1(\loc^1) \land \invariant^2(\loc^2)$,
$\Edges = \big\{ \big( (\loc^1, \loc^2), \guard, \resets, ({\loc^1}', \loc^2) \big) \mid (\loc^1, \guard, \resets, {\loc^1}') \in \Edges^1, \loc^2 \in \Loc^2, \Label^1({\loc^1}') \cap \Actions^2 = \Label^2(\loc^2) \cap \Actions^1\} \cup  \big \{ \big((\loc^1, \loc^2), \guard, \resets, (\loc^1, {\loc^2}') \big) \mid (\loc^2, \guard, \resets, {\loc^2}') \in \Edges^2, \loc^1 \in \Loc^1, \Label^1(\loc^1) \cap \Actions^2 = \Label^2({\loc^2}') \cap \Actions^1 \big\} \cup  \big \{ \big((\loc^1, \loc^2), \guard^1 \land \guard^2, \resets^1 \disjointUnion \resets^2, ({\loc^1}', {\loc^2}') \big) \mid (\loc^1, \guard^1, \resets^1, {\loc^1}') \in \Edges^1, (\loc^2, \guard^2, \resets^2, {\loc^2}') \in \Edges^2, \Label^1({\loc^1}') \cap \Actions^2 = \Label^2({\loc^2}') \cap \Actions^1 \big\}$, and
$\Label \big((\loc^1, \loc^2)\big) = \Label^1(\loc^1) \cup \Label^2(\loc^2)$.
\changed{For a finite set $I = \{1,2,\dots,n\}$ of indices, we let $\bigtimes_{i \in I} \A^i = \A^1 \times \A^2 \times \dots \times \A^n$, where each $\A^i$ is a PTA.}

%

%
%
\section{Hyper parametric timed CTL}\label{section:hyper-ptcl}

Here, we introduce \emph{hyper parametric timed computation tree logic (\HyperPTCTL{})}.
\HyperPTCTL{} is a generalization of PTCTL~\cite{BDR08} with quantifiers over paths to represent hyperproperties.

\subsection{Syntax of \ExtOrNotHyperPTCTL{}}
\begin{definition}
 [Syntax of \HyperPTCTL{}]%
 \label{definition:syntax_hyperptctl}
 For atomic propositions $\Actions$ and parameters $\Param$,
 the syntax of \HyperPTCTL{} formulas of the temporal level~$\fml$ and the top level~$\fullFml$ are defined as follows, where
 $\action \in \Actions$,
 \changed{$\PathVar$ is the set of path variables,}
 $\vrun, \vrun_1, \vrun_2, \dots, \vrun_n \in \PathVar$,
 $\paramOrInt \in \Param \cup \setN$,
 $\param \in \Param$,
 ${\compOp} \in \{<, \leq, =, \geq, >\}$, and 
 $\ltermNN$ is a \changed{non-negative} linear term over $\Param$: %
{%
\begin{align*}
 \fml \Coloneq & \top \mid \action_{\vrun} \mid \neg \fml \mid \fml \lor \fml \mid \HEUntilFml[\compOp \paramOrInt]{\vrun_1,\vrun_2,\dots,\vrun_n}{\fml}{\fml}
 \\
 & \mid \HAUntilFml[\compOp \paramOrInt]{\vrun_1,\vrun_2,\dots,\vrun_n}{\fml}{\fml}\\
 \fullFml \Coloneq & \fml \mid \param \compOp \ltermNN \mid \neg \fullFml \mid \fullFml \lor \fullFml \mid \PExists \param\, \fullFml
 \end{align*}
}
\end{definition}
As syntax sugar, we utilize the following formulas:
{\ShortVersion{\small}\LongVersion{\scriptsize}
\begin{gather*}
 \bot \equiv \neg \top \quad \fml_1 \land \fml_2 \equiv \neg (\neg \fml_1 \lor \neg \fml_2) \quad \fml_1 \implies \fml_2 \equiv \neg\fml_1 \lor \fml_2\\
 \fml_1 = \fml_2 \equiv (\fml_1 \land \fml_2) \lor ((\neg \fml_1) \land (\neg \fml_2)) \quad
 \fml_1 \neq \fml_2 \equiv \neg (\fml_1 = \fml_2) \\
 \HEReleaseFml[\compOp \paramOrInt]{\vrun_1,\vrun_2,\dots,\vrun_n}{\fml_1}{\fml_2} \equiv \neg \HAUntilFml[\compOp \paramOrInt]{\vrun_1,\vrun_2,\dots,\vrun_n}{\neg \fml_1}{\neg \fml_2} \\
 \HAReleaseFml[\compOp \paramOrInt]{\vrun_1,\vrun_2,\dots,\vrun_n}{\fml_1}{\fml_2} \equiv \neg \HEUntilFml[\compOp \paramOrInt]{\vrun_1,\vrun_2,\dots,\vrun_n}{\neg \fml_1}{\neg \fml_2} \\
 \LongVersion{\HEWeakUntilFml[\compOp \paramOrInt]{\vrun_1,\vrun_2,\dots,\vrun_n}{\fml_1}{\fml_2} \equiv \HEReleaseFml[\compOp \paramOrInt]{\vrun_1,\vrun_2,\dots,\vrun_n}{\fml_2}{(\fml_1 \lor \fml_2)} \\}
 \LongVersion{\HAWeakUntilFml[\compOp \paramOrInt]{\vrun_1,\vrun_2,\dots,\vrun_n}{\fml_1}{\fml_2} \equiv \HAReleaseFml[\compOp \paramOrInt]{\vrun_1,\vrun_2,\dots,\vrun_n}{\fml_2}{(\fml_1 \lor \fml_2)} \\}
 \HEDiaFml[\compOp \paramOrInt]{\vrun_1,\vrun_2,\dots,\vrun_n}{\fml} \equiv \HEUntilFml[\compOp \paramOrInt]{\vrun_1,\vrun_2,\dots,\vrun_n}{\top}{\fml} \\
 \HADiaFml[\compOp \paramOrInt]{\vrun_1,\vrun_2,\dots,\vrun_n}{\fml} \equiv \HAUntilFml[\compOp \paramOrInt]{\vrun_1,\vrun_2,\dots,\vrun_n}{\top}{\fml} \\
 \HEBoxFml[\compOp \paramOrInt]{\vrun_1,\vrun_2,\dots,\vrun_n}{\fml} \equiv \HEReleaseFml[\compOp \paramOrInt]{\vrun_1,\vrun_2,\dots,\vrun_n}{\bot}{\fml} \\
 \HABoxFml[\compOp \paramOrInt]{\vrun_1,\vrun_2,\dots,\vrun_n}{\fml} \equiv \HAReleaseFml[\compOp \paramOrInt]{\vrun_1,\vrun_2,\dots,\vrun_n}{\bot}{\fml}
\end{gather*}
}

\emph{\ExistsHyperPTCTL{}} is a subclass of \HyperPTCTL{} such that the (top-level) formulas are of the form $\PExists \param_1 \PExists \param_2 \dots \PExists \param_n\, \fullFml$, where $\fullFml$ has no quantifiers over parameters.
\emph{PTCTL}~\cite{BDR08} is a subclass of \HyperPTCTL{} with only one path variable~$\vrun$.\footnote{Although the original definition in~\cite{BDR08} does not use path variables, one can easily check that our definition coincides with the one in~\cite{BDR08}.}
\emph{\ExistsPTCTL{}}~\cite{BDR08} is defined analogously.

We extend \HyperPTCTL{} with counters and clocks to constrain the number of occurrences of atomic propositions and to compare the time difference between the occurrences of two atomic propositions, respectively.
Intuitively, for each atomic proposition $\action$ and path~$\vrun$,
$\LAST(\action_{\vrun})$ indicates the \changed{time elapsed since the last switching of~$\action$ from false to true} in~$\vrun$, and
$\COUNT(\action_{\vrun})$ indicates the total number of \changed{switches of~$\action$ from false to true} in~$\vrun$.
To ensure decidability, we only allow specific forms of constraints.
We denote the extended logic as \emph{\ExtHyperPTCTL{}}.

\begin{definition}
 [Syntax of \ExtHyperPTCTL{}]
 We extend the syntax of \HyperPTCTL{} formulas of the temporal level as follows:
 $\action \in \Actions$,
 \changed{$\PathVar$ is the set of path variables,}
 $\vrun \in \PathVar$,
 ${\compOp} \in \{{<}, {\leq}, {=}, {\geq}, {>}\}$,
 $d,N \in \setN$,
 $\lterm$ is a linear term over $\Param$,
 $\CountActions = \{\COUNT(\action_{\vrun}) \mid \action \in \Actions, \vrun \in \PathVar\}$,
 $\ltermCntNN$ is a \emph{non-negative} linear term over $\CountActions$ (\ie{} $\ltermCntNN = \sum_{\action \in \Actions, \vrun \in \PathVar} \alpha_{\action, \vrun} \COUNT(\action_{\vrun})$ for some $\alpha_{\action, \vrun} \in \setN$)\footnote{The restriction of $\ltermCntNN$ is to construct finite observers in \cref{ss:observers}. Our implementation accepts any linear term, which is encoded by variables in \imitator{}~\cite{Andre21}.}, and
 $\ltermCnt$ is a linear term over $\CountActions$ (\ie{} $\ltermCnt = \sum_{\action \in \Actions, \vrun \in \PathVar} \alpha_{\action, \vrun} \COUNT(\action_{\vrun})$ for some $\alpha_{\action, \vrun} \in \setZ$).
 \begin{align*}
  \fml \Coloneq & \top \mid \action_{\vrun} \mid \LAST(\action_{\vrun}) - \LAST(\action_{\vrun}) \compOp \lterm \\
  & \mid \ltermCntNN \compOp d \mid (\ltermCnt \bmod N) \compOp d \mid \cdots
 \end{align*}
\end{definition}

\changed{We call $\LAST(\action_{\vrun}) - \LAST(\action_{\vrun}) \compOp \lterm$, $\ltermCntNN \compOp d$, and $(\ltermCnt \bmod N) \compOp d$ as \LASTExpr{}, \COUNTNNExpr{}, and \COUNTModExpr{}, respectively.}
We let \emph{\ExtPTCTL{}} be the subclass of \ExtHyperPTCTL{} with only one path variable~$\vrun$.
\begin{example}
 [Drift of clock]%
 \label{example:skewed_clock}
 Let $\styleact{H}$ be the proposition showing the ``high'' value of a digital clock.
 For
 $\styleparam{\param} \in \Param$,
 $\mathrm{AtMostOneDiff} \equiv (\COUNT(\styleact{H}_{\vrun_1}) - \COUNT(\styleact{H}_{\vrun_2})) \bmod 4 \in \{0, 1, 3\}$ \changed{denotes that the deviation of $\COUNT(\styleact{H}_{\vrun_1})$ and $\COUNT(\styleact{H}_{\vrun_2})$ is at most by one, if we keep having $\mathrm{AtMostOneDiff}$ in the past%
},
 $\mathrm{SameCount} \equiv (\COUNT(\styleact{H}_{\vrun_1}) - \COUNT(\styleact{H}_{\vrun_2})) \bmod 4 = 0$ \changed{denotes that the number of times the clock became high is identical (mod~4) over two paths;} and
 $\mathrm{LargeDeviation} \equiv \LAST(\styleact{H}_{\vrun_1}) - \LAST(\styleact{H}_{\vrun_2}) \not\in [-\styleparam{\param}, \styleparam{\param}]$ \changed{denotes that the last time the clock became high differs by at least~$\styleparam{\param}$ time units over two paths, \ie{} consists in a ``large deviation''}.
 \changed{Then}
 the following \ExtHyperPTCTL{} formula shows the drift of near periodic clocks of duration at least $\styleparam{\param}$ time units, globally assuming $\mathrm{AtMostOneDiff}$:
 $\HEUntilFml{\vrun_1,\vrun_2}{(\mathrm{AtMostOneDiff})}{(\mathrm{SameCount} \land \mathrm{LargeDeviation})}$.
\end{example}

\LongVersion{Notice that using a larger modulus, one can generalize $\mathrm{AtMostOneDiff}$ to keep track of more than one deviation, whereas it increases the cost of the reduction in \cref{ss:observers}.}

\begin{example}
 [\changed{Execution-time} opacity]\label{example:opacity}
 Let
 $\styleact{Private}$ be the proposition showing the private locations \changed{of a (P)TA} and
 $\styleact{Goal}$ be the proposition showing the goal locations.
 The following \ExtHyperPTCTL{} formula shows a formulation of opacity focusing on the execution time~\cite{ALLMS23},  %
\ie{}
 there are executions of duration \styleparam{\param} with and without visiting any private locations:
 $\HEUntilFml[= \styleparam{\param}]{\stylepathvar{\vrun_1}, \stylepathvar{\vrun_2}}{(\neg \styleact{Goal}_{\stylepathvar{\vrun_1}} \land \neg \styleact{Goal}_{\stylepathvar{\vrun_2}})}{(\styleact{Goal}_{\stylepathvar{\vrun_1}} \land \styleact{Goal}_{\stylepathvar{\vrun_2}} \land}$ $\COUNT(\styleact{Private}_{\stylepathvar{\vrun_1}}) = 0 \land \COUNT(\styleact{Private}_{\stylepathvar{\vrun_2}}) > 0 )$.
 \changed{That is, the goal is not reached until, after exactly $ \styleparam{\param}$ time units (which encodes the unknown execution time), the goal is reached for both paths, and one of them did not visit any private location (``$\COUNT(\styleact{Private}_{\stylepathvar{\vrun_1}}) = 0$'') while the other one did.}
 \changed{Notice that this differs from another style of opacity for TAs in~\cite{Cassez09}.}
\end{example}
\begin{example}
 [\changed{Side-channel timing attack~\cite{BPS20}}]
 \label{example:side_channel_attack}
 \changed{Let $\styleact{Inv}$ and $\styleact{Idle}$ be the propositions denoting the invocation of a process and the idle state.
 For $i \in \{1,2\}$, we let
 $\mathrm{SyncInv} \equiv \styleact{Inv}_{\vrun_1} = \styleact{Inv}_{\vrun_2}$,
 $\mathrm{ImmediateExec}_i \equiv \styleact{Inv}_{\vrun_i} \implies \neg \styleact{Idle}_{\vrun_i}$,
 $\mathrm{ExecBound}_i \equiv \styleact{Idle}_{\vrun_i} \implies (\LAST(\styleact{Idle}_{\vrun_i}) - \LAST(\styleact{Inv}_{\vrun_i}) < \styleparam{\param_1})$, and
 $\mathrm{NearFinish} \equiv (\styleact{Idle}_{\vrun_1} = \styleact{Idle}_{\vrun_2}) \implies (\LAST(\styleact{Idle}_{\vrun_1}) - \LAST(\styleact{Idle}_{\vrun_2}) \in (-\styleparam{\param_2}, \styleparam{\param_2}))$.
 The following \ExtHyperPTCTL{} formula shows that the execution time of each process must be similar, which is necessary to prevent side-channel timing attacks:
 $\HAReleaseFml{\vrun_1,\vrun_2}{(\neg \mathrm{SyncInv})}{(\mathrm{ImmediateExec}_1} \land \mathrm{ImmediateExec}_2 \land \mathrm{ExecBound}_1 \land \mathrm{ExecBound}_2 \land \mathrm{NearFinish})$.
 More precisely, for any paths corresponding to two different sequences of process executions, while each pair of processes has been invoked simultaneously, their execution must be within $\styleparam{\param_1}$, and the execution time of each pair of processes must not differ more than $\styleparam{\param_2}$ time units.}
\end{example}
\begin{example}
 [Unfairness of schedulers]
 \label{example:unfairness}
 Let $\styleact{Sub}^i$ and $\styleact{Run}^i$ be the proposition showing the submission and execution of job $i \in \{1,2\}$.
 For
 $\mathrm{SyncSub} \equiv \styleact{Sub}^1_{\vrun_1} = \styleact{Sub}^2_{\vrun_2}$,
 $\mathrm{SameCount} \equiv (\COUNT(\styleact{Run}^1_{\vrun_1}) - \COUNT(\styleact{Run}^2_{\vrun_2})) \bmod 4 = 0$, and
 $\mathrm{LargeDeviation} \equiv \LAST(\styleact{Run}^1_{\vrun_1}) - \LAST(\styleact{Run}^2_{\vrun_2}) \not\in (-5, 5)$,
 the following \ExtHyperPTCTL{} formula shows unfair scheduling of jobs 1 and 2:
 $\HEUntilFml{\vrun_1,\vrun_2}{(\mathrm{SyncSub})}{(\mathrm{SameCount} \land \mathrm{LargeDeviation})}$.
 That is, given two paths corresponding to two different jobs, if they submit the job at the same time and eventually run, but with a time difference of $\geq 5$ time units, then the scheduler is unfair.
\end{example}
\begin{example}
 [Robust observational non-determinism]
 \label{example:observational-non-determinism}
 Let $\big\{\styleact{In}^i \mid i \in \{1,2,\dots,m\}\big\}$ and $\big\{\styleact{Out}^i \mid i \in \{1,2,\dots,n\} \big\}$ be the set of input and output propositions, respectively.
 $\mathrm{SyncIn} \equiv \bigwedge_{i \in \{1,2,\dots,m\}} (\styleact{In}^i_{\vrun_1} = \styleact{In}^i_{\vrun_2})$
 \changed{denotes that the inputs in two runs $\vrun_1$ and $\vrun_2$ are synchronized},
 $\mathrm{AtMostOneDiff} \allowbreak\equiv\allowbreak \bigwedge_{i \in \{1,2,\dots,n\}} \big(\COUNT(\styleact{Out}^i_{\vrun_1}) - \COUNT(\styleact{Out}^i_{\vrun_2}) \big) \bmod 4 \allowbreak\in\allowbreak \{0,1,3 \}$
 \changed{denotes that the deviation of the $\COUNT(\styleact{Out}^i_{\vrun_1})$ and $\COUNT(\styleact{Out}^i_{\vrun_2})$ is at most one, if we keep having $\mathrm{AtMostOneDiff}$ in the past}, and
 $\mathrm{LargeDeviation} \equiv \bigvee_{i \in \{1,2,\dots,n\}} \bigl((\COUNT(\styleact{Out}^i_{\vrun_1}) - \COUNT(\styleact{Out}^i_{\vrun_2})) \bmod 4 = 0 \allowbreak\land\allowbreak \LAST(\styleact{Out}^i_{\vrun_1}) - \LAST(\styleact{Out}^i_{\vrun_2}) \not\in [-\styleparam{\param}, \styleparam{\param}]\bigr)$ \changed{denotes that there is an output proposition that the number of times of the proposition became true is identical ($\bmod~4$) over two paths but the timing differs at least $\styleparam{\param}$ time units over two paths}.
 \changed{The} following \ExtHyperPTCTL{} formula shows robust observational non-determinism assuming $\mathrm{AtMostOneDiff}$, \ie{}
 even if the inputs are given at the same timing, the output timing may deviate more than \styleparam{\param}:
 $\HEUntilFml{\vrun_1,\vrun_2}{(\mathrm{SyncIn} \land \mathrm{AtMostOneDiff})}{(\mathrm{LargeDeviation})}$.
\end{example}
\subsection{Semantics of \ExtOrNotHyperPTCTL{}}

Before defining the semantics of \HyperPTCTL{}, we formalize the assignments of paths.
In addition to the partial function assigning the paths,
we use a total preorder to fix the order of the discrete transitions at the same time-point.

\begin{definition}
 [Path assignments]%
 \label{definition:path_assignment}
 For path variables $\PathVar$ and a TA $\valuate{\A}{\pval}$,
 a \emph{path assignment} $(\runs, \pathOrder)$ is a pair of 
 a partial function $\runs \colon \PathVar \partfun \Traces(\valuate{\A}{\pval})$ from path variables $\PathVar$ to paths $\Traces(\valuate{\A}{\pval})$ of $\valuate{\A}{\pval}$
 and a total preorder $\pathOrder$ on $\domain(\runs) \times \setN$ such that
 for any $\vrun, \vrun' \in \domain(\runs)$ and $i, j \in \setN$,
 $i < j$ implies $(\vrun, i) \pathOrder (\vrun, j)$ and $(\vrun, j) \npathOrder (\vrun, i)$,
 $(\vrun, i) \pathOrder (\vrun', j)$ implies $\sum_{k = 0}^{i} d^{\vrun}_k \leq \sum_{k = 0}^{j} d^{\vrun'}_k$, and
 $\sum_{k = 0}^{i} d^{\vrun}_k < \sum_{k = 0}^{j} d^{\vrun'}_k$ implies $(\vrun, i) \pathOrder (\vrun', j)$,
 where $d^{\vrun}_k$ and $d^{\vrun'}_k$ are the $k$-th delay in $\runs(\vrun)$ and $\runs(\vrun')$, respectively.
\end{definition}

 We let $(\emptyruns, \pathOrder[\emptyruns])$ be the empty path assignment,
 \ie{} the path assignment satisfying $\domain(\emptyruns) = \emptyset$.
 For path assignments $(\runs, \pathOrder)$ and $(\runs', \pathOrder[\runs'])$, 
 $(\runs', \pathOrder[\runs'])$ is an \emph{extension} of $(\runs, \pathOrder)$ if we have 
 $\domain(\runs) \subseteq \domain(\runs')$ and
 for any $\vrun, \vrun' \in \domain(\runs)$ and $i,j \in \setN$, we have
 $(\vrun, i) \pathOrder (\vrun', j) \iff (\vrun, i) \pathOrder[\runs'] (\vrun', j)$ and
 $\runs(\vrun) = \runs'(\vrun)$.
 For path assignments~$(\runs, \pathOrder)$ and $(\runs', \pathOrder[\runs'])$, we let $(\runs, \pathOrder) \suffixEqOp (\runs', \pathOrder[\runs'])$ (\resp{} $(\runs, \pathOrder) \suffixOp (\runs', \pathOrder[\runs'])$) if $\domain(\runs) = \domain(\runs')$ and there is $d \in \setRnn$ such that
 for any $\vrun \in \domain(\runs)$, we have $\runs(\vrun) \suffixEqOp \runs'(\vrun)$ (\resp{} $\runs(\vrun) \suffixOp \runs'(\vrun)$) and $\duration(\runs(\vrun) - \runs'(\vrun)) = d$, and
 for any $\vrun, \vrun' \in \domain(\runs)$ and $i, j \in \setN$,
 we have
 $(\vrun, i) \pathOrder[\runs'] (\vrun', j) \iff (\vrun, i + \removedEdges{\runs(\vrun)}{\runs'(\vrun)}) \pathOrder (\vrun', j + \removedEdges{\runs(\vrun')}{\runs'(\vrun')})$,
 and if $\removedEdges{\runs(\vrun)}{\runs'(\vrun)} \geq 1$ holds,
 $(\vrun, \removedEdges{\runs(\vrun)}{\runs'(\vrun)} - 1) \pathOrder (\vrun', \removedEdges{\runs(\vrun')}{\runs'(\vrun')})$ and
 $(\vrun', \removedEdges{\runs(\vrun')}{\runs'(\vrun')}) \npathOrder (\vrun, \removedEdges{\runs(\vrun)}{\runs'(\vrun)} - 1)$ also hold.
 We let $\duration(\runs - \runs') = d$ for such $\runs$, $\runs'$, and $d$.

\begin{definition}[Semantics of \HyperPTCTL{}]
 \label{definition:semantics_HyperPTCTL}
Let $\A$ be a PTA over parameters~$\Param_1$;
given a \HyperPTCTL{} formula over~parameters~$\Param_2$,
let $\Param = \Param_1 \cup \Param_2$.
For a parameter valuation $\pval \in \PVal$
a path assignment $(\runs, \pathOrder)$, and 
a concrete state $\state$ of $\valuate{\A}{\pval}$,
the satisfaction relation of the temporal level \HyperPTCTL{} formulas is defined as follows:
\begin{itemize}
  \item $(\runs, \pathOrder), \state \models_{\pval,\A} \action_{\vrun}$ if $\vrun \in \domain(\runs)$ and $\action \in \Label(\Init(\runs(\vrun)))$;
  \item $(\runs, \pathOrder), \state \models_{\pval,\A} \neg \fml$ if we have $(\runs, \pathOrder), \state \not\models_{\pval,\A} \fml$;
  \item $(\runs, \pathOrder), \state \models_{\pval,\A} \fml_1 \lor \fml_2$ if $(\runs, \pathOrder), \state \models_{\pval,\A} \fml_1$ or $(\runs, \pathOrder), \state \models_{\pval,\A} \fml_2$ holds;
  \item $(\runs, \pathOrder), \state \models_{\pval,\A} \HEUntilFml[\compOp \paramOrInt]{\vrun_1,\vrun_2,\dots,\vrun_n}{\fml_1}{\fml_2}$ if for some extension $(\runs^{1}, \pathOrder[\runs^{1}])$ of $(\runs, \pathOrder)$ satisfying $\domain(\runs^{1}) = \domain(\runs) \disjointUnion \{\vrun_1,\vrun_2,\dots,\vrun_n\}$ and $\runs^{1}(\vrun_i) \in \Traces(\valuate{\A}{\pval}, \state)$ for each $i \in \{1,2,\dots,n\}$, there is $(\runs^{2}, \pathOrder[\runs^{2}])$ satisfying $(\runs^{1}, \pathOrder[\runs^{1}]) \suffixEqOp (\runs^{2}, \pathOrder[\runs^{2}])$, $\duration(\runs^{1} - \runs^{2}) \compOp \valuate{\paramOrInt}{\pval}$, $(\runs^{2}, \pathOrder[\runs^{2}]), \Init(\runs^{2}(\vrun_n)) \models_{\pval,\A} \fml_2$, and for any $(\runs^{3}, \pathOrder[\runs^{3}])$ satisfying $(\runs^{1}, \pathOrder[\runs^{1}]) \suffixEqOp (\runs^{3}, \pathOrder[\runs^{3}]) \suffixOp (\runs^{2}, \pathOrder[\runs^{2}])$, $(\runs^{3}, \pathOrder[\runs^{3}]), \Init(\runs^{3}(\vrun_n)) \models_{\pval,\A} \fml_1$ holds.
  \item $(\runs, \pathOrder), \state \models_{\pval,\A} \HAUntilFml[\compOp \paramOrInt]{\vrun_1,\vrun_2,\dots,\vrun_n}{\fml_1}{\fml_2}$ if for any extension $(\runs^{1}, \pathOrder[\runs^{1}])$ of $(\runs, \pathOrder)$ satisfying $\domain(\runs^{1}) = \domain(\runs) \disjointUnion \{\vrun_1,\vrun_2,\dots,\vrun_n\}$ and $\runs^{1}(\vrun_i) \in \Traces(\valuate{\A}{\pval}, \state)$ for each $i \in \{1,2,\dots,n\}$, there is $(\runs^{2}, \pathOrder[\runs^{2}])$ satisfying $(\runs^{1}, \pathOrder[\runs^{1}]) \suffixEqOp (\runs^{2}, \pathOrder[\runs^{2}])$, $\duration(\runs^{1} - \runs^{2}) \compOp \valuate{\paramOrInt}{\pval}$, $(\runs^{2}, \pathOrder[\runs^{2}]), \Init(\runs^{2}(\vrun_n)) \models_{\pval,\A} \fml_2$, and for any $(\runs^{3}, \pathOrder[\runs^{3}])$ satisfying $(\runs^{1}, \pathOrder[\runs^{1}]) \suffixEqOp (\runs^{3}, \pathOrder[\runs^{3}]) \suffixOp (\runs^{2}, \pathOrder[\runs^{2}])$, $(\runs^{3}, \pathOrder[\runs^{3}]), \Init(\runs^{3}(\vrun_n)) \models_{\pval,\A} \fml_1$ holds.
\end{itemize}
For a PTA $\A$, a parameter valuation $\pval \in \PVal$, and a temporal-level \HyperPTCTL{} formula $\fml$,
we write $\A \models_{\pval} \fml$ if
we have $(\emptyruns, \pathOrder[\emptyruns]), \sinit \models_{\pval, \A} \fml$, where $\sinit$ is an initial state of $\valuate{\A}{\pval}$.

For a PTA $\A$ and a parameter valuation $\pval \in \PVal$,
the satisfaction relation of the top-level \HyperPTCTL{} formulas is defined as follows:
\begin{itemize}
  \item $\A \models_{\pval} \param \compOp \ltermNN$ if we have $\valuate{\param}{\pval} \compOp \valuate{\ltermNN}{\pval}$,
\changed{where $\valuate{\ltermNN}{\pval}$ denotes the expression obtained by replacing each~$\param$ with~$\pval(\param)$ in~$\ltermNN$;} %
  \item $\A \models_{\pval} \neg \fullFml$ if we have $\A \not\models_{\pval} \fullFml$;
  \item $\A \models_{\pval} \fullFml_1 \lor \fullFml_2$ if we have $\A \models_{\pval} \fullFml_1$ or $\A \models_{\pval} \fullFml_2$;
  \item $\A \models_{\pval} \PExists \param\, \fullFml$ if there is $\pval' \in \PVal$ satisfying $\pval(\param') = \pval'(\param')$ for any $\param' \in \Param \setminus\{\param\}$ and $\A \models_{\pval'} \fullFml$.
\end{itemize}
\end{definition}
\begin{example}
 Consider the formula $\fml : \PExists \param_2 \big( \param_2 > \param_1 \land \HEUntilFml[= \param_2]{\vrun_1,\vrun_2}{(\styleact{L}_{\vrun_1} \implies \styleact{H}_{\vrun_2})}{(\styleact{H}_{\vrun_1} \land \styleact{H}_{\vrun_2})} \big)$.
 Fix $\pval(\parami{1}) = 1.8$.
 For the PTA $\A$ in \cref{figure:system_running_example}, we have
 $\A \models_{\pval} \fml$ with
 $\pval(\parami{2}) = 2.0$ and 
 $\runs^1(\vrun_1)$ and $\runs^1(\vrun_2)$ are as follows:
 in $\vrun_1$, we jump from $\loc_0$ to $\loc_1$ at 1.5 and jump from $\loc_1$ to \changed{$\loc_0$} at 2.0;
 in $\vrun_2$, we jump from $\loc_0$ to $\loc_1$ at 0.5 and jump from $\loc_1$ to \changed{$\loc_0$} at 1.0.
\end{example}
To define the semantics of \ExtHyperPTCTL{}, we introduce valuations of $\COUNT(\action_{\vrun})$ and $\LAST(\action_{\vrun})$.
A counter valuation is a function %
$\countval\colon \Actions \times \PathVar \to \setN$.
A recording valuation is a function %
$\recordval\colon \Actions \times \PathVar \to \setRnn$.
We write $\CountZero$ and $\RecordZero$ for the counter and recording valuations assigning $0$ to all $(\action, \vrun) \in \Actions \times \PathVar$, respectively.
For a linear term $\ltermCnt$ over $\{\COUNT(\action_{\vrun}) \mid \action \in \Actions, \vrun \in \PathVar\}$,
we let $\countval(\ltermCnt)$ be the inequality obtained by replacing $\COUNT(\action_{\vrun})$ with $\countval(\action, \vrun)$ and 
$\Vars(\ltermCnt) = \{\vrun \in \PathVar \mid \exists \action \in \Actions, \alpha_{\action, \vrun} \neq 0]\}$, with
$\ltermCnt = \sum_{\action \in \Actions, \vrun \in \PathVar} \alpha_{\action, \vrun} \COUNT(\action_{\vrun})$.

For paths $\trace, \trace'$ satisfying $\trace \suffixEqOp \trace'$, we let $\RISING(\action, \trace - \trace')$ be the set of positions $\state$ in $\trace$ satisfying $\Init(\trace) < \state \leq \Init(\trace')$, $\action \in \Label(\state)$, and there is $\delta \in \Time$ such that for any $\state' < \state$, $\duration(\state) - \duration(\state') < \delta$ implies $\action \not\in \Label(\state)$.
Notice that $\RISING(\action, \trace - \trace')$ is finite because $\duration(\trace - \trace')$ is finite and we have no Zeno behavior.
For a counter valuation $\countval$ and path assignments $\runs, \runs'$ satisfying $\runs \suffixEqOp \runs'$,
$\UpdateCountval{\countval}{\runs}{\runs'}$ is the counter valuation such that for any $(\action, \vrun) \in \Actions \times \PathVar$,
$\UpdateCountval{\countval}{\runs}{\runs'}(\action, \vrun) = \countval(\action, \vrun) + |\RISING(\action, \runs(\vrun) - \runs'(\vrun))|$
if $\vrun \in \domain(\runs)$
and
$\UpdateCountval{\countval}{\runs}{\runs'}(\action, \vrun) = \countval(\action, \vrun)$ otherwise.
For a recording valuation $\recordval$ and path assignments $\runs, \runs'$ satisfying $\runs \suffixEqOp \runs'$,
$\UpdateRecordval{\recordval}{\runs}{\runs'}$ is the recording valuation such that for any $(\action, \vrun) \in \Actions \times \PathVar$,
$\UpdateRecordval{\recordval}{\runs}{\runs'}(\action, \vrun) = \duration(\runs - \runs')$
if $\vrun \not\in \domain(\runs)$ or $\RISING(\action, \runs(\vrun) - \runs'(\vrun)) = \emptyset$
and otherwise,
$\UpdateRecordval{\recordval}{\runs}{\runs'}(\action, \vrun)$ is the duration of $\Init(\runs'(\vrun))$ in the suffix of $\runs(\vrun)$ starting from the last position in $\RISING(\action, \runs(\vrun) - \runs'(\vrun))$.
\begin{definition}[Semantics of \ExtHyperPTCTL{}]%
 \label{definition:semantics_ext_hyperptctl}
Let $\A$ be a PTA over parameters~$\Param_1$;
given an \ExtHyperPTCTL{} formula over~parameters~$\Param_2$,
let $\Param = \Param_1 \cup \Param_2$.
For a parameter valuation $\pval \in \PVal$,
a path assignment $(\runs, \pathOrder)$,
 a concrete state $\state$ of $\valuate{\A}{\pval}$, and
 counter and recording valuations $\countval$ and $\recordval$,
the satisfaction relation of the temporal level \ExtHyperPTCTL{} formulas is defined as follows:
\begin{itemize}
  \item $(\runs, \pathOrder), \state, \countval, \recordval \models_{\pval,\A} \action_{\vrun}$ if $\vrun \in \domain(\runs)$ and $\action \in \Label(\Init(\runs(\vrun)))$;
  \item $(\runs, \pathOrder), \state, \countval, \recordval \models_{\pval,\A} \LAST(\action_{\vrun}) - \LAST(\action'_{\vrun'}) \compOp \lterm$ if we have $\vrun, \vrun' \in \domain(\runs)$ and $\recordval(\action, \vrun) - \recordval(\action', \vrun') \compOp \pval(\lterm)$;
  \item $(\runs, \pathOrder), \state, \countval, \recordval \models_{\pval,\A} \ltermCntNN \compOp d$ if we have $\Vars(\ltermCntNN) \subseteq \domain(\runs)$ and $\countval(\ltermCntNN) \compOp d$;
  \item $(\runs, \pathOrder), \state, \countval, \recordval \models_{\pval,\A} (\ltermCnt \bmod N) \compOp d$ if we have $\Vars(\ltermCnt) \subseteq \domain(\runs)$ and $(\countval(\ltermCnt) \bmod N) \compOp d$;
  \item $(\runs, \pathOrder), \state, \countval, \recordval \models_{\pval,\A} \neg \fml$ if we have $(\runs, \pathOrder), \state, \countval, \recordval \not\models_{\pval,\A} \fml$;
  \item $(\runs, \pathOrder), \state, \countval, \recordval \models_{\pval,\A} \fml_1 \lor \fml_2$ if we have $(\runs, \pathOrder), \state, \countval, \recordval \models_{\pval,\A} \fml_1$ or $(\runs, \pathOrder), \state, \countval, \recordval \models_{\pval,\A} \fml_2$;
  \item $(\runs, \pathOrder), \state, \countval, \recordval \models_{\pval,\A} \HEUntilFml[\compOp \paramOrInt]{\vrun_1,\vrun_2,\dots,\vrun_n}{\fml_1}{\fml_2}$ if for some extension $(\runs^{1}, \pathOrder[\runs^{1}])$ of $(\runs, \pathOrder)$ satisfying $\domain(\runs^{1}) = \domain(\runs) \disjointUnion \{\vrun_1,\vrun_2,\dots,\vrun_n\}$ and $\runs^{1}(\vrun_i) \in \Traces(\valuate{\A}{\pval}, \state)$ for each $i \in \{1,2,\dots,n\}$, there is $(\runs^{2}, \pathOrder[\runs^{2}])$ satisfying $(\runs^{1}, \pathOrder[\runs^{1}]) \suffixEqOp (\runs^{2}, \pathOrder[\runs^{2}])$, $\duration(\runs^{1} - \runs^{2}) \compOp \valuate{\paramOrInt}{\pval}$, $(\runs^{2}, \pathOrder[\runs^{2}]), \Init(\runs^{2}(\vrun_n)),\UpdateCountval{\countval}{\runs^{1}}{\runs^{2}}, \UpdateRecordval{\recordval}{\runs^{1}}{\runs^{2}} \models_{\pval,\A} \fml_2$, and for any $(\runs^{3}, \pathOrder[\runs^{3}])$ satisfying $(\runs^{1}, \pathOrder[\runs^{1}]) \suffixEqOp (\runs^{3}, \pathOrder[\runs^{3}]) \suffixOp (\runs^{2}, \pathOrder[\runs^{2}])$, we have $(\runs^{3}, \pathOrder[\runs^{3}]), \Init(\runs^{3}(\vrun_n)), \UpdateCountval{\countval}{\runs^{1}}{\runs^{3}}, \UpdateRecordval{\recordval}{\runs^{1}}{\runs^{3}} \models_{\pval,\A} \fml_1$.
  \item $(\runs, \pathOrder), \state, \countval, \recordval \models_{\pval,\A} \HAUntilFml[\compOp \paramOrInt]{\vrun_1,\vrun_2,\dots,\vrun_n}{\fml_1}{\fml_2}$ if for any extension $(\runs^{1}, \pathOrder[\runs^{1}])$ of $(\runs, \pathOrder)$ satisfying $\domain(\runs^{1}) = \domain(\runs) \disjointUnion \{\vrun_1,\vrun_2,\dots,\vrun_n\}$ and $\runs^{1}(\vrun_i) \in \Traces(\valuate{\A}{\pval}, \state)$ for each $i \in \{1,2,\dots,n\}$, there is $(\runs^{2}, \pathOrder[\runs^{2}])$ satisfying $(\runs^{1}, \pathOrder[\runs^{1}]) \suffixEqOp (\runs^{2}, \pathOrder[\runs^{2}])$, $\duration(\runs^{1} - \runs^{2}) \compOp \valuate{\paramOrInt}{\pval}$, $(\runs^{2}, \pathOrder[\runs^{2}]), \Init(\runs^{2}(\vrun_n)),\UpdateCountval{\countval}{\runs^{1}}{\runs^{2}}, \UpdateRecordval{\recordval}{\runs^{1}}{\runs^2} \models_{\pval,\A} \fml_2$, and for any $(\runs^{3}, \pathOrder[\runs^{3}])$ satisfying $(\runs^{1}, \pathOrder[\runs^{1}]) \suffixEqOp (\runs^{3}, \pathOrder[\runs^{3}]) \suffixOp (\runs^{2}, \pathOrder[\runs^{2}])$, we have $(\runs^{3}, \pathOrder[\runs^{3}]), \Init(\runs^{3}(\vrun_n)), \UpdateCountval{\countval}{\runs^{1}}{\runs^{3}}, \UpdateRecordval{\recordval}{\runs^{1}}{\runs^{3}} \models_{\pval,\A} \fml_1$.
\end{itemize}
For a PTA $\A$, a parameter valuation $\pval \in \PVal$, and a temporal-level \HyperPTCTL{} formula $\fml$,
we write $\A \models_{\pval} \fml$ if
we have $(\emptyruns, \pathOrder[\emptyruns]), \sinit, \CountZero, \RecordZero \models_{\pval, \A} \fml$, where $\sinit$ is an initial state of $\valuate{\A}{\pval}$.
The satisfaction relation of the top-level \ExtHyperPTCTL{} formulas is the same as that of \HyperPTCTL{}.\@
\end{definition}
\subsection{Problems}

Here, we formalize the problems we consider in this paper. 
We consider each problem under both continuous-time and discrete-time semantics, \ie{} $\Time$ is either $\setRnn$ or~$\setN$.
We let $\Param$ be the set of parameters shared between the PTA and the \ExtOrNotHyperPTCTL{} formula.

\defProblem{\ExtOrNotHyperPTCTL{} model checking}{%
PTA $\A$ and a top-level \ExtOrNotHyperPTCTL{} formula~$\fullFml$}{%
Decide if there is $\pval \in \PVal$ satisfying $\A \models_{\pval} \fullFml$
}

\defProblem{\ExtOrNotHyperPTCTL{} parameter synthesis}{%
PTA $\A$ and a top-level \ExtOrNotHyperPTCTL{} formula~$\fullFml$}{%
Return the set $\{\pval \in \PVal \mid \A \models_{\pval} \fullFml\}$
}

The solution to the latter problem can be \emph{effectively computed} whenever its representation is symbolic, and can be represented by decidable formalisms, typically a finite union of polyhedra.

Let $\fullFml$ be a top-level \HyperPTCTL{} formula with no quantifiers over parameters.
The \emph{emptiness} of the parameter valuations to have $\A \models_{\pval} \fullFml$ can be checked by model checking of~$\fullFml$.
The \emph{universality} of the parameter valuations to have $\A \models_{\pval} \fullFml$ can be checked by model checking of $\neg (\PExists \param_1 \PExists \param_2 \dots \PExists \param_n\, \neg \fullFml)$, where $\Param = \{\param_1, \param_2, \dots, \param_n\}$.

In the rest of this paper, we focus on the \emph{nest-free} fragment of \HyperPTCTL{} (\eg{} ``\NFHyperPTCTL'', ``\NFExtHyperPTCTL'', ``\NFExistsHyperPTCTL''…), \ie{} fragments with no nesting of temporal operators.
\LongVersion{%
	See \cref{definition:syntax_next_free} for its formal definition.
}%
Observe that all our \cref{example:opacity,example:skewed_clock,example:unfairness,example:observational-non-determinism} fit into this nest-free fragment.
\changed{The following is the definition of \NFHyperPTCTL{}. The other fragments are defined similarly.

\begin{definition}
 [Syntax of \NFHyperPTCTL{}]%
 \label{definition:syntax_next_free}
 For atomic propositions $\Actions$ and parameters $\Param$,
 the syntax of \NFHyperPTCTL{} formulas of the Boolean level~$\Boolean$, the temporal level~$\fml$, and the top level~$\fullFml$ are defined as follows, where
 $\action \in \Actions$,
 $\vrun, \vrun_1, \vrun_2, \dots, \vrun_n \in \PathVar$,
 $\paramOrInt \in \Param \cup \setN$,
 $\param \in \Param$,
 ${\compOp} \in \{<, \leq, =, \geq, >\}$, and
 $\ltermNN$ is a non-negative linear term over $\Param$:
 \begin{align*}
 \Boolean \Coloneq & \top \mid \action_{\vrun} \mid \neg \Boolean \mid \Boolean \lor \Boolean
\\
 \fml \Coloneq & \HEUntilFml[\compOp \paramOrInt]{\vrun_1,\vrun_2,\dots,\vrun_n}{\Boolean}{\Boolean} \mid \HAUntilFml[\compOp \paramOrInt]{\vrun_1,\vrun_2,\dots,\vrun_n}{\Boolean}{\Boolean}
\\
 \fullFml \Coloneq & \fml \mid \param \compOp \ltermNN \mid \neg \fullFml \mid \fullFml \lor \fullFml \mid \PExists \param\, \fullFml
 \end{align*}
\end{definition}
}

\section{Reduction of Nest-Free \ExtHyperPTCTL{} synthesis to PTCTL synthesis}\label{section:reduction}
\subsection{Reduction of path variables via self-composition of PTAs}\label{section:self-composition}

We show that model checking and parameter synthesis of nest-free \ExtOrNotHyperPTCTL{} is reducible to ones of (Ext-)PTCTL by self-composition of PTAs.
For a PTA $\A = (\Actions, \Loc, \LocInit, \Clock, \Param, \invariant, \Edges, \Label)$ and $n \in \setNpos$, we let $\A^{n} = \underbrace{\A \composeOp \A \composeOp \dots \composeOp \A}_\text{$n$ times}$, and for each $i \in \{1,2,\dots,n\}$ and $\action \in \Actions$ (\resp{} $\clock \in \Clock$), we denote the corresponding atomic proposition in $\Actions^n$ (\resp{} clock in $\Clock^n$) of the $i$-th component as $\action^{i}$ (\resp{} $\clock^{i}$), where $\Actions^{n}$ and $\Clock^{n}$ are the sets of atomic propositions and clocks in $\A^n$.
We generalize the projection of paths to such $n$-compositions, \ie{} for a path $\trace$ of $\A^n$,
we let $\project{\trace}{i}$ be the projection of $\trace$ to the $i$-th component.

We define an auxiliary function $\reduce^n$ to ``compose'' the atomic propositions in \ExtOrNotHyperPTCTL{} formulas.

\begin{definition}
 [$\reduce^n$]
 For $n \in \setNpos$, the function $\reduce^n$ from nest-free temporal-level \ExtHyperPTCTL{} formulas with atomic propositions $\Actions$ to
 nest-free temporal-level \ExtPTCTL{} formulas with atomic propositions $\Actions^n$ is inductively defined as follows, with
 $\reduce^n(\sum_{\action \in \Actions, i \in \{1,2,\dots,n\}} \alpha_{\action, \vrun_i} \COUNT(\action_{\vrun_i})) = \sum_{\action^i \in \Actions^n} \alpha_{\action^i, \vrun} \COUNT(\action_{\vrun}^i)$:
 \begin{itemize}
  \item $\reduce^n(\top) = \top$
  \item $\reduce^n(\action_{\vrun_i}) = \action_{\vrun}^{i}$
  \item $\reduce^n(\LAST(\action_{\vrun_i}) - \LAST(\action_{\vrun_j}) \compOp \lterm) = \LAST(\action_{\vrun}^{i}) - \LAST(\action_{\vrun}^{j}) \compOp \lterm$
  \item $\reduce^n(\ltermCntNN \compOp d) = \reduce^n(\ltermCntNN) \compOp d$
  \item $\reduce^n((\ltermCnt \bmod N) \compOp d) = (\reduce^n(\ltermCnt) \bmod N) \compOp d$
  \item $\reduce^n(\neg \fml) = \neg \reduce^n(\fml)$
  \item $\reduce^n(\fml_1 \lor \fml_2) = \reduce^n(\fml_1) \lor \reduce^n(\fml_2)$
  \item $\reduce^n(\HEUntilFml[\compOp \paramOrInt]{\vrun_1,\vrun_2,\dots,\vrun_n}{\fml_1}{\fml_2}) = \HEUntilFml[\compOp \paramOrInt]{\vrun}{\reduce^n(\fml_1)}{\reduce^n(\fml_2)}$
  \item $\reduce^n(\HAUntilFml[\compOp \paramOrInt]{\vrun_1,\vrun_2,\dots,\vrun_n}{\fml_1}{\fml_2}) = \HAUntilFml[\compOp \paramOrInt]{\vrun}{\reduce^n(\fml_1)}{\reduce^n(\fml_2)}$
 \end{itemize}
We naturally extend $\reduce^n$ to top-level \NFExtHyperPTCTL{} formulas.
\end{definition}
\SetKwFunction{Fsynthesis}{synthesisExtPTCTL}
\SetKwFunction{Freduce}{reduceSynth}
\SetKwProg{Fn}{def}{\string:}{}
\begin{algorithm}[tbp]
 \caption{Outline of our reduction of \NFExtHyperPTCTL{} synthesis to nest-free \ExtPTCTL{} synthesis.}%
 \label{algorithm:reduction}
 \DontPrintSemicolon{}
 \ShortVersion{\small}
 \newcommand{\myCommentFont}[1]{\texttt{\small{#1}}}
 \SetCommentSty{myCommentFont}
 \KwIn{A PTA $\A$ and a \NFExtHyperPTCTL{} formula $\fullFml$}
 \KwOut{The set $\{\pval \in \PVal \mid \A \models_{\pval} \fullFml\}$}
 \Fn{\Freduce{$\A$, $\fullFml$}}{
 \Switch{$\fullFml$}{%
   \lCase{$\top$}{%
     \Return{$\PVal$}\label{algorithm:reduction:top}
   }
   \Case{$\action_{\vrun_i}$ \KwOr{} $\LAST(\action_{\vrun_i}) - \LAST(\action_{\vrun_j})$ \KwOr{} $\ltermCntNN \compOp d$ \KwOr{} $(\ltermCnt \bmod N) \compOp d$}{%
     \tcp{$\fullFml$ does not hold for empty path assignments}\label{algorithm:reduction:bottom}
     \Return{$\emptyset$}
   }
   \Case{$\neg \fullFml$}{%
     \Return{$(\setQnn)^{\Param} \setminus \Freduce(\A, \fullFml)$}\label{algorithm:reduction:negation}
   }
   \Case{$\fullFml_1 \lor \fullFml_2$}{%
     \Return{$\Freduce(\A, \fullFml_1) \cup \Freduce(\A, \fullFml_2)$}\label{algorithm:reduction:disjunction}
   }
   \Case{$\HEUntilFml[\compOp \paramOrInt]{\vrun_1,\vrun_2,\dots,\vrun_n}{\fml_1}{\fml_2}$ \KwOr{} $\HAUntilFml[\compOp \paramOrInt]{\vrun_1,\vrun_2,\dots,\vrun_n}{\fml_1}{\fml_2}$}{%
     \tcp{Use nest-free \ExtPTCTL{} synthesis}
     \Return{$\Fsynthesis(\A^n, \reduce^n(\fullFml))$}\label{algorithm:reduction:temporal}
   }
   \Case{$\param \compOp \ltermNN$}{%
     \Return{$\{\pval \in \PVal \mid \valuate{\param}{\pval} \compOp \valuate{\ltermNN}{\pval}\}$}\label{algorithm:reduction:parameter_constraint}
   }
   \Case{$\PExists \param\, \fullFml$}{%
     $\mathrm{pre} \gets \Freduce(\A, \fullFml)$\;
     \Return{$\{\pval \in \PVal \mid \exists \pval' \in \mathrm{pre}. \forall \param' \in \Param \setminus\{\param\}.\, \pval(\param') = \pval'(\param')\}$}\label{algorithm:reduction:parameter_quantifier}
   }
 }}
\end{algorithm}

\cref{algorithm:reduction} outlines our reduction of the synthesis problem.
The reduction of model checking is similar.
Our reduction is inductive on the structure of the \ExtHyperPTCTL{} formula~$\fullFml$.
Since the path assignment $(\runs, \pathOrder)$ is empty,
for atomic formulas, $\fullFml$ is satisfied (\cref{algorithm:reduction:top}) or violated (\cref{algorithm:reduction:bottom}) independent of $\A$ and $\pval$.
For Boolean cases,
we obtain the result from the result of the reduction of the immediate subformula(s) (\cref{algorithm:reduction:negation,algorithm:reduction:disjunction}).
For the temporal operators,
we use the result of the synthesis for the composed PTA $\A^n$ and the reduced formula $\reduce^n(\fullFml)$ (\cref{algorithm:reduction:temporal}).
For the remaining cases,
the result is independent of $\A$ (\cref{algorithm:reduction:parameter_constraint}) or
obtained by un-constraining the result for $\param$ (\cref{algorithm:reduction:parameter_quantifier}).

The correctness of \cref{algorithm:reduction} is immediate from the following theorem.

\newcommand{\correctnessReductionStatement}{%
 For a PTA $\A$, a temporal-level \NFExtHyperPTCTL{} formulas
 $\fml_{\exists} = \HEUntilFml[\compOp \paramOrInt]{\vrun_1,\vrun_2,\dots,\vrun_n}{\fml_1}{\fml_2}$ and
 $\fml_{\forall} = \HAUntilFml[\compOp \paramOrInt]{\vrun_1,\vrun_2,\dots,\vrun_n}{\fml_1}{\fml_2}$, and
 a parameter valuation $\pval$,
 we have $\A \models_{\pval} \fml_{\exists}$ (\resp{} $\A \models_{\pval} \fml_{\forall}$) if and only if
 we have $\A^n \models_{\pval} \reduce^n(\fml_{\exists})$ (\resp{} $\A^n \models_{\pval} \reduce^n(\fml_{\forall})$).%
}
\begin{theorem}
 \label{theorem:correctness_reduction}
 \correctnessReductionStatement{}
\end{theorem}
\begin{proof}
 [Proof (sketch)]
 \changed{Since the other cases are similar,
 we only outline the proof of}
 $\A \models_{\pval} \fml_{\exists} \implies \A^n \models_{\pval} \reduce^n(\fml_{\exists})$.
 \LongVersion{See \cref{section:proof:correctness_reduction} for the details.}
 Suppose $\A \models_{\pval} \fml_{\exists}$ holds.
 By the semantics of \ExtHyperPTCTL{},
 for some extension $(\runs^{1}, \pathOrder[\runs^{1}])$ of $(\emptyruns, \pathOrder[\emptyruns])$
 satisfying $\domain(\runs) = \{\vrun_1,\vrun_2,\dots,\vrun_n\}$ and
 $\runs^{1}(\vrun_i) \in \Traces(\valuate{\A}{\pval})$ for each $i \in \{1,2,\dots,n\}$,
 \changed{there is a suffix} $(\runs^{2}, \pathOrder[\runs^{2}])$ \changed{of $(\runs^{1}, \pathOrder[\runs^{1}])$ such that}
 we have $\duration(\runs^{1} - \runs^{2}) \compOp \valuate{\paramOrInt}{\pval}$,
 \changed{$\fml_2$ holds at $(\runs^{2}, \pathOrder[\overline{\runs}^{2}])$}, and
 for any $(\runs^{3}, \pathOrder[\runs^{3}])$ \changed{between $(\runs^{1}, \pathOrder[\runs^{1}])$ and $(\runs^{2}, \pathOrder[\runs^{2}])$},
 \changed{$\fml_1$ holds at $(\runs^{3}, \pathOrder[\runs^{3}])$}.
 Since
 there are paths $\trace$ and $\trace'$ of $\valuate{\A^n}{\pval}$
 such that 
 \begin{ienumeration}
  \item \changed{$\trace'$ is a suffix of $\trace$} and
  \item for each $i \in \{1,2,\dots,n\}$, we have
  $\runs^{1}(\vrun_i) = \project{\trace}{i}$ and
  $\runs^{2}(\vrun_i) = \project{\trace'}{i}$, 
 \end{ienumeration}
 we can construct \changed{path assignments}
 $(\overline{\runs}^{1}, \pathOrder[\overline{\runs}^{1}])$ and $(\overline{\runs}^{2}, \pathOrder[\overline{\runs}^{2}])$
 \changed{mapping $\vrun$ to them}.
 \changed{Notice that $\reduce^n(\fml_2)$ holds at such $(\overline{\runs}^{2}, \pathOrder[\overline{\runs}^{2}])$}.
 Moreover, for any path assignment $(\overline{\runs}^{3}, \pathOrder[\overline{\runs}^{3}])$ between $(\overline{\runs}^{1}, \pathOrder[\overline{\runs}^{1}])$ and $(\overline{\runs}^{2}, \pathOrder[\overline{\runs}^{2}])$,
 \changed{since there is a corresponding path assignment $(\runs^{3}, \pathOrder[\runs^{3}])$ between $(\runs^{1}, \pathOrder[\runs^{1}])$ and $(\runs^{2}, \pathOrder[\runs^{2}])$,
 $\reduce^n(\fml_1)$ holds at $(\overline{\runs}^{3}, \pathOrder[\overline{\runs}^{3}])$}
 Therefore, we have $\A^n \models_{\pval} \reduce^n(\fml_{\exists})$.
\end{proof}
\subsection{Observers for extended predicates}\label{ss:observers}

We show that the satisfaction of the additional predicates
$\LAST(\action^1_{\vrun}) - \LAST(\action^2_{\vrun}) \compOp \lterm$, 
$\ltermCntNN \compOp d$, and 
$(\ltermCnt \bmod N) \compOp d$ in \ExtPTCTL{}
are observable by a PTA, and thus, \ExtPTCTL{} model checking and synthesis are reducible to \PTCTL{} model checking and synthesis, respectively.
Since we consider \ExtPTCTL{} formulas, we assume $\PathVar = \{\vrun\}$ without loss of generality.

For \changed{\LASTExpr{}} $\LAST(\action^1_{\vrun}) - \LAST(\action^2_{\vrun}) \compOp \lterm$,
since the truth value of $\LAST(\action^1_{\vrun}) - \LAST(\action^2_{\vrun}) \compOp \lterm$ changes only when
the truth value of $\action^1_{\vrun}$ or $\action^2_{\vrun}$ changes from $\bot$ to $\top$,
we can construct an observer by ``re-evaluating'' $\LAST(\action^1_{\vrun}) - \LAST(\action^2_{\vrun}) \compOp \lterm$ when $\action^1$ or $\action^2$ changes from false to true.
Additionally, we use invariants so that the initial states depend on the parameter valuations.

For \changed{\COUNTNNExpr{}} $\ltermCntNN \compOp d$,
since $\COUNT(\action_{\vrun})$ is monotonically increasing,
we can abstract the precise value once its value is sufficiently large.
Therefore, we can encode the counted value by finite locations.

For \changed{\COUNTModExpr{}} $(\ltermCnt \bmod N) \compOp d$,
since the value of $\COUNT(\action_{\vrun})$ is cycling back at $N \in \setN$,
we can also encode the counted value modulo $N$ by finite locations.
Observers were studied in~\ShortVersion{\cite{ABBL03}}\LongVersion{\cite{ABBL03,Andre13ICECCS}}, and we define them as follows.

\begin{definition}%
 [\changed{observers for \LASTExpr{}}]%
 \label{definition:observer_last}
 \newcommand{\isInitialized}{b}
 Let $\action^1, \action^2 \in \Actions$, ${\compOp} \in \{<, \leq, =, \geq, >\}$, and $\lterm$ be a linear term over $\Param$.
 The observer for $\fml = \LAST(\action^1_{\vrun}) - \LAST(\action^2_{\vrun}) \compOp \lterm$ is a PTA %
 $\observerof{\fml} = (\Loc, \Loc, \LocInit, \{\clock_{\action^{1}}, \clock_{\action^{2}}\}, \Param, \invariant, \Edges, \Label)$, where:
 $\Loc = \powerset{\{\action^1, \action^2, \fml\}} \times \{\top, \bot\}$,
 $\LocInit$ is $\LocInit = \{ (\props, \isInitialized) \in \Loc \mid \isInitialized = \bot \}$,
 $\invariant$ is such that
 $\invariant(\loc) = \top$ for any $\loc \not\in \LocInit$, and
 for $\loc = (\props, \bot) \in \LocInit$,
 $\invariant(\loc)$ is $0 \compOp \lterm$ if $\fml \in \props$ and otherwise,
 $\invariant(\loc)$ is $\neg (0 \compOp \lterm)$,
 $\Label$ is the identity function, and {%
 \newcommand{\norise}{\ensuremath{\bigl\{((\props, \isInitialized), \top, \emptyset,(\props', \top)) \mid \props' \subsetneq \props, \fml \in \props \cap \props' \text{ or } \fml \not\in \props \cup \props'\bigr\}}}%
 \newcommand{\riseActions}{\Actions_{\mathrm{rise}}}%
 \newcommand{\currentActions}{\{\action^{1}, \action^{2}\}}%
 \newcommand{\currentGuard}{(\clock_{\action^{1}} - \clock_{\action^{2}} \compOp \lterm)[\Clock_{\riseActions} \coloneq 0]}%
 \newcommand{\risetrue}{\ensuremath{\bigl\{((\props, \isInitialized), \currentGuard, \Clock_{\riseActions}, (\props' \cup \{\fml\}, \top)) \mid \props' \subseteq \currentActions, \riseActions = \props' \setminus \props, \riseActions \neq \emptyset\bigr\}}}%
 \newcommand{\risefalse}{\ensuremath{\bigl\{((\props, \isInitialized), \neg\currentGuard, \Clock_{\riseActions}, (\props', \top)) \mid \props' \subseteq \currentActions, \riseActions = \props' \setminus \props, \riseActions \neq \emptyset\bigr\}}}%
 $\Edges = \norise \cup \risetrue \cup \risefalse$, where
 $\Clock_{\riseActions} = \{\clock_{\action^{i}} \mid \action^{i} \in \riseActions\}$ and
 $\currentGuard$ is $-\clock_{\action^{2}} \compOp \lterm$ if $\Clock_{\riseActions} = \{\clock_{\action^1}\}$, $\clock_{\action^{1}} \compOp \lterm$ if $\Clock_{\riseActions} = \{\clock_{\action^2}\}$, and $0 \compOp \lterm$ if $\Clock_{\riseActions} = \{\clock_{\action^{1}}, \clock_{\action^{2}}\}$.
 }%
\end{definition}
\begin{definition}%
 [\changed{observers for \COUNTNNExpr{}}]%
 \label{definition:observer_count_nn}
 Let $\ltermCntNN = \sum_{\action \in \Actions} \alpha_{\action, \vrun} \COUNT(\action_{\vrun})$ be a non-negative linear term, \ie{} $\alpha_{\action, \vrun} \in \setN$.
 The observer for $\fml = \ltermCntNN \compOp d$ is a PTA
 $\observerof{\fml} = (\powerset{\Actions \cup \{\fml\}}, \powerset{\Actions} \times {\{0,1,\dots, d, d+1\}}^{\Actions}, \LocInit, \emptyset, \emptyset, \invariant, \Edges, \Label)$, where:
 $\LocInit = \powerset{\Actions} \times {\{0\}}^{\Actions}$,
 $\invariant(\loc) = \top$ for any $\loc \in \Loc$, %
 $\Label$ is such that $\Label((\props, \approxcountval)) = \props \cup \{\fml\}$ if $\sum_{\action \in \Actions} \alpha_{\action, \vrun} \approxcountval(\action) \compOp d$ holds and
 $\Label((\props, \approxcountval)) = \props$ otherwise, and
 $\Edges = \{((\props, \approxcountval), \top, \emptyset, (\props', \approxcountval[\props' \setminus \props \pluseq 1])) \mid \props, \props' \subseteq \Actions, \approxcountval \in \{0,1,\dots, d, d+1\}^{\Actions}\}$, where $\approxcountval[\props' \setminus \props \pluseq 1]$ is such that $v[\props' \setminus \props \pluseq 1](\action) = v(\action)$ for $\action \not\in \props' \setminus \props$, $\approxcountval[\props' \setminus \props \pluseq 1](\action) = \approxcountval(\action) + 1$ if $\action \in \props' \setminus \props$ and $\approxcountval(\action) < d$, and $\approxcountval[\props' \setminus \props \pluseq 1](\action) = d + 1$, otherwise.
\end{definition}

\ShortVersion{\changed{We omit the definition of the observers for $\fml = (\ltermCnt \bmod N) \compOp d$ since it is similar to \cref{definition:observer_count_nn}. The main difference is to reset the ``counter'' $\approxcountval$ to 0 when the value becomes $N$.}}
\LongVersion{\begin{definition}%
 [\changed{observers for \COUNTModExpr{}}]%
 \label{definition:observer_count_mod}
 Let $\ltermCnt = \sum_{\action \in \Actions} \alpha_{\action, \vrun} \COUNT(\action_{\vrun})$ be a linear term.
 The observer for $\fml = (\ltermCnt \bmod N) \compOp d$ is a PTA
 $\observerof{\fml} = (\powerset{\Actions \cup \{\fml\}}, \powerset{\Actions} \times \{0,1,\dots, N -1\}^{\Actions}, \LocInit, \emptyset, \emptyset, \invariant, \Edges, \Label)$, where:
 $\LocInit = \powerset{\Actions} \times \{0\}^{\Actions}$,
 $\invariant(\loc) = \top$ for any $\loc \in \Loc$,
 $\Label$ is such that $\Label((\props, \approxcountval)) = \props \cup \{\fml\}$ if $\sum_{\action \in \Actions} \alpha_{\action, \vrun} \approxcountval(\action) \bmod N \compOp d$ holds and
 $\Label((\props, \approxcountval)) = \props$ otherwise, and
 $\Edges = \{((\props, \approxcountval), \top, \emptyset, (\props', \approxcountval[\props' \setminus \props \pluseq 1])) \mid \props, \props' \subseteq \Actions, \approxcountval \in \{0,1,\dots, N-1\}^{\Actions}\}$, where $\approxcountval[\props' \setminus \props \pluseq 1]$ is such that $\approxcountval[\props' \setminus \props \pluseq 1](\action) = \approxcountval(\action)$ for $\action \not\in \props' \setminus \props$, $\approxcountval[\props' \setminus \props \pluseq 1](\action) = \approxcountval(\action) + 1$ if $\action \in \props' \setminus \props$ and $\approxcountval(\action) < N - 1$, and $\approxcountval[\props' \setminus \props \pluseq 1](\action) = 0$, otherwise.
\end{definition}}

\changed{The observers semantically capture the original expressions intuitively because
$\clock_{\action^{1}}$ and $\clock_{\action^{2}}$ correspond to $\LAST(\action^1_{\vrun})$ and $\LAST(\action^2_{\vrun})$, and 
$\approxcountval$ is a sound abstraction of $\UpdateCountval{\CountZero}{\runs}{\runs'}$.}

\begin{lemma}
 [correctness of the observers]%
 \label{lemma:correctness_observers}
 For each $\action \in \Actions$, we let $\alpha_{\action} \in \setN$ and $\alpha'_{\action} \in \setZ$.
 Let $N \in \setN$, $\action^1, \action^2 \in \Actions$, ${\compOp} \in \{<, \leq, =, \geq, >\}$, $d \in \setN$, and $\lterm$ be a linear term over $\Param$.
 Let $\fml$ be one of the following:
 $\LAST(\action^1_{\vrun}) - \LAST(\action^2_{\vrun}) \compOp \lterm$,
 $\sum_{\action \in \Actions} \alpha_{\action} \COUNT(\action_{\vrun}) \compOp d$, or
 $(\sum_{\action \in \Actions} \alpha'_{\action} \COUNT(\action_{\vrun}) \bmod N) \compOp d$.
 Let $\pval$ be a valuation over $\Param$, 
 $(\runs, \pathOrder)$ be a path assignment satisfying
 $\domain(\runs) = \{\vrun\}$ and
 $\runs(\vrun) \in \Traces(\valuate{\observerof{\fml}}{\pval})$.
 For any $(\runs', \pathOrder[\runs'])$ satisfying $(\runs, \pathOrder) \suffixEqOp (\runs', \pathOrder[\runs'])$,
 we have
 $\runs', \Init(\runs'(\vrun)), \UpdateCountval{\CountZero}{\runs}{\runs'}, \UpdateRecordval{\RecordZero}{\runs}{\runs'} \models_{\pval,\observerof{\fml}} \fml$ if and only if we have $\fml \in \Label(\Init(\runs'(\vrun)))$.
 \ShortVersion{\qed{}}
\end{lemma}
\LongVersion{\begin{proof}
 [Proof (sketch)]
 \ShortVersion{\changed{Since the proof for the other expressions are similar, we prove only for \LASTExpr{}.
 Let $\fml = \LAST(\action^1_{\vrun}) - \LAST(\action^2_{\vrun}) \compOp \lterm$.}}
 Let $(\runs, \pathOrder)$ and $(\runs', \pathOrder[\runs'])$ be path assignments satisfying
 $\domain(\runs) = \domain(\runs') = \{\vrun\}$,
 $\runs(\vrun) \in \Traces(\valuate{\observerof{\fml}}{\pval})$, and
 $(\runs, \pathOrder) \suffixEqOp (\runs', \pathOrder[\runs'])$.
 Let\LongVersion{ $\countval(\action) = |\RISING(\action, \runs(\vrun) - \runs'(\vrun))|$ and let}
 $\recordval\colon\Actions\to\setRnn$ be such that
 $\recordval(\action) = \duration(\runs(\vrun) - \runs'(\vrun))$ if $\RISING(\action, \runs(\vrun) - \runs'(\vrun)) = \emptyset$
 and otherwise, $\recordval(\action)$
 is the duration of $\Init(\runs'(\vrun))$ in the suffix of $\runs(\vrun)$ starting from the last position in $\RISING(\action, \runs(\vrun) - \runs'(\vrun))$.
 Notice that, from their definition,
 \LongVersion{$\countval$ and} $\recordval$ correspond\ShortVersion{s} to
 \LongVersion{$\UpdateCountval{\CountZero}{\runs}{\runs'}$ and} $\UpdateRecordval{\RecordZero}{\runs}{\runs'}$\LongVersion{, respectively}.
 Let $(\loc, \clockval) = \Init(\runs'(\vrun))$.

 \LongVersion{For $\fml = \LAST(\action^1_{\vrun}) - \LAST(\action^2_{\vrun}) \compOp \lterm$,}
 \LongVersion{by}\ShortVersion{By} the construction of $\observerof{\fml}$,
 we have $\clockval(\action^1) = \recordval(\action^1)$ and $\clockval(\action^2) = \recordval(\action^2)$.
 Moreover, by the definition of $\Label$,
 we have $\clockval(\action^1) - \clockval(\action^2) \compOp \lterm \iff \fml \in \Label(\loc)$.
 Thus, by the semantics of \ExtPTCTL{},
 $\runs', \Init(\runs'(\vrun)), \UpdateCountval{\CountZero}{\runs}{\runs'}, \UpdateRecordval{\RecordZero}{\runs}{\runs'} \models_{\pval,\observerof{\fml}} \fml \iff \fml \in \Label(\Init(\runs'(\vrun)))$ holds.

 \LongVersion{For $\fml = \sum_{\action \in \Actions} \alpha_{\action} \COUNT(\action_{\vrun}) \compOp d$,
 let $\tilde{\countval}$ be the second element of $\loc$.
 From the construction of $\observerof{\fml}$, for any $\action \in \Actions$,
 we have $\countval(\action) = \tilde{\countval}(\action)$ if $\countval(\action) \leq d$, and otherwise,
 $\tilde{\countval}(\action) = d + 1$.
 Therefore, if $\countval(\action) \leq d$ holds for any $\action \in \Actions$,
 we have $\sum_{\action \in \Actions} \alpha_{\action} \countval(\action) \compOp d \iff \sum_{\action \in \Actions} \alpha_{\action} \tilde{\countval}(\action) \compOp d$.
 If $\countval(\action) > d$ holds for some $\action \in \Actions$,
 since we have $\alpha_{\action} \in \setN$ for any $\action \in \Actions$,
 we have $\sum_{\action \in \Actions} \alpha_{\action} \countval(\action) > d$.
 Thus, $\sum_{\action \in \Actions} \alpha_{\action} \countval(\action) \compOp d \iff \sum_{\action \in \Actions} \alpha_{\action} \tilde{\countval}(\action) \compOp d$ holds.
 Moreover, by definition of $\Label$, we have
 $\sum_{\action \in \Actions} \alpha_{\action} \countval(\action) \compOp d \iff \fml \in \Label(\Init(\runs'(\vrun)))$.
 By the semantics of \ExtPTCTL{}, we have
 $\runs', \Init(\runs'(\vrun)), \UpdateCountval{\CountZero}{\runs}{\runs'}, \UpdateRecordval{\RecordZero}{\runs}{\runs'} \models_{\pval,\observerof{\fml}} \fml \iff \fml \in \Label(\Init(\runs'(\vrun)))$.}

 \LongVersion{
 For $\fml = (\sum_{\action \in \Actions} \alpha'_{\action} \COUNT(\action_{\vrun}) \bmod N) \compOp d$,
 let $\tilde{\countval}$ be the second element of $\loc$.
 From the construction of $\observerof{\fml}$, for any $\action \in \Actions$,
 we have $\countval(\action) = \tilde{\countval}(\action) \bmod N$.
 Therefore,
 we have $\sum_{\action \in \Actions} \alpha'_{\action} \countval(\action) \bmod N \compOp d \iff \sum_{\action \in \Actions} \alpha'_{\action} \tilde{\countval}(\action) \bmod N \compOp d$.
 By definition of $\Label$, we have
 $\sum_{\action \in \Actions} \alpha'_{\action} \countval(\action) \bmod N \compOp d \iff \fml \in \Label(\Init(\runs'(\vrun)))$.
 By the semantics of \ExtPTCTL{}, we have
 $\runs', \Init(\runs'(\vrun)), \UpdateCountval{\CountZero}{\runs}{\runs'}, \UpdateRecordval{\RecordZero}{\runs}{\runs'} \models_{\pval,\observerof{\fml}} \fml$ if and only if $\fml \in \Label(\Init(\runs'(\vrun)))$ holds.}
\end{proof}}

For an \ExtPTCTL{} formula $\fullFml$, we let $\observerof{\fullFml}$ be the PTA $\observerof{\fullFml} = \bigtimes_{\mathit{ext} \in \extof{\fullFml}} \observerof{\mathit{ext}}$, where $\extof{\fullFml}$ is the set of 
\changed{\LASTExpr{}, \COUNTNNExpr{}, and \COUNTModExpr{} in $\fullFml$}.
For an \ExtPTCTL{} formula $\fullFml$,
we let $\rmextfrom{\fullFml}$ be the PTCTL formula with the same syntax but having $\extof{\fullFml}$ as additional atomic propositions.
The following shows that the model checking and synthesis for \ExtPTCTL{} formulas are reducible to those for PTCTL formulas.
\LongVersion{Intuitively, \cref{theorem:correctness_observers} holds because
\begin{ienumeration}
 \item since $\observerof{\fullFml}$ is complete and by \cref{lemma:correctness_observers},
 for any path of $\A$, there is a corresponding path of $\A \productOp \observerof{\fullFml}$ capturing the semantics of $\extof{\fullFml}$, and
 \item  for any path of $\A \productOp \observerof{\fullFml}$, its projection to $\A$ and $\observerof{\fullFml}$ captures the behavior of $\A$ and $\extof{\fullFml}$, respectively.
\end{ienumeration}
}

\newcommand{\correctnessObserversStatement}{%
 For a PTA $\A$, a parameter valuation $\pval \in \PVal$, and a top-level \ExtPTCTL{} formula $\fullFml$,
 we have $\A \models_{\pval} \fullFml$ if and only if $\A \productOp \observerof{\fullFml} \models_{\pval} \rmextfrom{\fullFml}$.%
}
\begin{theorem}
 [correctness of the reduction with $\observerof{\fullFml}$]%
 \label{theorem:correctness_observers}
 \correctnessObserversStatement{}
\end{theorem}
\begin{proof}
 [Proof (sketch)]
 \changed{Since the other cases are similar,
 we only outline the proof of} $\A \models_{\pval} \fml_{\exists} \implies \A \productOp \observerof{\fullFml} \models_{\pval} \rmextfrom{\fml_{\exists}}$, where
 $\fml_{\exists} = \HEUntilFml[\compOp \paramOrInt]{\vrun}{\fml_1}{\fml_2}$.
 \changed{Moreover, since we have $\PathVar = \{\vrun\}$, we discuss it based on paths rather than path assignments for simplicity.} \LongVersion{See \cref{section:proof:correctness_observers} for a full and more formal proof.}
 Suppose $\A \models_{\pval} \fml_{\exists}$ holds.
 By the semantics of ExtPTCTL, 
 \changed{there is a path $\trace_1$ of $\valuate{\A}{\pval}$ and a suffix $\trace_2$ of $\trace_1$ such that 
 $\duration(\trace_1 - \trace_2) \compOp \valuate{\paramOrInt}{\pval}$, 
 $\fml_2$ holds at $\trace_2$, and
 $\fml_1$ holds at any position between $\trace_1$ and $\trace_2$.
 }
 Since the observer $\observerof{\fullFml}$ is complete, there is a path \changed{$\trace'_{1}$} of \changed{$\valuate{\A \productOp \observerof{\fullFml}}{\pval}$} satisfying
 \changed{$\project{\trace_{1}}{\valuate{\A}{\pval}} = \trace'_1$}.
 Moreover, by taking a suitable suffix of \changed{$\trace'_{1}$}, 
 there is a path \changed{$\trace'_{2}$} of \changed{$\valuate{\A \productOp \observerof{\fullFml}}{\pval}$} satisfying
 \changed{$\project{\trace'_{2}}{\valuate{\A}{\pval}} = \trace_{2}$}.
 \changed{By \cref{lemma:correctness_observers}, $\rmextfrom{\fml_2}$ holds at $\trace'_{2}$.}
 \changed{For any position between $\trace'_1$ and $\trace'_2$, since there is a corresponding position between $\trace_1$ and $\trace_2$,
 $\rmextfrom{\fml_1}$ holds at that position.}
 Therefore, we have $\A \productOp \observerof{\fullFml} \models_{\pval} \rmextfrom{\fml_{\exists}}$.
\end{proof}
\subsection{Worked example}\label{section:example}

Here, we present an illustrative example of our \ExtHyperPTCTL{} synthesis semi-algorithm.
Consider again the PTA~$\A$ in \cref{figure:system_running_example}.
Let $\fullFml$ be the \NFExtHyperPTCTL{} formula in \cref{example:skewed_clock}.
\ShortVersion{\begin{figure}[tbp]}%
\LongVersion{\begin{figure*}[tbp]}
 \centering
 \begin{tikzpicture}[shorten >=1pt,scale=\ShortVersion{0.68}\LongVersion{.8},every node/.style={transform shape},every initial by arrow/.style={initial text={}}]
  \def\xDistance{10}
  \node[location,initial] (a_a) at (0,0)  [align=center]{$(\styleloc{\loc_0},\styleloc{\loc_0})$\\ $\{\styleact{H^1}, \styleact{H^2}\}$};
  \node[location] (b_a) at (\xDistance,0) [align=center]{$(\styleloc{\loc_1}, \styleloc{\loc_0})$\\ $\{\styleact{L^1}, \styleact{H^2}\}$};
  \node[location] (a_b) at (0,5)  [align=center]{$(\styleloc{\loc_0}, \styleloc{\loc_1})$\\ $\{\styleact{H^1}, \styleact{L^2}\}$};
  \node[location] (b_b) at (\xDistance,5) [align=center]{$(\styleloc{\loc_1}, \styleloc{\loc_1})$\\ $\{\styleact{L^1}, \styleact{L^1}\}$};

  \node[invariant,yshift=-1em] (a_b_invariant) at (0,6.4) {$\styleclock{c^1} \leq \styleparam{\param_1} \land \styleclock{c^2} \leq 3$};
  \node[invariant,yshift=-1em] (b_b_invariant) at (\xDistance,6.4) {$\styleclock{c^1} \leq 3 \land \styleclock{c^2} \leq 3$};
  \node[invariant,yshift=1em] (a_a_invariant) at (0,-1.4) {$\styleclock{c^1} \leq \styleparam{\param_1} \land \styleclock{c^2} \leq \styleparam{\param_1}$};
  \node[invariant,yshift=1em] (b_a_invariant) at (\xDistance,-1.4) {$\styleclock{c^1} \leq 3 \land \styleclock{c^2} \leq 3$};

  \path[->]
  (a_a) edge [bend left=5] node[below] {$\styleclock{c^1} < \styleparam{\param_1}/\styleclock{c^1} \coloneq 0$} (b_a)
  (b_a) edge [bend left=10] node[below] {$\styleclock{c^1} < 3$} (a_a)

  (a_a) edge [bend left=5] node[left,align=center] {$\styleclock{c^2} < \styleparam{\param_1}/$\\ $\styleclock{c^2} \coloneq 0$} (a_b)
  (a_b) edge [bend left=5] node[pos=0.8,right] {$\styleclock{c^2} < 3$} (a_a)

  (a_b) edge [bend left=10] node[above] {$\styleclock{c^1} < \styleparam{\param_1}/\styleclock{c^1} \coloneq 0$} (b_b)
  (b_b) edge [bend left=5] node[above] {$\styleclock{c^1} < 3$} (a_b)

  (b_a) edge [bend left=5] node[pos=0.2,left,align=center] {$\styleclock{c^2} < \styleparam{\param_1}/$\\ $\styleclock{c^2} \coloneq 0$} (b_b)
  (b_b) edge [bend left=5] node[right] {$\styleclock{c^2} < 3$} (b_a)

  (a_a) edge [bend left=5] node[pos=0.3,above left,align=center] {$\styleclock{c^1} < \styleparam{\param_1} \land \styleclock{c^2} < \styleparam{\param_1}/$\\ $\styleclock{c^1} \coloneq 0, \styleclock{c^2} \coloneq 0$} (b_b)
  (b_b) edge [bend left=5] node[pos=0.3,below right,align=center] {$\styleclock{c^1} < 3 \land \styleclock{c^2} < 3$} (a_a)
  (a_b) edge [bend left=5] node[pos=0.3,above right,align=center] {$\styleclock{c^1} < \styleparam{\param_1} \land \styleclock{c^2} < 3/$\\ $\styleclock{c^1} \coloneq 0$} (b_a)
  (b_a) edge [bend left=5] node[pos=0.3,below left,align=center] {$\styleclock{c^1} < 3 \land \styleclock{c^2} < \styleparam{\param_1}/$\\ $\styleclock{c^2} \coloneq 0$} (a_b)
  ;
 \end{tikzpicture}
 \caption{The self-composition $\A \composeOp \A$ of $\A$ in \cref{figure:system_running_example}.}%
 \label{figure:system_running_example:composed}
\ShortVersion{\end{figure}}%
\LongVersion{\end{figure*}}
First, we reduce \NFExtHyperPTCTL{} model checking to \NFExtPTCTL{} model checking via self-composition (\cref{section:self-composition}).
Since $\fullFml$ contains two path variables $\stylepathvar{\vrun_{1}}$ and~$\stylepathvar{\vrun_{2}}$,
we take the self-composition $\A \composeOp \A$ of $\A$ (\cref{figure:system_running_example:composed}).
The corresponding \NFExtPTCTL{} formula is $\reduce^2(\fullFml) = \HEUntilFml{\stylepathvar{\vrun}}{\fml'_1}{\bigl(\fml'_2 \land \fml'_3 \bigr)}$, where
$\fml'_1 = \reduce^2(\mathrm{AtMostOneDiff})$,
$\fml'_2 = \reduce^2(\mathrm{SameCount})$,
$\fml'_3 = \reduce^2(\mathrm{LargeDeviation})$.
\ShortVersion{\begin{figure}[tbp]}
\LongVersion{\begin{figure*}[tbp]}
 \centering
 \scalebox{\ShortVersion{.80}\LongVersion{1}}{\begin{tikzpicture}[shorten >=1pt,scale=\ShortVersion{0.70}\LongVersion{.9},every node/.style={transform shape},every initial by arrow/.style={initial text={}}]
  \tikzmath{
   \xDistance = 3.0;
   \x1 = 0;
   \x2 = \xDistance;
   \x3 = -\xDistance;
   \x4 = 2 * \xDistance;
   \x5 = -2 * \xDistance;
  }
  \def\xDistance{10}
  \node[initial,location] (1_1_0_0) at (\x1,3)  [align=center]{$(\{\styleact{H^1}_{\stylepathvar{\vrun}}, \styleact{H^2}_{\stylepathvar{\vrun}}, \fml'_2\}, 0, 0)$};
  \node[initial,location] (0_1_0_0) at (\x2,0)  [align=center]{$(\{\styleact{H^2}_{\stylepathvar{\vrun}}, \fml'_2\}, 0, 0)$};
  \node[initial,location] (0_0_0_0) at (\x1,0) [align=center]{$(\{\fml'_2\}, 0, 0)$};
  \node[initial,location] (1_0_0_0) at (\x3,0) [align=center]{$(\{\styleact{H^1}_{\stylepathvar{\vrun}}, \fml'_2\}, 0, 0)$};

  \node[location] (1_1_1_1) at (\x1,-3)  [align=center]{$(\{\styleact{H^1}_{\stylepathvar{\vrun}}, \styleact{H^2}_{\stylepathvar{\vrun}}, \fml'_2\}, 1, 1)$};
  \node[location] (1_0_1_0) at (\x2,-3)  [align=center]{$(\{\styleact{H^1}_{\stylepathvar{\vrun}}\}, 1, 0)$};
  \node[location] (0_1_0_1) at (\x3,-3) [align=center]{$(\{\styleact{H^2}_{\stylepathvar{\vrun}}\}, 0, 1)$};
  \node[location] (1_1_1_0) at (\x4,-3) [align=center]{$(\{\styleact{H^1}_{\stylepathvar{\vrun}}, \styleact{H^2}_{\stylepathvar{\vrun}}\}, 1, 0)$};
  \node[location] (1_1_0_1) at (\x5,-3) [align=center]{$(\{\styleact{H^1}_{\stylepathvar{\vrun}}, \styleact{H^2}_{\stylepathvar{\vrun}}\}, 0, 1)$};

  \tikzmath{
    \dotsY = -4.8;
  }
  \foreach \x in {-6, -3, 0, 3, 6}
    \node at (\x,\dotsY) {\Large $\vdots$};
  \node[rotate=30] at (7.5,\dotsY) {\Large $\vdots$};
  \node[rotate=-30] at (-7.5,\dotsY) {\Large $\vdots$};

  \path[->]
  (1_1_0_0) edge node {} (0_1_0_0)
  (1_1_0_0) edge node {}  (0_0_0_0)
  (1_1_0_0) edge node{}  (1_0_0_0)
  (0_1_0_0) edge [bend left=10] node {} (0_0_0_0)
  (1_0_0_0) edge [bend left=10] node {}  (0_0_0_0)
  (0_0_0_0) edge node {}  (1_1_1_1)
  (0_0_0_0) edge node {} (1_0_1_0)
  (0_0_0_0) edge node {}  (0_1_0_1)
  (0_1_0_0) edge node {} (1_1_1_0)
  (0_1_0_0) edge node {} (1_0_1_0)
  (1_0_0_0) edge node {} (1_1_0_1)
  (1_0_0_0) edge node {} (0_1_0_1)
  ;
 \end{tikzpicture}
 }
 \caption{A part of the observer of $\fml'_2 = (\COUNT(\styleact{H^1}_{\stylepathvar{\vrun}}) - \COUNT(\styleact{H^2}_{\stylepathvar{\vrun}}) \bmod 4) = 0$.}%
 \label{figure:system_running_example:observer_count}
\ShortVersion{\end{figure}}
\LongVersion{\end{figure*}}
\begin{figure*}[tbp]
 \centering
 \begin{tikzpicture}[shorten >=1pt,scale=\ShortVersion{0.60}\LongVersion{.7},every node/.style={transform shape},every initial by arrow/.style={initial text={}}]
  \def\xDistance{10}
  \def\initialLocationsY{8}
  \node[initial,location] (t_t_t_i) at (2.0,\initialLocationsY)  [align=center]{$(\{\styleact{H^1}_{\stylepathvar{\vrun}}, \styleact{H^2}_{\stylepathvar{\vrun}}, \fml'_2\}, \bot)$};
  \node[initial,location] (t_t_f_i) at (7.0,\initialLocationsY) [align=center]{$(\{\styleact{H^1}_{\stylepathvar{\vrun}}, \styleact{H^2}_{\stylepathvar{\vrun}}\}, \bot)$};
  \node[initial,location] (t_f_t_i) at (-3.0,\initialLocationsY)  [align=center]{$(\{\styleact{H^1}_{\stylepathvar{\vrun}}, \fml'_2\}, \bot)$};
  \node[initial,location] (t_f_f_i) at (12.0,\initialLocationsY) [align=center]{$(\{\styleact{H^1}_{\stylepathvar{\vrun}}\}, \bot)$};

  \node at (-5.5, \initialLocationsY) {\Large $\cdots$};
  \node at (13.5, \initialLocationsY) {\Large $\cdots$};

  \node[invariant,yshift=-1em] (t_t_t_i_invariant) at (2.0,10.2) {$0 \not\in [-\styleparam{\param_2}, \styleparam{\param_2}]$};
  \node[invariant,yshift=-1em] (t_t_f_i_invariant) at (7.0,10.2) {$0 \in [-\styleparam{\param_2}, \styleparam{\param_2}]$};
  \node[invariant,yshift=-1em] (t_f_t_i_invariant) at (-3.0,10.2) {$0 \not\in [-\styleparam{\param_2}, \styleparam{\param_2}]$};
  \node[invariant,yshift=-1em] (t_f_f_i_invariant) at (12.0,10.2) {$0 \in [-\styleparam{\param_2}, \styleparam{\param_2}]$};

  \node[location] (t_t_t) at (-2.75,0)  [align=center]{$(\{\styleact{H^1}_{\stylepathvar{\vrun}}, \styleact{H^2}_{\stylepathvar{\vrun}}, \fml'_2\}, \top)$};
  \node[location] (t_t_f) at (12.0,0) [align=center]{$(\{\styleact{H^1}_{\stylepathvar{\vrun}}, \styleact{H^2}_{\stylepathvar{\vrun}}\}, \top)$};
  \node[location] (t_f_t) at (0,5)  [align=center]{$(\{\styleact{H^1}_{\stylepathvar{\vrun}}, \fml'_2\}, \top)$};
  \node[location] (t_f_f) at (\xDistance,5) [align=center]{$(\{\styleact{H^1}_{\stylepathvar{\vrun}}\}, \top)$};
  \node[location] (f_f_t) at (2.5,0)  [align=center]{$(\{\fml'_2\}, \top)$};
  \node[location] (f_f_f) at (7.5,0) [align=center]{$(\emptyset, \top)$};

  \tikzmath{
    \dotsY = -1.5;
  }
  \node[rotate=30] at (-1.5,\dotsY) {\Large $\vdots$};
  \node[rotate=-30] at (2.0,\dotsY) {\Large $\vdots$};
  \node[rotate=30] at (8.0,\dotsY) {\Large $\vdots$};
  \node[rotate=-30] at (11.5,\dotsY) {\Large $\vdots$};

  \path[->]
  (t_t_t) edge [bend left=25] (t_f_t)
  (t_t_t) edge [bend left=10] (f_f_t)
  (t_f_t) edge [bend left=15] (f_f_t)
  (t_t_f) edge [bend left=5] (t_f_f)
  (t_t_f) edge (f_f_f)
  (t_f_f) edge [bend left=5] (f_f_f)
  (t_f_t) edge [bend right=35] node[above left,align=center] {$\styleclock{c}_{\styleact{H^1}_{\stylepathvar{\vrun}}} \not\in [-\styleparam{\param_2}, \styleparam{\param_2}]$\\ $/\styleclock{c}_{\styleact{H^2}_{\stylepathvar{\vrun}}} \coloneq 0$} (t_t_t)
  (f_f_t) edge [bend right=5] node[pos=0.7,below left,align=center] {$-\styleclock{c}_{\styleact{H^2}_{\stylepathvar{\vrun}}} \not\in [-\styleparam{\param_2}, \styleparam{\param_2}]$\\ $/\styleclock{c}_{\styleact{H^1}_{\stylepathvar{\vrun}}} \coloneq 0$} (t_f_t)
  (f_f_t) edge [bend left=10]node[below,align=center] {$/\styleclock{c}_{\styleact{H^1}_{\stylepathvar{\vrun}}} \coloneq 0, \styleclock{c}_{\styleact{H^2}_{\stylepathvar{\vrun}}} \coloneq 0$} (t_t_t)
  (t_f_f) edge [bend right=15] node[pos=0.1,above left,align=center] {$\styleclock{c}_{\styleact{H^1}_{\stylepathvar{\vrun}}} \not\in [-\styleparam{\param_2}, \styleparam{\param_2}]$\\ $/\styleclock{c}_{\styleact{H^2}_{\stylepathvar{\vrun}}} \coloneq 0$} (t_t_t)
  (f_f_f) edge [bend left=5] node[pos=0.05,below left,align=center] {$-\styleclock{c}_{\styleact{H^2}_{\stylepathvar{\vrun}}} \not\in [-\styleparam{\param_2}, \styleparam{\param_2}]$\\ $/\styleclock{c}_{\styleact{H^1}_{\stylepathvar{\vrun}}} \coloneq 0$} (t_f_t)
  (t_f_t) edge [bend left=15] node[pos=0.1,above right,align=center] {$\styleclock{c}_{\styleact{H^1}_{\stylepathvar{\vrun}}} \in [-\styleparam{\param_2}, \styleparam{\param_2}]$\\ $/\styleclock{c}_{\styleact{H^2}_{\stylepathvar{\vrun}}} \coloneq 0$} (t_t_f)
  (f_f_t) edge [bend left=5] node[pos=0.5,above left,align=center] {$-\styleclock{c}_{\styleact{H^2}_{\stylepathvar{\vrun}}} \in [-\styleparam{\param_2}, \styleparam{\param_2}]$\\ $/\styleclock{c}_{\styleact{H^1}_{\stylepathvar{\vrun}}} \coloneq 0$} (t_f_f)
  (t_f_f) edge [bend left=5] node[pos=0.5,above right,align=center] {$\styleclock{c}_{\styleact{H^1}_{\stylepathvar{\vrun}}} \in [-\styleparam{\param_2}, \styleparam{\param_2}]$\\ $/\styleclock{c}_{\styleact{H^2}_{\stylepathvar{\vrun}}} \coloneq 0$} (t_t_f)
  (f_f_f) edge [bend left=5] node[pos=0.2,above left,align=center] {$-\styleclock{c}_{\styleact{H^2}_{\stylepathvar{\vrun}}} \in [-\styleparam{\param_2}, \styleparam{\param_2}]$\\ $/\styleclock{c}_{\styleact{H^1}_{\stylepathvar{\vrun}}} \coloneq 0$} (t_f_f)
  ;

  \path[->]
  (t_t_t_i) edge [bend left=0] (t_f_t)
  (t_t_t_i) edge [bend left=0] (f_f_t)
  (t_f_t_i) edge [bend left=45] (f_f_t)
  (t_t_f_i) edge [bend left=5] (t_f_f)
  (t_t_f_i) edge [bend left=20] (f_f_f)
  (t_f_f_i) edge [bend left=25] (f_f_f)
  (t_f_t_i) edge [bend right=45] node[above left,align=center] {$\styleclock{c}_{\styleact{H^1}_{\stylepathvar{\vrun}}} \not\in [-\styleparam{\param_2}, \styleparam{\param_2}]$\\ $/\styleclock{c}_{\styleact{H^2}_{\stylepathvar{\vrun}}} \coloneq 0$} (t_t_t)
  (t_f_f_i) edge [bend left=45] node[pos=0.5,above right,align=center] {$\styleclock{c}_{\styleact{H^1}_{\stylepathvar{\vrun}}} \in [-\styleparam{\param_2}, \styleparam{\param_2}]$\\ $/\styleclock{c}_{\styleact{H^2}_{\stylepathvar{\vrun}}} \coloneq 0$} (t_t_f)
  ;
 \end{tikzpicture}
 \caption{A part of the observer of $\fml'_3 = \LAST(\styleact{H^1}_{\stylepathvar{\vrun}}) - \LAST(\styleact{H^2}_{\stylepathvar{\vrun}}) \not\in [-\styleparam{\param_2}, \styleparam{\param_2}]$. Most of the edges from initial locations are omitted for simplicity. The initial satisfaction of $\fml'_3$ is conditioned with the invariant.}%
 \label{figure:system_running_example:observer_last}
\end{figure*}
Then, we construct the observers $\observerof{\fml'_1}$, $\observerof{\fml'_2}$, and $\observerof{\fml'_3}$\LongVersion{ following \cref{definition:observer_last,definition:observer_count_mod}}.
\cref{figure:system_running_example:observer_count,figure:system_running_example:observer_last} show a part of $\observerof{\fml'_2}$ and $\observerof{\fml'_3}$.
Finally, we apply PTCTL synthesis to $(\A \composeOp \A) \productOp \observerof{\fml'_1} \productOp \observerof{\fml'_2} \productOp \observerof{\fml'_3}$ with $\reduce^2(\fullFml)$.
In this case, the synthesized parameter constraint is as follows: $2 \times \styleparam{\param_1} > \styleparam{\param_2} \land 3 \times \styleparam{\param_1} + 3 > 2 \times \styleparam{\param_2} \land \styleparam{\param_1} + 3 > \styleparam{\param_2} \land \styleparam{\param_1}$.
We remark that our implementation supports more general formulas, \eg{}
$\HEDiaFml{\vrun_1,\vrun_2}{(\mathrm{SameCount}' \land \mathrm{LargeDeviation})}$, with
$\mathrm{SameCount}' \equiv \COUNT(\styleact{H}_{\vrun_1}) = \COUNT(\styleact{H}_{\vrun_2})$.
\section{Decidable subclasses\LongVersion{ for model checking and synthesis}}\label{section:decidability}

The model checking problem (and the synthesis counterpart) against general PTAs is trivially undecidable,
even for the nest-free existential fragment.

\begin{proposition}
	Model checking PTAs against a \NFExistsHyperPTCTL{} formula is undecidable.
\end{proposition}
\begin{proof}
	The \NFExistsHyperPTCTL{} formula ``$\HEDiaFml[\geq 0]{\vrun}{\action}$'' is equivalent to the TCTL formula $\CTLEF{} \action$ denoting reachability.
	Reachability-emptiness is known to be undecidable for PTAs~\cite{AHV93}, which gives the result.
\end{proof}

This negative result leads us to exhibit subclasses of either the model or the formula for which decidability can be achieved,
which we do in the following.

\subsection{Non-parametric model against parametric formula}
We consider here non-parametric TAs against a restriction of \NFExtHyperPTCTL{} defined as follows:
\begin{ienumeration}
	\item parameters cannot be used in $\LAST$\LongVersion{, \ie{} their syntax is restricted to %
      $\LAST(\action_{\vrun}) - \LAST(\action_{\vrun}) \compOp n$},
      and
	\item parameters are integer-valued.
\end{ienumeration}%
\LongVersion{%
	(See formal definition in \cref{definition:class-BDR08} in \cref{appendix:NFExtHyperPTCTLBDR}.)
}%
Put it differently, parameters cannot be used in the extended syntax (the constructs that are turned into observers during our transformation in \cref{ss:observers}); we insist that parameters can be used anywhere else in the formula.
Let \NFExtHyperPTCTLBDR{} denote this class (with ``RP'' denoting a restricted use of parameters).
For instance, the opacity in \cref{example:opacity} is in this class.
\begin{theorem}[complexity of the model checking problem]\label{theorem:MC:BDR08}
	Model checking timed automata against a \NFExtHyperPTCTLBDR{} formula is in \NEXPTIME{6}.
\end{theorem}
\begin{proof}
    \changed{We reduce to model checking a non-parametric TA against PTCTL~\cite{BDR08}.}
	Recall that our general construction (\cref{section:reduction}) reduces model checking a PTA against a \NFExtHyperPTCTL{} formula to model checking a network of PTAs and a set of observers against a PTCTL formula.
	Since the observers are created for the extended syntax, and since they do not contain parameters in \NFExtHyperPTCTLBDR{}, the synchronized product of the multiple TAs\LongVersion{ (by self-composition)} and the observers does not contain parameters.

	Now, model checking a non-parametric TA against PTCTL can be done in \NEXPTIME{5} in the synchronized product of the size\footnote{%
		As in~\cite{BDR08}, by size of a formula~$\fullFml$, we mean the number of symbols used to write it, and we denote it by~$|\fullFml|$.
	} of~$\A$ and~$\fullFml$ from~\cite[Theorem~7.5]{BDR08}.
	Therefore, since $|\A^n| \leq |\A|^{|\fullFml|}$
	and due to the fact that the observers are in constant size, and come in number at most linear in~$|\fullFml|$\LongVersion{ (the size of the formula)},
	model checking a TA against a \NFExtHyperPTCTLBDR{} formula is in \NEXPTIME{6}.
	Observe that this is only an upper bound because
	\begin{oneenumeration}%
		\item we only provide a one-way reduction, and
		\item the complexity in~\cite[Theorem~7.5]{BDR08} only gives an upper bound anyway.
	\end{oneenumeration}%
\end{proof}
\begin{remark}
	Following the same reasoning, model checking a TA against a formula expressed in \LongVersion{a restriction of \NFExtHyperPTCTLBDR{} to the case where (top-level) formulas are of the form $\PExists \param_1 \PExists \param_2 \dots \PExists \param_n\, \fullFml$, where $\fullFml$ is a \NFExtHyperPTCTLBDR{} formula with no quantifiers over parameters (which could be called ``}\NFExistsExtHyperPTCTLBDR{}\LongVersion{''),}
	is in \NEXPTIME{4}, reusing the fact that model checking a non-parametric TA against \ExistsPTCTL{} can be done in \NEXPTIME{3}~\cite[Proposition~7.6]{BDR08}.
\end{remark}
\LongVersion{%
\begin{remark}\label{remark:BDR08:discrete}
	Over discrete time ($\Time = \setN$), model checking TAs against a \NFExtHyperPTCTLBDR{} formula (resp.\ \NFExistsExtHyperPTCTLBDR{}) remains in \NEXPTIME{6} (resp.\ \NEXPTIME{4}), following \cite[Corollary~7.3]{BDR08} (resp.\ \cite[Proposition~7.4]{BDR08}).
\end{remark}
}

\begin{theorem}[effective parameter synthesis]\label{theorem:NFExtHyperPTCTL:synthesis}
	The solution to the parameter synthesis problem for TAs against a \NFExtHyperPTCTL{} formula can be effectively computed.
\end{theorem}
\begin{proof}
    \changed{As in the proof of \cref{theorem:MC:BDR08}, we reduce to model checking a non-parametric TA against PTCTL~\cite{BDR08}.}
\LongVersion{%
   The solution of the synthesis problem can be effectively computed using the first-order arithmetic over the structure $\tuple{\setR, +, <, \setN, 0, 1}$ of the real numbers, where $\setN$ being the predicate ``is a natural number''.
  This arithmetic has a decidable theory with complexity in \NEXPTIME{3} in the size of the sentence~\cite{Weispfenning99}.
  }%
	From \LongVersion{\cite[Theorem~6.5]{BDR08}}\ShortVersion{\cite{BDR08}}, the solution of the parameter synthesis of a non-parametric TA against PTCTL can be effectively computed.
\end{proof}
\subsection{L/U-PTAs against \NFExtHyperPTCTL{}}
\begin{definition}[L/U-PTA~\cite{HRSV02}]\label{def:LUPTA}
	An \emph{L/U-PTA} (lower-bound/upper-bound PTA) is a PTA where the set of parameters is partitioned into lower-bound parameters and upper-bound parameters,
	where each upper-bound (resp.\ lower-bound) parameter~$\parami{i}$ must be such that,
	for every guard or invariant constraint $\clock \compOp \parami{i}$, we have ${\compOp} \in \set{ \leq, < }$ (resp.\ ${\compOp} \in \set{ \geq, > }$).
\end{definition}
\begin{example}
	The PTA in \cref{figure:system_running_example} is an L/U-PTA with an upper-bound parameter $\parami{1}$.
	\LongVersion{(This is in fact even a U-PTA~\cite{BlT09}.)

	}%
	The \LongVersion{observer fragment }PTA in \cref{figure:system_running_example:observer_last} is not L/U.
\end{example}

Because the mere $\forall \Diamond$-emptiness (emptiness of the valuations set for which a \LongVersion{given }location is always eventually reachable) is undecidable for L/U-PTAs~\cite{JLR15}\LongVersion{ and so is the nested TCTL fragment~\cite{ALR18FORMATS}}, we restrict ourselves to the reachability fragment of \NFExtHyperPTCTL{}, \LongVersion{\ie{} }where temporal operators are only~$\CTLEF{}$\LongVersion{ (CTL for reachability)}.
Let \NFExistsExtReachHyperPTCTL{} denote this fragment.

\begin{theorem}[decidability of non-parametric \NFExistsExtReachHyperPTCTL{} for L/U-PTAs]\label{theorem:MC:LU}
	Model checking L/U-PTAs against a non-parametric \NFExistsExtReachHyperPTCTL{} formula is \PSPACE{}-complete.
	The synthesis is, however, intractable.
\end{theorem}
\begin{proof}
	\changed{We reduce to reachability for L/U-PTAs.}
	Our general construction (\cref{section:reduction}) reduces model checking a PTA against a \NFExtHyperPTCTL{} formula to model checking a network of (L/U-)PTAs and a set of non-parametric observers against a TCTL formula.
	Here, we consider the reachability fragment only, leading to a reachability property.
	Reachability-emptiness is \PSPACE{}-complete for L/U-PTAs~\cite{HRSV02}, and therefore model checking L/U-PTAs against a non-parametric \NFExistsExtReachHyperPTCTL{} formula is \PSPACE{}-complete (the hardness \LongVersion{argument }following immediately).

	The non-parametric \NFExistsExtReachHyperPTCTL{} formula ``$\HEDiaFml[]{\vrun}{\action}$'' is equivalent to the (T)CTL formula $\CTLEF{} \action$ denoting reachability.
	Reachability-synthesis is known to be intractable\footnote{%
      More precisely, it cannot be effectively computed: the result cannot be represented in general using a finite union of polyhedra.
	} for L/U-PTAs~\cite{JLR15}, and therefore the synthesis for L/U-PTAs against a non-parametric \NFExistsExtReachHyperPTCTL{} is intractable.
\end{proof}

By using as proof argument a result from~\cite{ALR18FORMATS} showing that nest-free TCTL emptiness is decidable for L/U-PTAs with integer-valued parameters and \emph{without invariants}, we can show\LongVersion{ the following result}:

\begin{theorem}[decidability of non-parametric \NFExtHyperPTCTL{} for L/U-PTAs]\label{theorem:MC:LUflat}
	Model checking L/U-PTAs with integer-valued parameters without invariants against a non-parametric \NFExtHyperPTCTL{} formula is \PSPACE{}-complete.
	The synthesis is however intractable.
\end{theorem}
\LongVersion{%
	\begin{proof}
		In~\cite{ALR18FORMATS}, we showed that nest-free TCTL emptiness is \PSPACE{}-complete for L/U-PTAs with integer-valued parameters and \emph{without invariants}.
		Again, our general construction reduces model checking a PTA against a \NFExtHyperPTCTL{} formula to model checking a network of (L/U-)PTAs and a set of non-parametric observers against a TCTL formula, which is by definition nest-free.
		Therefore, model checking L/U-PTAs without invariants against a non-parametric \NFExtHyperPTCTL{} formula is \PSPACE{}-complete.

		The intractability of the synthesis comes from \cite[Remark~2]{ALR18FORMATS}.
	\end{proof}
}

\subsection{$(1, \arbitrarilyMany, 1)$-PTAs against non-parametric formula}\label{ss:1-arbitrarilyMany-1}

We use here a common notation $(n, \arbitrarilyMany, m)$ to denote $n$ parametric clocks,
arbitrarily many non-parametric clocks and $m$ parameters.
We finally show decidability in a restrictive setting, by reduction to the decidable setting of~\LongVersion{\cite{BO17,GH21}}\ShortVersion{\cite{GH21}}.

\begin{theorem}[Decidability \LongVersion{over discrete time with 1~clock}\ShortVersion{with 1 discrete clock}]\label{theorem:BO:BDR08}
	Model checking a $(1, \arbitrarilyMany, 1)$-PTA is decidable over discrete time against a non-parametric \NFExistsExtReachHyperPTCTL{} with (at most) two path quantifiers for each temporal level formula.
\end{theorem}
\begin{proof}
	\changed{We reduce to reachability for $(2, \arbitrarilyMany{}, 1)$-PTAs.}
    Model checking non-parametric \NFExistsExtReachHyperPTCTL{} reduces to model checking a network of PTAs (including the observers necessary to encode the extended syntax of the formula) against a reachability property.
    Further, the model contains a single parametric clock, and the formula contains 2 paths quantifiers for each temporal level formula, leading the self-composed model to contain 2 parametric clocks.
    Since the (unique) parameter is shared between both\LongVersion{ composed} copies, then the resulting composition is a $(2, \arbitrarilyMany{}, 1)$-PTA.
    Reachability-emptiness is \EXPSPACE{}-complete for $(2, \arbitrarilyMany{}, 1)$-PTAs for $\Time = \setN$~\cite{GH21}.
\end{proof}

It remains open whether the \emph{synthesis} problem is tractable in this latter case.

This result is not tight in the number of clocks, in the sense that it remains open whether model checking a $(2, \arbitrarilyMany, 1)$-PTA against  \NFExistsExtReachHyperPTCTL{} with 2 path quantifiers per temporal level formula is decidable or not.
However, by allowing $\CTLE \CTLU$ instead of $\CTLEF$ in the formula, we can show undecidability.
The proof encodes a $(4, 0, 1)$-PTA (for which \CTLEF{}-emptiness is undecidable~\cite{BBLS15}) into two $(2, 0, 1)$-PTAs, the synchronization of which is enforced thanks to an $\CTLE \CTLU$-based \NFExistsHyperPTCTL{}.

\begin{theorem}[Undecidability over discrete time with 2~clocks]\label{theorem:BO:BDR08:undec:EU}
	Model checking a $(2, 0, 1)$-PTA is undecidable over discrete time against a non-parametric \NFExistsHyperPTCTL{} formula using only $\CTLE \CTLU$ and two path quantifiers.
\end{theorem}
	\begin{proof}[Proof (sketch)]
		We reduce from the reachability for $(3, 0, 1)$-PTAs~\cite{BBLS15}.
		Let $\A$ be a $(4, 0, 1)$-PTA with clocks $t, x, y, z$.
		Let us ``split'' it into two $(2, 0, 1)$-PTAs with the same structure (same locations and edges), such that the first (resp.\ second) PTA only contains clock constraints containing $t$ and~$x$ (resp.\ $y$ and~$z$).
		Add to each location of the first PTA a unique location label with the same name as the location, \ie{} $\Label(\loc_i) = \{ \loc_i \}$, and to the second a primed label \ie{} $\Label(\loc_i) = \{ \loc_i' \}$.
		Fix $\loc$ a target location.
		Then let $\A'$ be the PTA union of these two PTAs, \ie{} starting with an initial non-deterministic choice in 0-time, and then ``choosing'' between either components (we assume that $y$ and~$z$ are renamed into $t$ and~$x$).
		Reaching $\loc$ in $\A$ is equivalent to checking the following \NFExistsHyperPTCTL{} formula in~$\A'$:
		\ShortVersion{$\HEUntilFml{\vrun_1, \vrun_2}{\big(\bigwedge_i (\loc_i)_{\vrun_1} = (\loc_i')_{\vrun_2} \big)}{}{( \loc_{\vrun_1} \land \loc'_{\vrun_2} )}$.}
		\LongVersion{\[ \HEUntilFml{\vrun_1, \vrun_2}{\big(\bigwedge_i (\loc_i)_{\vrun_1} = (\loc_i')_{\vrun_2} \big)}{}{( \loc_{\vrun_1} \land \loc'_{\vrun_2} )} \]}

		Since \CTLEF{}-emptiness is undecidable for $(3, 0, 1)$-PTAs~\cite{BBLS15}\LongVersion{ (and therefore for $(4, 0, 1)$-PTAs)}, then model checking this formula against $\A'$ is undecidable.
		$\A'$ contains only two clocks, and the formula is made of a single $\CTLE \CTLU$, with only two path quantifiers.
	\end{proof}
\section{Experiments}\label{section:experiments}

We experimentally evaluated the efficiency of our model checking semi-algorithm using our prototype tool~\ourTool{}\footnote{\ourTool{} is publicly available at \url{https://github.com/MasWag/HyPTCTLChecker} \changed{in an open-source manner} with all the data to reproduce the experiments.}.
Given a PTA and a \NFExtHyperPTCTL{} formula, \ourTool{} translates them into a PTA and a PTCTL formula via the reduction presented in \cref{section:reduction}, and outputs them as the format supported by \imitator{}~\cite{Andre21}, a \changed{verification} tool for~PTAs.
Then, we execute \imitator{} to solve the synthesis problem.
\changed{\ourTool{} supports all the \NFExtHyperPTCTL{} formulas except for the following operators only because \imitator{} does not support their non-hyper versions:
$\HEBoxFml[\compOp \paramOrInt]{\vrun_1, \vrun_2, \dots, \vrun_n}{\fml}$,
$\HEReleaseFml[\compOp \paramOrInt]{\vrun_1, \vrun_2, \dots, \vrun_n}{\fml_1}{\fml_2}$, and
$\HEWeakUntilFml[\compOp \paramOrInt]{\vrun_1, \vrun_2, \dots, \vrun_n}{\fml_1}{\fml_2}$.}
We pose the following research questions.
\begin{itemize}
 \item[RQ1] Is \ourTool{} efficient for practical properties?
 \item[RQ2] How many path variables can \ourTool{} handle at most?
\end{itemize}

We conducted all the experiments on an AWS EC2 m7i.4xlarge instance (with 16vCPU and 64 GiB RAM) that runs Ubuntu 22.04 LTS.\@
We used \imitator{} \href{https://github.com/imitator-model-checker/imitator/releases/tag/v3.4.0-alpha}{3.4-alpha}.
We set 6~hours as the timeout.

\subsection{Benchmarks}\label{section:benchmarks}
\begin{table}[tbp]
 \centering
 \caption{Summary of the benchmarks and the runtime of \imitator{}.
	Columns $|\Loc|$ and $|\Clock|$ show the number of locations and clocks in the PTAs.
	Columns $|\Param|_\fullFml$ and $|\Param|_\A$ show the number of parameters used in the properties and the PTAs\LongVersion{, respectively}.
	Column $|\PathVar|$ shows the number of the quantified path variables in~$\fullFml$.
	``\TIMEOUT{}'' denotes no termination within 6 hours.}%
 \label{table:summary_benchmark_result}
 \ShortVersion{\footnotesize}
 \setlength{\tabcolsep}{2pt} %
 \begin{tabular}{@{}llrrrrr|r}
  \toprule
  Prop.\ ($\fullFml$) & PTA ($\A$) & $|\Loc|$ & $|\Clock|$ & $|\Param|_\fullFml$ & $|\Param|_\A$ & $|\PathVar|$ & Time [sec.]\\
  \midrule
  \ClockDeviation{} & \Running{} & 2 & 1 & 1 & 1 & 2 & \changed{4.116}\\
  \Opacity{} & \Coffee{} & 6 & 2 & 0 & 3 & 2 & 0.723\\
  \Opacity{} & \STAC{}1:n & 8 & 2 & 0 & 2 & 2 & 0.178\\
  \Opacity{} & \STAC{}4:n & 9 & 2 & 0 & 5 & 2 & $< 0.001$\\
  \Unfairness{} & \FIFO{} & 63 & 2 & 0 & 4 & 2 & 71.955\\
  \Unfairness{} & \Priority{} & 72 & 2 & 0 & 4 & 2 & 6.855\\
  \Unfairness{} & \RoundRobin{} & 81 & 3 & 0 & 4 & 2 & 12550.979\\
  \RobustObservationalNonDeterminism{} & \Coffee{} & 6 & 2 & 1 & 3 & 2 & 3.182\\
  \RobustObservationalNonDeterminism{} & \WFAS[1]{0} & 24 & 4 & 1 & 0 & 2 & 1.665\\
  \RobustObservationalNonDeterminism{} & \WFAS[2]{0} & 24 & 4 & 1 & 0 & 2 & 2.570\\
  \RobustObservationalNonDeterminism{} & \WFAS{1} & 24 & 4 & 1 & 1 & 2 & 67.644\\
  \RobustObservationalNonDeterminism{} & \WFAS{2} & 24 & 4 & 1 & 2 & 2 & 1332.310\\
  \RobustObservationalNonDeterminism{} & \ATM{} & 7 & 2 & \changed{1} & 0 & 2 & \TIMEOUT{}\\
  \RobustObservationalNonDeterminism{} & $\ATM{}'$ & 5 & 2 & 1 & 0 & 2 & 4179.197\\
  \midrule
  \EF{2} & \Coffee{} & 6 & 2 & 1 & 0 & 2 & 0.034\\
  \EF{3} & \Coffee{} & 6 & 2 & 1 & 0 & 3 & 159.541\\
  \EF{4} & \Coffee{} & 6 & 2 & 1 & 0 & 4 & \TIMEOUT{}\\
  \bottomrule
 \end{tabular}
\end{table}
\cref{table:summary_benchmark_result} summarizes the benchmarks we used and the experimental results.
The translation time is negligible (typically $< 0.05$ sec) and is not integrated in \cref{table:summary_benchmark_result}.
We used five classes of properties: \ClockDeviation{}, \Opacity{}, \Unfairness{}, \RobustObservationalNonDeterminism{}, and \EF{i}.
\ClockDeviation{}, \Opacity{}, \Unfairness{}, and \RobustObservationalNonDeterminism{} are the properties shown in \cref{example:skewed_clock,example:opacity,example:unfairness,example:observational-non-determinism}, respectively.
\EF{i} is an artificial property to evaluate the scalability of our semi-algorithm with respect to the number of path variables.
Concretely, \EF{i} is $\HEDiaFml[\ensuremath{[\styleparam{\param},\styleparam{\param}]}]{\vrun_1, \vrun_2,\dots,\vrun_i}{\bigwedge_{j\in \{1,2,\cdots,i-1\}} \COUNT(\styleact{a}_{\stylepathvar{\vrun_j}})\allowbreak}$ $- \linebreak\allowbreak\COUNT(\styleact{a}_{\stylepathvar{\vrun_{j+1}}}) = 1$.

\Running{} is the PTA in \cref{figure:system_running_example}.
\Coffee{} \changed{(a toy coffee machine)}, \STAC{}1:n and \STAC{}4:n \changed{(two Java programs without timing leaks, translated to PTAs)} are based on the PTAs in~\cite{ALMS22}.
\FIFO{}, \Priority{}, and \RoundRobin{}\ are our original PTAs modeling FIFO, Fixed-Priority, and Round-Robin schedulers, respectively.
\WFAS{i} is \changed{a wireless fire alarm system} taken from~\cite{BBLS15}, where $i$ shows the number of parameters.
\WFAS[1]{0} and \WFAS[2]{0} are instances of \WFAS{} with different parameter valuations.
\ATM{} is \changed{a simple PTA model of an ATM from~\cite[Fig.~1]{DCLXZL18}}, and $\ATM{}'$ is its variant without the branch ``Check''.
The PTAs taken from the literature are %
 modified to align with our \changed{encoding and evaluation, \eg{} by adding locations and edges to encode input and output propositions with labels on locations and by instantiating some parameters to evaluate the scalability}.
\subsection{RQ1: Performance on practical properties}\label{section:practicality}

In \cref{table:summary_benchmark_result}, we observe that for most of the benchmarks,
the runtime of \imitator{} is less than a few minutes.
Particularly, the runtime for \Opacity{} is always less than one second.
This aligns with the efficiency of opacity verification with a similar reduction in~\cite{ALMS22}.
For \Unfairness{} and \RobustObservationalNonDeterminism{}, the runtime largely depends on the complexity of the PTA.\@
For instance, \RoundRobin{} has more locations and clocks than \FIFO{} and \Priority{} for preemptive scheduling,
which blows up the result of the self-composition in \cref{section:self-composition} and increases the runtime of \imitator{}.
Similarly, having more parameters (in \WFAS{i}) or locations (in \ATM{}) increases the runtime.
Nevertheless, \ourTool{} can still handle benchmarks with parameters in properties or PTAs if the PTAs are of mild complexity.
Thus, we answer RQ1 as follows.

\rqanswer{RQ1}{\ourTool{} can efficiently handle practical properties for mild size of PTAs\changed{, \ie{} with roughly up to 4 clocks, and against formulas with up to 3 path variables}.}

\changed{We failed to verify $\ATM{}$ within 6 hours although it has smaller $|\Loc|$, $|\Clock|$, and $|\Param|_\A$ than \WFAS{2}.
It is difficult to discuss its detail, but
this can be partly due to the structure of the PTAs.
$\ATM{}$ has two loops whereas $\ATM{}'$ has only one of them. %
After the self-composition, they have four and two loops, respectively.
It is possible that this caused a combinatorial explosion of the search space.}

\subsection{RQ2: Scalability to the number of path variables}\label{section:scalability}

In \cref{table:summary_benchmark_result}, we observe that \ourTool{} can handle \EF{3} but not \EF{4}.
This is because the self-composition in \cref{section:self-composition} exponentially blows up the PTAs with respect to the number of path variables.
Given the simplicity of \Coffee{} and \EF{3}, we answer RQ2 as follows.

\rqanswer{RQ2}{\ourTool{} can handle at most three path variables in a reasonable time.}

Although the above answer might seem quite restrictive,
we remark that three path variables are likely enough to capture most of the interesting properties.
For example, all the HyperLTL or HyperCTL* formulas in the case studies in~\cite{FRS15,BF23} are with at most two path variables
(potentially with nested temporal operators, which is out of the scope of our semi-algorithm).

\section{Conclusion and perspectives}\label{section:conclusion}

We introduced \HyperPTCTL{} as the first extension to hyperlogics of parametric timed CTL, enabling reasoning simultaneously on different execution traces.
After giving a syntax and semantics for the general logics, we restricted ourselves to a nest-free fragment, extended with $\COUNT$ and $\LAST$ constructs, allowing for reasoning about the number of actions and the duration from their last occurrence, respectively.
To our knowledge, this logic is the first of its kind to reason about parametric timed hyperproperties.
Model checking this logic \NFExtHyperPTCTL{} reduces to model checking PTCTL.
While this is, in general, undecidable, we exhibited \LongVersion{a number of }decidable subclasses.
In addition, our implementation within \ourTool{} (built on top of \LongVersion{the PTA model checker }\imitator{}) goes beyond the decidable fragment, and showed good results, both for the non-parametric case and the parametric case.

Future works include exhibiting further decidable subclasses, perhaps forbidding equality (``$=\param$'') in the formula, as in~\cite{BR07}, or with restrictions in the formula, such as in~\cite{BlT09}.
Some \LongVersion{of our }undecidability results \LongVersion{in \cref{section:decidability} }are not tight, \ie{} the exact border between decidability and undecidability remains blurred.
Devising and implementing a semi-algorithm for the full \HyperPTCTL{}, beyond the nest-free fragment, is also on our agenda.
A comparison of the expressive power between \ExtHyperPTCTL{} and other hyperlogics is also a possible future direction.

Finally, optimizing \imitator{} in order to address \HyperPTCTL{} will be an interesting challenge.
The blowup due to the self-composition may be addressed using partial order (\eg{}~\cite{ACR16}) or symmetry reductions.

\ifdefined\VersionLong
	\newcommand{\CCIS}{Communications in Computer and Information Science}
	\newcommand{\ENTCS}{Electronic Notes in Theoretical Computer Science}
	\newcommand{\FAC}{Formal Aspects of Computing}
	\newcommand{\FundInf}{Fundamenta Informaticae}
	\newcommand{\FMSD}{Formal Methods in System Design}
	\newcommand{\IJFCS}{International Journal of Foundations of Computer Science}
	\newcommand{\IJSSE}{International Journal of Secure Software Engineering}
	\newcommand{\IPL}{Information Processing Letters}
	\newcommand{\JAIR}{Journal of Artificial Intelligence Research}
	\newcommand{\JLAP}{Journal of Logic and Algebraic Programming}
	\newcommand{\JLAMP}{Journal of Logical and Algebraic Methods in Programming} %
	\newcommand{\JLC}{Journal of Logic and Computation}
	\newcommand{\LMCS}{Logical Methods in Computer Science}
	\newcommand{\LNCS}{Lecture Notes in Computer Science}
	\newcommand{\RESS}{Reliability Engineering \& System Safety}
	\newcommand{\RTS}{Real-Time Systems}
	\newcommand{\SCP}{Science of Computer Programming}
	\newcommand{\SOSYM}{Software and Systems Modeling} %
	\newcommand{\STTT}{International Journal on Software Tools for Technology Transfer}
	\newcommand{\TCS}{Theoretical Computer Science}
	\newcommand{\TOPLAS}{{ACM} Transactions on Programming Languages and Systems} %
	\newcommand{\ToPNoC}{Transactions on {P}etri Nets and Other Models of Concurrency}
	\newcommand{\TOSEM}{{ACM} Transactions on Software Engineering and Methodology} %
	\newcommand{\TSE}{{IEEE} Transactions on Software Engineering}

	\renewcommand*{\bibfont}{\small}
	\printbibliography[title={References}]
\else
	\bibliographystyle{IEEEtran}
	\newcommand{\CCIS}{CCIS}
	\newcommand{\ENTCS}{ENTCS}
	\newcommand{\FAC}{FAC}
	\newcommand{\FundInf}{FI}
	\newcommand{\FMSD}{FMSD}
	\newcommand{\IJFCS}{IJFCS}
	\newcommand{\IJSSE}{IJSSE}
	\newcommand{\IPL}{IPL}
	\newcommand{\JAIR}{JAIR}
	\newcommand{\JLAP}{JLAP}
	\newcommand{\JLAMP}{JLAMP}
	\newcommand{\JLC}{JLC}
	\newcommand{\LMCS}{LMCS}
	\newcommand{\LNCS}{LNCS}
	\newcommand{\RESS}{RESS}
	\newcommand{\RTS}{RTS}
	\newcommand{\SCP}{SCP}
	\newcommand{\SOSYM}{{SoSyM}}
	\newcommand{\STTT}{STTT}
	\newcommand{\TCS}{TCS}
	\newcommand{\TOPLAS}{ToPLAS}
	\newcommand{\ToPNoC}{ToPNOC}
	\newcommand{\TOSEM}{ToSEM}
	\newcommand{\TSE}{TSE}
  \bibliography{HyperPTCTL}
\fi

\LongVersion{%

\appendix

\section{Omitted proofs}
\subsection{Full proof of \cref{theorem:correctness_reduction}}\label{section:proof:correctness_reduction}

\LongVersion{
Before proving \cref{theorem:correctness_reduction},
we show a property on path assignments and lemmas between path assignments of the original TAs and the paths of self-composed TAs.
In the remaining part of this section,
we assume that any PTA does not have a self-loop with no guard or reset.
Note that removing such edges does not change the semantics.
For any edge $\edge$ of $\A^n$ and $i \in \{1,2,\dots,n\}$, we let $\project{\edge}{i}$ be the projection of $\edge$ to the $i$-th component.

\begin{proposition}
 \label{proposition:path_assignment_equal}
 Let $\A$ be a PTA, $\pval$ be a parameter valuation, and
 $(\runs, \pathOrder)$ be a path assignment of $\valuate{\A}{\pval}$ satisfying
 $\domain(\runs) = \{\vrun_1,\vrun_2,\dots,\vrun_n\}$ and
 $\runs(\vrun) \in \Traces(\valuate{\A}{\pval})$
 for any $\vrun \in \domain(\runs)$.
 For any $\vrun, \vrun' \in \domain(\runs)$ and $i, j \in \setN$,
 if we have $(\vrun, i) \pathOrder (\vrun', j)$ and $(\vrun', j) \pathOrder (\vrun, i)$,
 $\sum_{k = 0}^{i} d^{\vrun}_k = \sum_{k = 0}^{j} d^{\vrun'}_k$ holds, where
 $d^{\vrun}_k$ and $d^{\vrun}_k$ are the $k$-th delay in $\runs(\vrun)$ and $\runs(\vrun')$, respectively.
\end{proposition}
\begin{proof}
 It is
 immediate from the second requirement on $\pathOrder$ in \cref{definition:path_assignment}.
\end{proof}
\begin{lemma}%
 \label{lemma:composition_lemma}
 Let $\A$ be a PTA, $\pval$ be a parameter valuation, and $n \in \setNpos$.
 For any path assignment $(\runs, \pathOrder)$ of $\valuate{\A}{\pval}$ satisfying
 $\domain(\runs) = \{\vrun_1,\vrun_2,\dots,\vrun_n\}$, if
 for any $\vrun \in \domain(\runs)$, $\runs(\vrun) \in \Traces(\valuate{\A}{\pval})$ holds,
 there is a path $\trace \in \Traces(\valuate{\A^n}{\pval})$ of $\valuate{\A^n}{\pval}$
 satisfying the following conditions, where
 $\pathOrderToN$ is the monotone surjection  from $\domain(\runs) \times \setN$ to $\setN$.
 \begin{itemize}
  \item For any $\vrun_k \in \domain(\runs)$, $\runs(\vrun_k) = \project{\trace}{k}$ holds.
  \item For any $i \in \setN$, the inverse image $\pathOrderToN^{-1}(\{i\})$ satisfies
        $(\vrun_k, j) \in \pathOrderToN^{-1}(\{i\}) \iff \project{\edge_i}{k} = \edge^k_j$,
        where $\edge_i$ is the $i$-th edge of $\trace$ and $\edge^k_j$ is the $j$-th edge of $\runs(\vrun_{k})$.
 \end{itemize}
\end{lemma}
\begin{proof}
 For the sake of simplicity, we show the case of $n = 2$.
 The proof of the general case is by induction on $n$.
 Let
 $\runs(\vrun_1) = (\loc^1_0, \clockval^1_0), (d^1_0, \edge^1_0), (\loc^1_1, \clockval^1_1), \ldots$ and
 $\runs(\vrun_2) = (\loc^2_0, \clockval^2_0), (d^2_0, \edge^2_0), (\loc^2_1, \clockval^2_1), \ldots$.
 By using the inverse image $\pathOrderToN^{-1}$,
 we inductively define $\trace = (\loc_0, \clockval_0), (d_0, \edge_0), (\loc_1, \clockval_1), \ldots$ as follows:
 \begin{itemize}
  \item $\loc_0 = (\loc^1_0, \loc^2_0)$ and $\clockval_0$ is such that $\clockval_0(\clock) = \clockval^1_0(\clock)$ if $\clock$ is a clock of the first component, and otherwise, $\clockval_0(\clock) = \clockval^2_0(\clock)$.
  \item $d_i = \sum_{l = 0}^{j} d^k_l - \sum_{l = 0}^{i-1} d_l$, where $(\vrun_k, j) \in \pathOrderToN^{-1}(\{i\})$.
  \item $\edge_i$ is such that:
        \begin{itemize}
         \item if $\pathOrderToN^{-1}(\{i\}) = \{(\vrun_1, j)\}$, $\edge_i = ( (\loc^1, \loc^2), \guard, \resets, ({\loc^1}', \loc^2))$, with $\loc_i = (\loc^1, \loc^2)$ and $\edge^1_j = (\loc^1, \guard, \resets, {\loc^1}')$;
         \item if $\pathOrderToN^{-1}(\{i\}) = \{(\vrun_2, j)\}$, $\edge_i = ( (\loc^1, \loc^2), \guard, \resets, ({\loc^1}, {\loc^2}'))$, with $\loc_i = (\loc^1, \loc^2)$ and $\edge^2_j = (\loc^2, \guard, \resets, {\loc^2}')$;
         \item if $\pathOrderToN^{-1}(\{i\}) = \{(\vrun_1, j), (\vrun_2, k)\}$, $\edge_i = ((\loc^1, \loc^2), \guard^1 \land \guard^2, \resets^1\cup\resets^2, ({\loc^1}', {\loc^2}'))$, with $\loc_i = (\loc^1, \loc^2)$, $\edge^1_j = (\loc^1, \guard^1, \resets^1, {\loc^1}')$, and $\edge^2_j = (\loc^2, \guard^2, \resets^2, {\loc^2}')$.
        \end{itemize}
  \item For $\edge_{i-1} = (\loc, \guard, \resets, \loc')$, $\loc_i = \loc'$ and $\clockval_i = \reset{\clockval_{i-1} + d_{i-1}}{\resets}$.
 \end{itemize}
 Notice that $d_i$ is irrelevant to the choice of $(\vrun_k, j) \in \pathOrderToN^{-1}(\{i\})$ by \cref{proposition:path_assignment_equal}.
 From the definition of the edges of $\valuate{\A^2}{\pval}$,
 such $\trace$ is a path of $\valuate{\A^2}{\pval}$.
 Moreover, we have $\project{\trace}{1} = \runs(\vrun_1)$ and $\project{\trace}{2} = \runs(\vrun_2)$, and $\pathOrderToN^{-1}(\{i\})$ is such that $(\vrun_k, j) \in \pathOrderToN^{-1}(\{i\})$ if and only if we have $\project{\edge_i}{k} = \edge^k_j$.
\end{proof}
\begin{lemma}%
 \label{lemma:decomposition_lemma}
 Let $\A$ be a PTA, $\pval$ be a parameter valuation, and $n \in \setNpos$.
 For any path $\trace \in \Traces(\valuate{\A^n}{\pval})$ of $\valuate{\A^n}{\pval}$,
 there is a path assignment $(\runs, \pathOrder)$ of $\valuate{\A}{\pval}$ 
 satisfying the following conditions, where
 $\pathOrderToN$ is the monotone surjection from $\domain(\runs) \times \setN$ to $\setN$.
 \begin{itemize}
  \item $\domain(\runs) = \{\vrun_1,\vrun_2,\dots,\vrun_n\}$ holds.
  \item For any $\vrun_k \in \domain(\runs)$, $\runs(\vrun_k) = \project{\trace}{k}$ holds.
  \item For any $i \in \setN$, the inverse image $\pathOrderToN^{-1}(\{i\})$ satisfies
        $(\vrun_k, j) \in \pathOrderToN^{-1}(\{i\}) \iff \project{\edge_i}{k} = \edge^k_j$,
        where $\edge_i$ is the $i$-th edge of $\trace$ and $\edge^k_j$ is the $j$-th edge of $\runs(\vrun_{k})$.
 \end{itemize}
\end{lemma}
\begin{proof}
 For the sake of simplicity, we show the case of $n = 2$.
 The proof of the general case is by induction on $n$.
 Let $\trace = (\loc_0, \clockval_0), (d_0, \edge_0), (\loc_1, \clockval_1), \ldots$ and
 let $\runs$ be such that
 $\domain(\runs) = \{\vrun_1, \vrun_2\}$, and
 for each $k \in \{1,2\}$, $\runs(\vrun_k) = \project{\trace}{k}$.
 We let
 $\runs(\vrun_1) = (\loc^1_0, \clockval^1_0), (d^1_0, \edge^1_0), (\loc^1_1, \clockval^1_1), \ldots$ and
 $\runs(\vrun_2) = (\loc^2_0, \clockval^2_0), (d^2_0, \edge^2_0), (\loc^2_1, \clockval^2_1), \ldots$.
 By definition of $\A^2$ and because of $\runs(\vrun_1) = \project{\trace}{1}$ and $\runs(\vrun_2) = \project{\trace}{2}$,
 each $\edge_i$ in $\trace$ satisfies one of the following, where $\Edges$ is the set of edges in $\A$:
 \begin{itemize}
  \item $\edge_i = ( (\loc^1, \loc^2), \guard, \resets, ({\loc^1}', \loc^2))$, with $(\loc^1, \guard, \resets, {\loc^1}') = \edge^1_j \in \Edges$ for some $j \in \setN$;
  \item $\edge_i = ( (\loc^1, \loc^2), \guard, \resets, ({\loc^1}, {\loc^2}'))$, with $(\loc^2, \guard, \resets, {\loc^2}') = \edge^2_j \in \Edges$ for some $j \in \setN$;
  \item $\edge_i = ((\loc^1, \loc^2), \guard^1 \land \guard^2, \resets^1\cup\resets^2, ({\loc^1}', {\loc^2}'))$, with $(\loc^1, \guard^1, \resets^1, {\loc^1}') = \edge^1_j \in \Edges$ and $(\loc^2, \guard^2, \resets^2, {\loc^2}') = \edge^2_l \in \Edges$ for some $j,l \in \setN$.
 \end{itemize}
 We let $\pathOrderToN$ be such that
 $\pathOrderToN^{-1}(i) = \{(\vrun_1, j)\}$ for the first case,
 $\pathOrderToN^{-1}(i) = \{(\vrun_2, j)\}$ for the second case, and
 $\pathOrderToN^{-1}(i) = \{(\vrun_1, j), (\vrun_2, l)\}$ for the third case.
 Since for each $k \in \{1,2\}$, $\runs(\vrun_k) = \project{\trace}{k}$ holds,
 such $\pathOrderToN$ is a surjection from $\domain(\runs) \times \setN$ to $\setN$.
 We let $\pathOrder \subseteq (\domain(\runs) \times \setN) \times (\domain(\runs) \times \setN)$ be such that
 $(\vrun, i) \pathOrder (\vrun', j) \iff \pathOrderToN((\vrun, i)) \leq \pathOrderToN((\vrun', j))$.
 Notice that $\pathOrder$ is a total preorder and $\pathOrderToN$ is monotone.
 By definition of $\pathOrderToN$,
 for any $k, k' \in \{1,2\}$ and $i, j \in \setN$,
 $i < j$ implies $\pathOrderToN((\vrun_k, i)) < \pathOrderToN((\vrun_{k}, j))$,
 $\pathOrderToN((\vrun_k, i)) \leq \pathOrderToN((\vrun_{k'}, j))$ implies $\sum_{l = 0}^{i} d^{k}_l \leq \sum_{l = 0}^{j} d^{k'}_l$, and
 $\sum_{l = 0}^{i} d^{k}_l < \sum_{l = 0}^{j} d^{{k'}}_l$ implies $\pathOrderToN((\vrun_k, i)) \leq \pathOrderToN((\vrun_{k'}, j))$.
 By definition of $\pathOrder$,
 for any $k, k' \in \{1,2\}$ and $i, j \in \setN$,
 $i < j$ implies $(\vrun_k, i) \pathOrder (\vrun_k, j)$ and $(\vrun_k, j) \npathOrder (\vrun_k, i)$,
 $(\vrun_k, i) \pathOrder (\vrun_{k'}, j)$ implies $\sum_{l = 0}^{i} d^{k}_l \leq \sum_{l = 0}^{j} d^{k'}_l$, and
 $\sum_{l = 0}^{i} d^{k}_l < \sum_{l = 0}^{j} d^{{k'}}_l$ implies $(\vrun_k, i) \pathOrder (\vrun_{k'}, j)$.
 Overall, such $(\runs, \pathOrder)$ is a path assignment.
\end{proof}
\begin{lemma}%
 \label{lemma:suffix_lemma}
 Let $\A$ be a PTA, $\pval$ be a parameter valuation, and $n \in \setNpos$.
 Let $(\runs, \pathOrder)$ be a path assignment of $\valuate{\A}{\pval}$ satisfying
 $\domain(\runs) = \{\vrun_1,\vrun_2,\dots,\vrun_n\}$ and
 let $\trace \in \Traces(\valuate{\A^n}{\pval})$ be a path of $\valuate{\A^n}{\pval}$.
 Let $\edge_i$ be the $i$-th edge of $\trace$ and $\edge^k_j$ be the $j$-th edge of $\runs(\vrun_{k})$.
 If for any $\vrun_k \in \domain(\runs)$, we have $\runs(\vrun_k) = \project{\trace}{k}$ and
 for any $i \in \setN$, the inverse image $\pathOrderToN^{-1}(\{i\})$ satisfies $(\vrun, j) \in \pathOrderToN^{-1}(\{i\}) \iff \project{\edge_i}{k} = \edge^k_j$,
 we have the following.
 \begin{itemize}
  \item  For any $\trace'$ and $\trace''$ satisfying $\trace \suffixEqOp \trace'' \suffixOp \trace'$,
         there are $(\runs', \pathOrder[\runs'])$ and $(\runs'', \pathOrder[\runs''])$
         such that we have $(\runs, \pathOrder) \suffixEqOp (\runs'', \pathOrder'') \suffixOp (\runs', \pathOrder[\runs'])$ and
         for any $i \in \{1,2,\dots,n\}$,
         we have $\runs'(\vrun_i) = \project{\trace'}{i}$ and $\runs''(\vrun_i) = \project{\trace''}{i}$.
  \item  For any $(\runs', \pathOrder[\runs'])$ and $(\runs'', \pathOrder[\runs''])$
         satisfying $(\runs, \pathOrder) \suffixEqOp (\runs'', \pathOrder[\runs'']) \suffixOp (\runs', \pathOrder[\runs'])$,
         there are $\trace'$ and $\trace''$ such that 
         we have $\trace \suffixEqOp \trace'' \suffixOp \trace'$ and
         for any $i \in \{1,2,\dots,n\}$, we have
         $\runs'(\vrun_i) = \project{\trace'}{i}$ and
         $\runs''(\vrun_i) = \project{\trace''}{i}$.
 \end{itemize}
\end{lemma}
\begin{proof}
 For the sake of simplicity, we show the case of $n = 2$.
 The proof of the general case is by induction on $n$.
 Let $\trace = (\loc_0, \clockval_0), (d_0, \edge_0), (\loc_1, \clockval_1), \ldots$.

 For any
 $\trace' = (\loc_i, \clockval_i + d), (d_i - d, \edge_i), (\loc_{i+1}, \clockval_{i+1}), \ldots$ and
 $\trace'' = (\loc_{i'}, \clockval_{i'} + d'), (d_{i'} - d', \edge_{i'}), (\loc_{i'+1}, \clockval_{i'+1}), \ldots$
 satisfying $\trace \suffixEqOp \trace'' \suffixOp \trace'$,
 let $\runs'$, $\runs''$, $\pathOrder[\runs']$, and $\pathOrder[\runs'']$ be such that
 $\runs'(\vrun_1) = \project{\trace'}{1}$, $\runs'(\vrun_2) = \project{\trace'}{2}$,
 $\runs''(\vrun_1) = \project{\trace''}{1}$, $\runs''(\vrun_2) = \project{\trace''}{2}$, and
 for any $\vrun,\vrun' \in \domain(\runs)$ and $l, m \in \setN$, we have
 $(\vrun, l) \pathOrder[\runs''] (\vrun', m) \iff (\vrun, l + \removedEdges{\runs(\vrun)}{\runs''(\vrun)}) \pathOrder (\vrun'', m + \removedEdges{\runs(\vrun')}{\runs''(\vrun')})$ and
 $(\vrun, l) \pathOrder[\runs'] (\vrun', m) \iff (\vrun, l + \removedEdges{\runs''(\vrun)}{\runs'(\vrun)}) \pathOrder[\runs''] (\vrun', m + \removedEdges{\runs''(\vrun')}{\runs'(\vrun')})$.
 By the construction of $\trace'$, $\trace''$, $\runs'$, and $\runs''$, we have the following.
 \begin{itemize}
  \item $\duration(\runs - \runs'') = \duration(\trace - \trace'') (= d' + \sum_{j = 0}^{i' - 1} d_j)$.
  \item $\duration(\runs'' - \runs') = \duration(\trace'' - \trace') (= d + \sum_{j = 0}^{i - 1} d_j - \duration(\runs - \runs''))$.
  \item For any $k \in \{1, 2\}$, $\runs(\vrun_k) = \project{\trace}{k} \suffixEqOp \project{\trace''}{k} = \runs''(\vrun_k)$.
  \item For any $k \in \{1, 2\}$, $\runs''(\vrun_k) = \project{\trace''}{k} \suffixOp \project{\trace'}{k} = \runs'(\vrun_k)$.
  \item For any $\vrun \in \domain(\runs)$ and $j \in \setN$, $\removedEdges{\runs(\vrun)}{\runs'(\vrun)} \leq j \iff i \leq \pathOrderToN((\vrun, j))$.
  \item For any $\vrun \in \domain(\runs)$ and $j \in \setN$, $\removedEdges{\runs(\vrun)}{\runs''(\vrun)} \leq j \iff i' \leq \pathOrderToN((\vrun, j))$.
 \end{itemize}
 Particularly,
 we have $\pathOrderToN((\vrun, \removedEdges{\runs(\vrun)}{\runs'(\vrun)} - 1)) < i \leq \pathOrderToN((\vrun, \removedEdges{\runs(\vrun)}{\runs'(\vrun)}))$ and
 $\pathOrderToN((\vrun, \removedEdges{\runs(\vrun)}{\runs''(\vrun)} - 1)) < i' \leq \pathOrderToN((\vrun, \removedEdges{\runs(\vrun)}{\runs''(\vrun)}))$.
 Since $\pathOrderToN$ is monotone,
 for any $\vrun, \vrun' \in \domain(\runs)$, we have
 $(\vrun, \removedEdges{\runs(\vrun)}{\runs'(\vrun)} - 1) \pathOrder (\vrun', \removedEdges{\runs(\vrun')}{\runs'(\vrun')})$,
 $(\vrun', \removedEdges{\runs(\vrun')}{\runs'(\vrun')}) \npathOrder (\vrun, \removedEdges{\runs(\vrun)}{\runs'(\vrun)} - 1)$,
 $(\vrun, \removedEdges{\runs(\vrun)}{\runs''(\vrun)} - 1) \pathOrder (\vrun', \removedEdges{\runs(\vrun')}{\runs''(\vrun')})$, and
 $(\vrun', \removedEdges{\runs(\vrun')}{\runs''(\vrun')}) \npathOrder (\vrun, \removedEdges{\runs(\vrun)}{\runs''(\vrun)} - 1)$.
 Since $\removedEdges{\runs(\vrun')}{\runs'(\vrun')} = \removedEdges{\runs(\vrun')}{\runs''(\vrun')} + \removedEdges{\runs''(\vrun')}{\runs'(\vrun')}$,
 we have
 $(\vrun, \removedEdges{\runs''(\vrun)}{\runs'(\vrun)} - 1) \pathOrder[\runs''] (\vrun', \removedEdges{\runs''(\vrun')}{\runs'(\vrun')})$ and
 $(\vrun', \removedEdges{\runs''(\vrun')}{\runs'(\vrun')}) \npathOrder[\runs''] (\vrun, \removedEdges{\runs''(\vrun)}{\runs'(\vrun)} - 1)$.
 Overall, we have $(\runs, \pathOrder) \suffixEqOp (\runs'', \pathOrder[\runs'']) \suffixOp (\runs', \pathOrder[\runs'])$.

 For any $(\runs', \pathOrder[\runs'])$ and $(\runs'', \pathOrder[\runs''])$ satisfying $(\runs, \pathOrder) \suffixEqOp (\runs'', \pathOrder[\runs'']) \suffixOp (\runs', \pathOrder[\runs'])$,
 we let
 $\trace' = (\loc_i, \clockval_i + d), (d_i - d, \edge_i), (\loc_{i+1}, \clockval_{i+1}), \ldots$ and
 $\trace'' = (\loc_{i'}, \clockval_{i'} + d'), (d_{i'} - d', \edge_{i'}), (\loc_{i'+1}, \clockval_{i'+1}), \ldots$, where
 $i = \min_{\vrun \in \domain(\runs)} \pathOrderToN((\vrun, \removedEdges{\runs(\vrun)}{\runs'(\vrun)}))$,
 $d = \duration(\runs - \runs') - \sum_{j = 0}^{i-1} d_j$,
 $i' = \min_{\vrun \in \domain(\runs)} \pathOrderToN((\vrun, \removedEdges{\runs(\vrun)}{\runs''(\vrun)}))$, and
 $d' = \duration(\runs - \runs'') - \sum_{j = 0}^{i'-1} d_j$.
 Notice that we have $\trace \suffixEqOp \trace'' \suffixOp \trace'$.
 Since we have $\project{\trace}{1} = \runs(\vrun_1)$, $\project{\trace}{2} = \runs(\vrun_2)$,
 $\duration(\runs - \runs'') = \duration(\trace - \trace'')$, and
 for any $\vrun \in \domain(\runs)$, $j \in \setN$,
 $\removedEdges{\runs(\vrun)}{\runs''(\vrun)} \leq j \iff i' \leq \pathOrderToN((\vrun, j))$ holds,
 we have
 $\project{\trace''}{1} = \runs''(\vrun_1)$ and $\project{\trace''}{2} = \runs''(\vrun_2)$.
 Similarly,
 since we have $\project{\trace}{1} = \runs(\vrun_1)$, $\project{\trace}{2} = \runs(\vrun_2)$,
 $\duration(\runs - \runs') = \duration(\trace - \trace')$, and
 for any $\vrun \in \domain(\runs)$, $j \in \setN$,
 $\removedEdges{\runs(\vrun)}{\runs'(\vrun)} \leq j \iff i \leq \pathOrderToN((\vrun, j))$ holds,
 we have
 $\project{\trace'}{1} = \runs'(\vrun_1)$ and $\project{\trace'}{2} = \runs'(\vrun_2)$.
\end{proof}

Proof of \cref{theorem:correctness_reduction} is as follows.

\recallResult{theorem:correctness_reduction}{\correctnessReductionStatement}

\begin{proof}
 [Proof of \cref{theorem:correctness_reduction}]
 Suppose $\A \models_{\pval} \fml_{\exists}$ holds.
 By the semantics of \ExtHyperPTCTL{},
 for some extension $(\runs^{1}, \pathOrder[\runs^{1}])$ of $(\emptyruns, \pathOrder[\emptyruns])$
 satisfying $\domain(\runs) = \{\vrun_1,\vrun_2,\dots,\vrun_n\}$ and
 $\runs^{1}(\vrun_i) \in \Traces(\valuate{\A}{\pval})$ for each $i \in \{1,2,\dots,n\}$,
 there is $(\runs^{2}, \pathOrder[\runs^{2}])$ satisfying 
 $(\runs^{1}, \pathOrder[\runs^{1}]) \suffixEqOp (\runs^{2}, \pathOrder[\runs^{2}])$,
 $\duration(\runs^{1} - \runs^{2}) \compOp \valuate{\paramOrInt}{\pval}$,
 $(\runs^{2}, \pathOrder[\runs^{2}]), \Init(\runs^{2}(\vrun_n)),\UpdateCountval{\countval}{\runs^{1}}{\runs^{2}}, \UpdateRecordval{\recordval}{\runs^{1}}{\runs^{2}} \models_{\pval,\A} \fml_2$, and
 for any $(\runs^{3}, \pathOrder[\runs^{3}])$ satisfying $(\runs^{1}, \pathOrder[\runs^{1}]) \suffixEqOp (\runs^{3}, \pathOrder[\runs^{3}]) \suffixOp (\runs^{2}, \pathOrder[\runs^{2}])$, $(\runs^{3}, \pathOrder[\runs^{3}]), \Init(\runs^{3}(\vrun_n)), \UpdateCountval{\countval}{\runs^{1}}{\runs^{3}}, \UpdateRecordval{\recordval}{\runs^{1}}{\runs^{3}} \models_{\pval,\A} \fml_1$ holds.
 By \cref{lemma:composition_lemma,lemma:suffix_lemma},
 for any such $\runs^{1}$ and $\runs^{2}$,
 there are paths $\trace$ and $\trace'$ of $\valuate{\A^n}{\pval}$
 such that $\trace \suffixEqOp \trace'$ holds and
 for each $i \in \{1,2,\dots,n\}$, we have
 $\runs^{1}(\vrun_i) = \project{\trace}{i}$ and
 $\runs^{2}(\vrun_i) = \project{\trace'}{i}$.
 Let $\overline{\runs}^{1}$ and $\overline{\runs}^{2}$,
 be such that
 $\domain(\overline{\runs}^{1}) = \domain(\overline{\runs}^{2}) = \{\vrun\}$,
 $\overline{\runs}^{1}(\vrun) = \trace$, $\overline{\runs}^{2}(\vrun) = \trace'$, and
 ${\pathOrder[\overline{\runs}^{1}]} = {\pathOrder[\overline{\runs}^{2}]}$ are such that
 $(\vrun, i) \pathOrder[\overline{\runs}^{1}] (\vrun, j) \iff i \leq j$.
 For such $(\overline{\runs}^{1}, \pathOrder[\overline{\runs}^{1}])$ and $(\overline{\runs}^{2}, \pathOrder[\overline{\runs}^{2}])$,
 we have $(\overline{\runs}^{1}, \pathOrder[\overline{\runs}^{1}]) \suffixEqOp (\overline{\runs}^{2}, \pathOrder[\overline{\runs}^{2}])$,
 $\duration(\overline{\runs}^{1} - \overline{\runs}^{2}) \compOp \valuate{\paramOrInt}{\pval}$, and
 $(\overline{\runs}^{2}, \pathOrder[\overline{\runs}^{2}]), \Init(\overline{\runs}^{2}(\vrun)),\UpdateCountval{\countval}{\overline{\runs}^{1}}{\overline{\runs}^{2}}, \UpdateRecordval{\recordval}{\overline{\runs}^{1}}{\overline{\runs}^{2}} \models_{\pval,\A} \reduce^n(\fml_2)$.
 Let $(\overline{\runs}^{3}, \pathOrder[\overline{\runs}^{3}])$ be a path assignment satisfying $(\overline{\runs}^{1}, \pathOrder[\overline{\runs}^{1}]) \suffixEqOp (\overline{\runs}^{3}, \pathOrder[\overline{\runs}^{3}]) \suffixOp (\overline{\runs}^{2}, \pathOrder[\overline{\runs}^{2}])$.
 By \cref{lemma:suffix_lemma},
 for any such $(\overline{\runs}^{3}, \pathOrder[\overline{\runs}^{3}])$,
 there is a path assignment $(\runs^{3}, \pathOrder[\runs^{3}])$ satisfying
 $(\runs^{1}, \pathOrder[\runs^{1}]) \suffixEqOp (\runs^{3}, \pathOrder[\runs^{3}]) \suffixOp (\runs^{2}, \pathOrder[\runs^{2}])$, and
 $\runs^{3}(\vrun_i) = \project{\overline{\runs}^{3}(\vrun)}{i}$ for each $i \in \{1,2,\dots,n\}$.
 Since we have $\A \models_{\pval} \fml_{\exists}$, for such $(\runs^{3}, \pathOrder[\runs^{3}])$,
 we have $(\runs^{3}, \pathOrder[\runs^{3}]), \Init(\runs^{3}(\vrun_n)), \UpdateCountval{\countval}{\runs^{1}}{\runs^{3}}, \UpdateRecordval{\recordval}{\runs^{1}}{\runs^{3}} \models_{\pval,\A} \fml_1$.
 Since for each $i \in \{1,2,\dots,n\}$, we have $\runs^{3}(\vrun_i) = \project{\overline{\runs}^{3}(\vrun)}{i}$,
 for such $(\overline{\runs}^{3}, \pathOrder[\overline{\runs}^{3}])$,
 we have $(\overline{\runs}^{3}, \pathOrder[\overline{\runs}^{3}]), \Init(\overline{\runs}^{3}(\vrun)), \UpdateCountval{\countval}{\overline{\runs}^{1}}{\overline{\runs}^{3}}, \UpdateRecordval{\recordval}{\overline{\runs}^{1}}{\overline{\runs}^{3}} \models_{\pval,\A} \reduce^n(\fml_1)$.
 Therefore, we have $\A^n \models_{\pval} \reduce^n(\fml_{\exists})$.

 Suppose $\A^n \models_{\pval} \reduce^n(\fml_{\exists})$ holds.
 By the semantics of \ExtHyperPTCTL{},
 for some extension $(\overline{\runs}^{1}, \pathOrder[\overline{\runs}^{1}])$ of $(\emptyruns, \pathOrder[\emptyruns])$
 satisfying $\domain(\overline{\runs}) = \{\vrun\}$ and $\overline{\runs}^{1}(\vrun) \in \Traces(\valuate{\A^n}{\pval})$,
 there is $(\overline{\runs}^{2}, \pathOrder[\overline{\runs}^{2}])$ satisfying 
 $(\overline{\runs}^{1}, \pathOrder[\overline{\runs}^{1}]) \suffixEqOp (\overline{\runs}^{2}, \pathOrder[\overline{\runs}^{2}])$,
 $\duration(\overline{\runs}^{1} - \overline{\runs}^{2}) \compOp \valuate{\paramOrInt}{\pval}$,
 $(\overline{\runs}^{2}, \pathOrder[\overline{\runs}^{2}]), \Init(\overline{\runs}^{2}(\vrun)),\UpdateCountval{\countval}{\overline{\runs}^{1}}{\overline{\runs}^{2}}, \UpdateRecordval{\recordval}{\overline{\runs}^{1}}{\overline{\runs}^{2}} \models_{\pval,{\A}^{n}} \reduce^n(\fml_2)$, and
 for any $(\overline{\runs}^{3}, \pathOrder[\overline{\runs}^{3}])$ satisfying $(\overline{\runs}^{1}, \pathOrder[\overline{\runs}^{1}]) \suffixEqOp (\overline{\runs}^{3}, \pathOrder[\overline{\runs}^{3}]) \suffixOp (\overline{\runs}^{2}, \pathOrder[\overline{\runs}^{2}])$, $(\overline{\runs}^{3}, \pathOrder[\overline{\runs}^{3}]), \Init(\overline{\runs}^{3}(\vrun)), \UpdateCountval{\countval}{\overline{\runs}^{1}}{\overline{\runs}^{3}}, \UpdateRecordval{\recordval}{\overline{\runs}^{1}}{\overline{\runs}^{3}} \models_{\pval,{\A}^{n}} \reduce^n(\fml_1)$ holds.
 By \cref{lemma:decomposition_lemma,lemma:suffix_lemma},
 for any such $\overline{\runs}^{1}(\vrun)$ and $\overline{\runs}^{2}(\vrun)$,
 there are path assignments $(\runs^{1}, \pathOrder[\runs^{1}])$ and $(\runs^{2}, \pathOrder[\runs^{2}])$
 such that
 $\domain(\runs^{1}) = \domain(\runs^{2}) = \{\vrun_1, \vrun_2, \dots, \vrun_n\}$,
 $(\runs^{1}, \pathOrder[\runs^{1}]) \suffixEqOp (\runs^{2}, \pathOrder[\runs^{2}])$, and
 for each $i \in \{1,2,\dots,n\}$, we have
 $\runs^{1}(\vrun_i) = \project{\overline{\runs}^{1}(\vrun)}{i}$ and
 $\runs^{2}(\vrun_i) = \project{\overline{\runs}^{2}(\vrun)}{i}$.
 For such $(\runs^{1}, \pathOrder[\runs^{1}])$ and $(\runs^{2}, \pathOrder[\runs^{2}])$,
 we have $\duration(\runs^{1} - \runs^{2}) \compOp \valuate{\paramOrInt}{\pval}$ and
 $(\runs^{2}, \pathOrder[\runs^{2}]), \Init(\runs^{2}(\vrun_n)),\UpdateCountval{\countval}{\runs^{1}}{\runs^{2}}, \UpdateRecordval{\recordval}{\runs^{1}}{\runs^{2}} \models_{\pval,\A} \fml_2$.
 Let $(\runs^{3}, \pathOrder[\runs^{3}])$ be a path assignment satisfying
 $(\runs^{1}, \pathOrder[\runs^{1}]) \suffixEqOp (\runs^{3}, \pathOrder[\runs^{3}]) \suffixOp (\runs^{2}, \pathOrder[\runs^{2}])$.
 By \cref{lemma:suffix_lemma},
 for any such $(\runs^{3}, \pathOrder[\runs^{3}])$,
 there is a path assignment $(\overline{\runs}^{3}, \pathOrder[\overline{\runs}^{3}])$ satisfying
 $(\overline{\runs}^{1}, \pathOrder[\overline{\runs}^{1}]) \suffixEqOp (\overline{\runs}^{3}, \pathOrder[\overline{\runs}^{3}]) \suffixOp (\overline{\runs}^{2}, \pathOrder[\overline{\runs}^{2}])$, and
 $\runs^{3}(\vrun_i) = \project{\overline{\runs}^{3}(\vrun)}{i}$ for each $i \in \{1,2,\dots,n\}$.
 Since we have $\A^n \models_{\pval} \reduce^n(\fml_{\exists})$, for such $(\overline{\runs}^{3}, \pathOrder[\overline{\runs}^{3}])$,
 we have $(\overline{\runs}^{3}, \pathOrder[\overline{\runs}^{3}]), \Init(\overline{\runs}^{3}(\vrun)), \UpdateCountval{\countval}{\overline{\runs}^{1}}{\overline{\runs}^{3}}, \UpdateRecordval{\recordval}{\overline{\runs}^{1}}{\overline{\runs}^{3}} \models_{\pval,\A} \reduce^n(\fml_1)$.
 Since for each $i \in \{1,2,\dots,n\}$, we have $\runs^{3}(\vrun_i) = \project{\overline{\runs}^{3}(\vrun)}{i}$,
 for such $(\runs^{3}, \pathOrder[\runs^{3}])$, 
 we have $(\runs^{3}, \pathOrder[\runs^{3}]), \Init(\runs^{3}(\vrun_n)), \UpdateCountval{\countval}{\runs^{1}}{\runs^{3}}, \UpdateRecordval{\recordval}{\runs^{1}}{\runs^{3}} \models_{\pval,\A} \fml_1$.
 Therefore, we have $\A \models_{\pval} \fml_{\exists}$.

 Suppose $\A \models_{\pval} \fml_{\forall}$ holds.
 Let $(\overline{\runs}^{1}, \pathOrder[\overline{\runs}^1])$ be a path assignment satisfying
 $\domain(\overline{\runs}^{1}) = \{\vrun\}$ and
 $\overline{\runs}^{1}(\vrun) \in \Traces(\valuate{\A^n}{\pval})$.
 By \cref{lemma:decomposition_lemma},
 there is a path assignment $(\runs^{1}, \pathOrder[\runs^{1}])$ of $\valuate{\A}{\pval}$ 
 satisfying the following conditions, where
 $\pathOrderToN$ is the monotone surjection from $\domain(\runs^{1}) \times \setN$ to $\setN$.
 \begin{itemize}
  \item $\domain(\runs^{1}) = \{\vrun_1,\vrun_2,\dots,\vrun_n\}$ holds.
  \item For any $\vrun_k \in \domain(\runs^{1})$, $\runs^{1}(\vrun_k) = \project{\trace}{k}$ holds.
  \item For any $i \in \setN$, the inverse image $\pathOrderToN^{-1}(\{i\})$ satisfies
        $(\vrun_k, j) \in \pathOrderToN^{-1}(\{i\}) \iff \project{\edge_i}{k} = \edge^k_j$,
        where $\edge_i$ is the $i$-th edge of $\trace$ and $\edge^k_j$ is the $j$-th edge of $\runs^{1}(\vrun_{k})$.
 \end{itemize}
 Since we have $\A \models_{\pval} \fml_{\forall}$,
 for any such $(\runs^{1}, \pathOrder[\runs^{1}])$,
 there is $(\runs^{2}, \pathOrder[\runs^{2}])$ satisfying 
 $(\runs^{1}, \pathOrder[\runs^{1}]) \suffixEqOp (\runs^{2}, \pathOrder[\runs^{2}])$,
 $\duration(\runs^{1} - \runs^{2}) \compOp \valuate{\paramOrInt}{\pval}$,
 $(\runs^{2}, \pathOrder[\runs^{2}]), \Init(\runs^{2}(\vrun_n)),\UpdateCountval{\countval}{\runs^{1}}{\runs^{2}}, \UpdateRecordval{\recordval}{\runs^{1}}{\runs^{2}} \models_{\pval,\A} \fml_2$, and
 for any $(\runs^{3}, \pathOrder[\runs^{3}])$ satisfying $(\runs^{1}, \pathOrder[\runs^{1}]) \suffixEqOp (\runs^{3}, \pathOrder[\runs^{3}]) \suffixOp (\runs^{2}, \pathOrder[\runs^{2}])$, $(\runs^{3}, \pathOrder[\runs^{3}]), \Init(\runs^{3}(\vrun_n)), \UpdateCountval{\countval}{\runs^{1}}{\runs^{3}}, \UpdateRecordval{\recordval}{\runs^{1}}{\runs^{3}} \models_{\pval,\A} \fml_1$ holds.
 By \cref{lemma:suffix_lemma},
 for any such $\runs^{2}$,
 there is a path $\trace'$ of $\valuate{\A^n}{\pval}$
 such that for each $i \in \{1,2,\dots,n\}$, we have
 $\runs^{2}(\vrun_i) = \project{\trace'}{i}$.
 Let $\overline{\runs}^{2}$ be such that
 $\domain(\overline{\runs}^{2}) = \{\vrun\}$,
 $\overline{\runs}^{2}(\vrun) = \trace'$, and
 ${\pathOrder[\overline{\runs}^{2}]}$ is such that
 $(\vrun, i) \pathOrder[\overline{\runs}^{2}] (\vrun, j) \iff i \leq j$.
 Notice that we have ${\pathOrder[\overline{\runs}^{1}]} = {\pathOrder[\overline{\runs}^{2}]}$
 because of $\domain(\overline{\runs}^{1}) = \domain(\overline{\runs}^{2}) = \{\vrun\}$.
 For such $\overline{\runs}^{2}$,
 we have $(\overline{\runs}^{1}, \pathOrder[\overline{\runs}^{1}]) \suffixEqOp (\overline{\runs}^{2}, \pathOrder[\overline{\runs}^{2}])$,
 $\duration(\overline{\runs}^{1} - \overline{\runs}^{2}) \compOp \valuate{\paramOrInt}{\pval}$, and
 $(\overline{\runs}^{2}, \pathOrder[\overline{\runs}^{2}]), \Init(\overline{\runs}^{2}(\vrun)),\UpdateCountval{\countval}{\overline{\runs}^{1}}{\overline{\runs}^{2}}, \UpdateRecordval{\recordval}{\overline{\runs}^{1}}{\overline{\runs}^{2}} \models_{\pval,\A} \reduce^n(\fml_2)$.
 Let $(\overline{\runs}^{3}, \pathOrder[\overline{\runs}^{3}])$ be a path assignment satisfying $(\overline{\runs}^{1}, \pathOrder[\overline{\runs}^{1}]) \suffixEqOp (\overline{\runs}^{3}, \pathOrder[\overline{\runs}^{3}]) \suffixOp (\overline{\runs}^{2}, \pathOrder[\overline{\runs}^{2}])$.
 By \cref{lemma:suffix_lemma},
 for any such $(\overline{\runs}^{3}, \pathOrder[\overline{\runs}^{3}])$,
 there is a path assignment $(\runs^{3}, \pathOrder[\runs^{3}])$ satisfying
 $(\runs^{1}, \pathOrder[\runs^{1}]) \suffixEqOp (\runs^{3}, \pathOrder[\runs^{3}]) \suffixOp (\runs^{2}, \pathOrder[\runs^{2}])$, and
 $\runs^{3}(\vrun_i) = \project{\overline{\runs}^{3}(\vrun)}{i}$ for each $i \in \{1,2,\dots,n\}$.
 Since we have $\A \models_{\pval} \fml_{\forall}$, for such $(\runs^{3}, \pathOrder[\runs^{3}])$,
 we have $(\runs^{3}, \pathOrder[\runs^{3}]), \Init(\runs^{3}(\vrun_n)), \UpdateCountval{\countval}{\runs^{1}}{\runs^{3}}, \UpdateRecordval{\recordval}{\runs^{1}}{\runs^{3}} \models_{\pval,\A} \fml_1$.
 Since for each $i \in \{1,2,\dots,n\}$, we have $\runs^{3}(\vrun_i) = \project{\overline{\runs}^{3}(\vrun)}{i}$,
 for such $(\overline{\runs}^{3}, \pathOrder[\overline{\runs}^{3}])$,
 we have $(\overline{\runs}^{3}, \pathOrder[\overline{\runs}^{3}]), \Init(\overline{\runs}^{3}(\vrun)), \UpdateCountval{\countval}{\overline{\runs}^{1}}{\overline{\runs}^{3}}, \UpdateRecordval{\recordval}{\overline{\runs}^{1}}{\overline{\runs}^{3}} \models_{\pval,\A} \reduce^n(\fml_1)$.
 Therefore, we have $\A^n \models_{\pval} \reduce^n(\fml_{\forall})$.

 Suppose $\A^n \models_{\pval} \reduce^n(\fml_{\forall})$ holds.
 Let $(\runs^{1}, \pathOrder[\runs^1])$ be a path assignment satisfying
 $\domain(\runs^{1}) = \{\vrun_1, \vrun_2, \dots, \vrun_n\}$ and
 for any $i \in \{1,2,\dots,n\}$,
 $\runs^{1}(\vrun_i) \in \Traces(\valuate{\A}{\pval})$.
 By \cref{lemma:composition_lemma},
 there is a path $\trace \in \Traces(\valuate{\A^n}{\pval})$ of $\valuate{\A^n}{\pval}$
 satisfying the following conditions, where
 $\pathOrderToN$ is the monotone surjection  from $\domain(\runs) \times \setN$ to $\setN$.
 \begin{itemize}
  \item For any $\vrun_k \in \domain(\runs)$, $\runs(\vrun_k) = \project{\trace}{k}$ holds.
  \item For any $i \in \setN$, the inverse image $\pathOrderToN^{-1}(\{i\})$ satisfies
        $(\vrun_k, j) \in \pathOrderToN^{-1}(\{i\}) \iff \project{\edge_i}{k} = \edge^k_j$,
        where $\edge_i$ is the $i$-th edge of $\trace$ and $\edge^k_j$ is the $j$-th edge of $\runs^{1}(\vrun_{k})$.
 \end{itemize}
 Since we have $\A^n \models_{\pval} \reduce^n(\fml_{\forall})$,
 for the extension $(\overline{\runs}^{1}, \pathOrder[\overline{\runs}^{1}])$ of $(\emptyruns, \pathOrder[\emptyruns])$
 satisfying $\domain(\runs) = \{\vrun\}$ and $\overline{\runs}^{1}(\vrun) = \trace$,
 there is $(\overline{\runs}^{2}, \pathOrder[\overline{\runs}^{2}])$ satisfying 
 $(\overline{\runs}^{1}, \pathOrder[\overline{\runs}^{1}]) \suffixEqOp (\overline{\runs}^{2}, \pathOrder[\overline{\runs}^{2}])$,
 $\duration(\overline{\runs}^{1} - \overline{\runs}^{2}) \compOp \valuate{\paramOrInt}{\pval}$,
 $(\overline{\runs}^{2}, \pathOrder[\overline{\runs}^{2}]), \Init(\overline{\runs}^{2}(\vrun)),\UpdateCountval{\countval}{\overline{\runs}^{1}}{\overline{\runs}^{2}}, \UpdateRecordval{\recordval}{\overline{\runs}^{1}}{\overline{\runs}^{2}} \models_{\pval,\A} \fml_2$, and
 for any $(\overline{\runs}^{3}, \pathOrder[\overline{\runs}^{3}])$ satisfying $(\overline{\runs}^{1}, \pathOrder[\overline{\runs}^{1}]) \suffixEqOp (\overline{\runs}^{3}, \pathOrder[\overline{\runs}^{3}]) \suffixOp (\overline{\runs}^{2}, \pathOrder[\overline{\runs}^{2}])$,
 $(\overline{\runs}^{3}, \pathOrder[\overline{\runs}^{3}]), \Init(\overline{\runs}^{3}(\vrun)), \UpdateCountval{\countval}{\overline{\runs}^{1}}{\overline{\runs}^{3}}, \UpdateRecordval{\recordval}{\overline{\runs}^{1}}{\overline{\runs}^{3}} \models_{\pval,\A} \fml_1$ holds.
 By \cref{lemma:suffix_lemma},
 for any such $\overline{\runs}^{2}(\vrun)$,
 there is a path assignment $(\runs^{2}, \pathOrder[\runs^{2}])$
 such that
 $\domain(\runs^{1}) = \domain(\runs^{2}) = \{\vrun_1, \vrun_2, \dots, \vrun_n\}$,
 $(\runs^{1}, \pathOrder[\runs^{1}]) \suffixEqOp (\runs^{2}, \pathOrder[\runs^{2}])$, and
 for each $i \in \{1,2,\dots,n\}$, we have
 $\runs^{2}(\vrun_i) = \project{\overline{\runs}^{2}(\vrun)}{i}$.
 For such $(\runs^{2}, \pathOrder[\runs^{2}])$,
 we have $\duration(\runs^{1} - \runs^{2}) \compOp \valuate{\paramOrInt}{\pval}$ and
 $(\runs^{2}, \pathOrder[\runs^{2}]), \Init(\runs^{2}(\vrun_n)),\UpdateCountval{\countval}{\runs^{1}}{\runs^{2}}, \UpdateRecordval{\recordval}{\runs^{1}}{\runs^{2}} \models_{\pval,\A} \fml_2$.
 Let $(\runs^{3}, \pathOrder[\runs^{3}])$ be a path assignment satisfying
 $(\runs^{1}, \pathOrder[\runs^{1}]) \suffixEqOp (\runs^{3}, \pathOrder[\runs^{3}]) \suffixOp (\runs^{2}, \pathOrder[\runs^{2}])$.
 By \cref{lemma:suffix_lemma},
 for any such $(\runs^{3}, \pathOrder[\runs^{3}])$,
 there is a path assignment $(\overline{\runs}^{3}, \pathOrder[\overline{\runs}^{3}])$ satisfying
 $(\overline{\runs}^{1}, \pathOrder[\overline{\runs}^{1}]) \suffixEqOp (\overline{\runs}^{3}, \pathOrder[\overline{\runs}^{3}]) \suffixOp (\overline{\runs}^{2}, \pathOrder[\overline{\runs}^{2}])$, and
 $\runs^{3}(\vrun_i) = \project{\overline{\runs}^{3}(\vrun)}{i}$ for each $i \in \{1,2,\dots,n\}$.
 Since we have $\A^n \models_{\pval} \reduce^n(\fml_{\forall})$, for such $(\overline{\runs}^{3}, \pathOrder[\overline{\runs}^{3}])$,
 we have $(\overline{\runs}^{3}, \pathOrder[\overline{\runs}^{3}]), \Init(\overline{\runs}^{3}(\vrun)), \UpdateCountval{\countval}{\overline{\runs}^{1}}{\overline{\runs}^{3}}, \UpdateRecordval{\recordval}{\overline{\runs}^{1}}{\overline{\runs}^{3}} \models_{\pval,\A} \reduce^n(\fml_1)$.
 Since for each $i \in \{1,2,\dots,n\}$, we have $\runs^{3}(\vrun_i) = \project{\overline{\runs}^{3}(\vrun)}{i}$,
 for such $(\runs^{3}, \pathOrder[\runs^{3}])$, 
 we have $(\runs^{3}, \pathOrder[\runs^{3}]), \Init(\runs^{3}(\vrun_n)), \UpdateCountval{\countval}{\runs^{1}}{\runs^{3}}, \UpdateRecordval{\recordval}{\runs^{1}}{\runs^{3}} \models_{\pval,\A} \fml_1$.
 Therefore, we have $\A \models_{\pval} \fml_{\forall}$.
\end{proof}
}

\subsection{Detailed proof of \cref{theorem:correctness_observers}}\label{section:proof:correctness_observers}

Before proving \cref{theorem:correctness_observers}, we show the completeness of the observers.

\begin{lemma}
 [completeness of the observers]%
 \label{lemma:completeness_observers}
 Let $\A$ be a PTA, $\fullFml$ be an \ExtPTCTL{} formula, and $\pval$ be a parameter valuation for the union of the parameters in $\A$ and~$\fullFml$.
 For any path $\trace$ of $\valuate{\A}{\pval}$,
 there is a path $\trace'$ of $\valuate{\A \productOp \observerof{\fullFml}}{\pval}$ satisfying $\project{\trace'}{\valuate{\A}{\pval}} = \trace$, where
 $\project{\trace'}{\valuate{\A}{\pval}}$ is the projection of $\trace'$ to the path of $\valuate{\A}{\pval}$.
\end{lemma}
\begin{proof}
 The result of the general case is immediate from the result of the base cases (\ie{} $\fullFml$ is one of
$\LAST(\action^1_{\vrun}) - \LAST(\action^2_{\vrun}) \compOp \lterm$,
$\sum_{\action \in \Actions} \alpha_{\action} \COUNT(\action_{\vrun}) \compOp d$, and
$(\sum_{\action \in \Actions} \alpha'_{\action} \COUNT(\action_{\vrun}) \bmod N) \compOp d$), and we focus on them.

 {%
 \newcommand{\riseActions}{\Actions_{\mathrm{rise}}}%
 \newcommand{\currentGuard}{(\action^{1} - \action^{2} \compOp \lterm)[\riseActions \coloneq 0]}%
 \newcommand{\isInitialized}{b}
 For $\fullFml = \LAST(\action^{1}_{\vrun}) - \LAST(\action^{2}_{\vrun}) \compOp \lterm$,
 since $\Label$ is the identity function,
 it suffices to show that
 for any $\loc = (\props, \isInitialized) \in \Loc = \powerset{\{\action^1, \action^2, \fml\}} \times \{\top, \bot\}$ and
 $\props'_{\Actions} \subseteq \{\action^1, \action^2\}$ with $\props \cap \Actions \neq \props'_{\Actions}$ and
 for any clock valuation $\clockval \colon \{\action^{1}, \action^{2}\} \to \setRnn$,
 there is another clock valuation $\clockval'$ and $\loc' = (\props', \isInitialized') \in \Loc$
 satisfying $\props'_{\Actions} = \props' \cap \Actions$ and
 $(\loc, \clockval) \flecheRel (\loc', \clockval')$.
 Let $\props_{\Actions} = \props \cap \Actions$.
 If $\props'_{\Actions} \subsetneq \props_{\Actions}$,
 we have $(\loc, \top, \emptyset, \loc') \in \Edges$, and
 $(\loc, \clockval) \flecheRel (\loc', \clockval')$ holds
 for $\clockval' = \clockval$ and
 $\loc' = (\props'_{\Actions} \cup (\props \cap \{\fullFml\}), \top)$.
 If $\props'_{\Actions} \not\subsetneq \props_{\Actions}$,
 there are transitions from $\loc$ with guards
 $\currentGuard$ and $\neg\currentGuard$ to appropriate $\loc'$.
 Therefore,
 $(\loc, \clockval) \flecheRel (\loc', \clockval')$ is satisfied
 for appropriately chosen $\loc'$ and $\clockval'$.
 }

 Assume
 $\fullFml = \sum_{\action \in \Actions} \alpha_{\action} \COUNT(\action_{\vrun}) \compOp d$. %
 For any location $(p, \approxcountval) \in \powerset{\Actions} \times \{0,1,\dots, d, d+1\}^{\Actions}$ of $\observerof{\fullFml}$,
 we have $\Label((p, \approxcountval)) \cap \Actions = p$.
 Moreover, since for any transition of $\observerof{\fullFml}$, the guard is $\top$,
 it suffices to show that
 for any $p, p' \in \powerset{\Actions}$ and $\approxcountval \in \{0,1,\dots, d, d+1\}^{\Actions}$ with $p \neq p'$,
 there is $\approxcountval' \in \{0,1,\dots, d, d+1\}^{\Actions}$ such that
 $((p, \approxcountval), \top, \emptyset, (p', \approxcountval')) \in \Edges$, which is straightforward from the definition.
 The proof is similar for $\fullFml = (\sum_{\action \in \Actions} \alpha'_{\action} \COUNT(\action_{\vrun}) \bmod N) \compOp d$.
\end{proof}

Proof of \cref{theorem:correctness_observers} is as follows.

\recallResult{theorem:correctness_observers}{\correctnessObserversStatement}

\begin{proof}
 [Proof \LongVersion{of \cref{theorem:correctness_observers}}\ShortVersion{(sketch)}]
 Since the other cases are similar, we only show \ShortVersion{that}\LongVersion{$\A \models_{\pval} \fml_{\exists} \iff \A \productOp \observerof{\fullFml} \models_{\pval} \rmextfrom{\fml_{\exists}}$}\ShortVersion{ $\A \models_{\pval} \fml_{\exists}$ \changed{implies} $\A \productOp \observerof{\fullFml} \models_{\pval} \rmextfrom{\fml_{\exists}}$}, where
 $\fml_{\exists} = \HEUntilFml[\compOp \paramOrInt]{\vrun}{\fml_1}{\fml_2}$.
 Suppose $\A \models_{\pval} \fml_{\exists}$ holds.
 By the semantics of ExtPTCTL,
 for some extension $(\runs^{1}, \pathOrder[\runs^{1}])$ of $(\emptyruns, \pathOrder[\emptyruns])$
 satisfying $\domain(\runs) = \{\vrun\}$ and
 $\runs^{1}(\vrun) \in \Traces(\valuate{\A}{\pval})$,
 there is $(\runs^{2}, \pathOrder[\runs^{2}])$ satisfying
 $(\runs^{1}, \pathOrder[\runs^{1}]) \suffixEqOp (\runs^{2}, \pathOrder[\runs^{2}])$,
 $\duration(\runs^{1} - \runs^{2}) \compOp \valuate{\paramOrInt}{\pval}$,
 $(\runs^{2}, \pathOrder[\runs^{2}]), \Init(\runs^{2}(\vrun)),\UpdateCountval{\countval}{\runs^{1}}{\runs^{2}}, \UpdateRecordval{\recordval}{\runs^{1}}{\runs^{2}} \models_{\pval,\A} \fml_2$, and
 for any $(\runs^{3}, \pathOrder[\runs^{3}])$ satisfying $(\runs^{1}, \pathOrder[\runs^{1}]) \suffixEqOp (\runs^{3}, \pathOrder[\runs^{3}]) \suffixOp (\runs^{2}, \pathOrder[\runs^{2}])$, $(\runs^{3}, \pathOrder[\runs^{3}]), \Init(\runs^{3}(\vrun)), \UpdateCountval{\countval}{\runs^{1}}{\runs^{3}}, \UpdateRecordval{\recordval}{\runs^{1}}{\runs^{3}} \models_{\pval,\A} \fml_1$ holds.
 \LongVersion{By \cref{lemma:completeness_observers}}\ShortVersion{Since the observer $\observerof{\fullFml}$ is complete}, there is a path $\trace^{1}$ of $\A \productOp \observerof{\fullFml}$ satisfying $\project{\trace^{1}}{\valuate{\A}{\pval}} = \runs^{1}(\vrun)$.
 Moreover, by taking a suitable prefix of $\trace^{1}$, for the above $\runs^{2}$,
 there is a path $\trace^{2}$ of $\A \productOp \observerof{\fullFml}$ satisfying
 $\project{\trace^{2}}{\valuate{\A}{\pval}} = \runs^{2}(\vrun)$.
 For any $\trace^{3}$ satisfying $\trace^{1} \suffixEqOp \trace^{3} \suffixOp \trace^{2}$,
 since we have $\project{\trace^{1}}{\valuate{\A}{\pval}} = \runs^{1}(\vrun)$ and $\project{\trace^{2}}{\valuate{\A}{\pval}} = \runs^{2}(\vrun)$,
 by taking a suitable prefix of $\runs^{1}(\vrun)$, %
 there is a path assignment $(\runs^{3}, \pathOrder[\runs^{3}])$ satisfying
 $(\runs^{1}, \pathOrder[\runs^{1}]) \suffixEqOp (\runs^{3}, \pathOrder[\runs^{3}]) \suffixOp (\runs^{2}, \pathOrder[\runs^{2}])$ and
 $\project{\trace^{3}}{\valuate{\A}{\pval}} = \runs^{3}(\vrun)$.
 Overall, by \cref{lemma:correctness_observers},
 for some extension $(\runs^{1}_{\observerof{}}, \pathOrder[\runs^{1}_{\observerof{}}])$ of $(\emptyruns, \pathOrder[\emptyruns])$
 satisfying $\domain(\runs) = \{\vrun\}$ and
 $\runs^{1}_{\observerof{}}(\vrun) \in \Traces(\valuate{\A \productOp \observerof{\fullFml}}{\pval})$,
 there is $(\runs^{2}_{\observerof{}}, \pathOrder[\runs^{2}_{\observerof{}}])$ satisfying
 $(\runs^{1}_{\observerof{}}, \pathOrder[\runs^{1}_{\observerof{}}]) \suffixEqOp (\runs^{2}_{\observerof{}}, \pathOrder[\runs^{2}_{\observerof{}}])$,
 $\duration(\runs^{1}_{\observerof{}} - \runs^{2}_{\observerof{}}) \compOp \valuate{\paramOrInt}{\pval}$,
 $(\runs^{2}_{\observerof{}}, \pathOrder[\runs^{2}_{\observerof{}}]), \Init(\runs^{2}_{\observerof{}}(\vrun)),\UpdateCountval{\countval}{\runs^{1}_{\observerof{}}}{\runs^{2}_{\observerof{}}}, \UpdateRecordval{\recordval}{\runs^{1}_{\observerof{}}}{\runs^{2}_{\observerof{}}} \models_{\pval,\productOp \observerof{\fullFml}} \rmextfrom{\fml_2}$, and
 for any $(\runs^{3}_{\observerof{}}, \pathOrder[\runs^{3}_{\observerof{}}])$ satisfying $(\runs^{1}_{\observerof{}}, \pathOrder[\runs^{1}_{\observerof{}}]) \suffixEqOp (\runs^{3}_{\observerof{}}, \pathOrder[\runs^{3}_{\observerof{}}]) \suffixOp (\runs^{2}_{\observerof{}}, \pathOrder[\runs^{2}_{\observerof{}}])$, we have $(\runs^{3}_{\observerof{}}, \pathOrder[\runs^{3}_{\observerof{}}]), \Init(\runs^{3}_{\observerof{}}(\vrun)), \UpdateCountval{\countval}{\runs^{1}_{\observerof{}}}{\runs^{3}_{\observerof{}}}, \UpdateRecordval{\recordval}{\runs^{1}_{\observerof{}}}{\runs^{3}_{\observerof{}}} \models_{\pval,\A \productOp \observerof{\fullFml}} \rmextfrom{\fml_1}$.
 Therefore, we have $\A \productOp \observerof{\fullFml} \models_{\pval} \rmextfrom{\fml_{\exists}}$.
 \LongVersion{

 Suppose $\A \productOp \observerof{\fullFml} \models_{\pval} \rmextfrom{\fml_{\exists}}$ holds.
 By the semantics of PTCTL,
 for some extension $(\runs^{1}_{\observerof{}}, \pathOrder[\runs^{1}_{\observerof{}}])$ of $(\emptyruns, \pathOrder[\emptyruns])$
 satisfying $\domain(\runs) = \{\vrun\}$ and
 $\runs^{1}_{\observerof{}}(\vrun) \in \Traces(\valuate{\A \productOp \observerof{\fullFml}}{\pval})$,
 there is $(\runs^{2}_{\observerof{}}, \pathOrder[\runs^{2}_{\observerof{}}])$ satisfying
 $(\runs^{1}_{\observerof{}}, \pathOrder[\runs^{1}_{\observerof{}}]) \suffixEqOp (\runs^{2}_{\observerof{}}, \pathOrder[\runs^{2}_{\observerof{}}])$,
 $\duration(\runs^{1}_{\observerof{}} - \runs^{2}_{\observerof{}}) \compOp \valuate{\paramOrInt}{\pval}$,
 $(\runs^{2}_{\observerof{}}, \pathOrder[\runs^{2}_{\observerof{}}]), \Init(\runs^{2}_{\observerof{}}(\vrun)),\UpdateCountval{\countval}{\runs^{1}_{\observerof{}}}{\runs^{2}_{\observerof{}}}, \UpdateRecordval{\recordval}{\runs^{1}_{\observerof{}}}{\runs^{2}_{\observerof{}}} \models_{\pval,\productOp \observerof{\fullFml}} \rmextfrom{\fml_2}$, and
 for any $(\runs^{3}_{\observerof{}}, \pathOrder[\runs^{3}_{\observerof{}}])$ satisfying
 $(\runs^{1}_{\observerof{}}, \pathOrder[\runs^{1}_{\observerof{}}]) \suffixEqOp (\runs^{3}_{\observerof{}}, \pathOrder[\runs^{3}_{\observerof{}}]) \suffixOp (\runs^{2}_{\observerof{}}, \pathOrder[\runs^{2}_{\observerof{}}])$, $(\runs^{3}_{\observerof{}}, \pathOrder[\runs^{3}_{\observerof{}}]), \Init(\runs^{3}_{\observerof{}}(\vrun)), \UpdateCountval{\countval}{\runs^{1}_{\observerof{}}}{\runs^{3}_{\observerof{}}}, \UpdateRecordval{\recordval}{\runs^{1}_{\observerof{}}}{\runs^{3}_{\observerof{}}} \models_{\pval,\A \productOp \observerof{\fullFml}} \rmextfrom{\fml_1}$ holds.
 By \cref{lemma:correctness_observers},
 the satisfaction of the new propositions in
 $\runs^{1}_{\observerof{}}(\vrun)$ and
 $\runs^{2}_{\observerof{}}(\vrun)$ coincide with the satisfaction of the extended predicates in
 $\project{\runs^{1}_{\observerof{}}(\vrun)}{\valuate{\A}{\pval}}$ and
 $\project{\runs^{2}_{\observerof{}}(\vrun)}{\valuate{\A}{\pval}}$, respectively.
 \LongVersion{By \cref{lemma:completeness_observers}}\ShortVersion{Since the observer $\observerof{\fullFml}$ is complete},
 for any $\trace^3$ satisfying
 $\project{\runs^{1}_{\observerof{}}(\vrun)}{\valuate{\A}{\pval}} \suffixEqOp \trace^3 \suffixEqOp \project{\runs^{2}_{\observerof{}}(\vrun)}{\valuate{\A}{\pval}}$,
 there is an path assignment $(\runs^{3}_{\observerof{}}, \pathOrder[\runs^{3}_{\observerof{}}])$ satisfying
 $(\runs^{1}_{\observerof{}}, \pathOrder[\runs^{1}_{\observerof{}}]) \suffixEqOp (\runs^{3}_{\observerof{}}, \pathOrder[\runs^{3}_{\observerof{}}]) \suffixOp (\runs^{2}_{\observerof{}}, \pathOrder[\runs^{2}_{\observerof{}}])$ and
 $\project{\runs^{3}_{\observerof{}}}{\valuate{\A}{\pval}} = \trace^3$.
 Overall, we have $\A \models_{\pval} \fml_{\exists}$.}
\end{proof}
\section{Omitted definition}
\subsection{Syntax of \NFExtHyperPTCTLBDR{}}\label{appendix:NFExtHyperPTCTLBDR}
\begin{definition}[Syntax of \NFExtHyperPTCTLBDR{}]\label{definition:class-BDR08}
 For atomic propositions $\Actions$ and parameters $\Param$,
 the syntax of \NFExtHyperPTCTLBDR{} formulas of the Boolean level~$\Boolean$, the temporal level~$\fml$, and the top level~$\fullFml$ are defined as follows, where
 $\action \in \Actions$,
 $\paramOrInt \in \Param \cup \setN$,
 $n \in \setN$,
 $\param \in \Param$, and
 $\lterm$ is a linear term over~$\Param$:
 \begin{align*}
  \Boolean \Coloneq & \top \mid \action_{\vrun} \mid \COUNT(\action_{\vrun}) \compOp n
  \\
  & \mid \LAST(\action_{\vrun}) - \LAST(\action_{\vrun}) \compOp n \mid \neg \Boolean \mid \Boolean \lor \Boolean
  \\
  \fml \Coloneq & \HEUntilFml[\compOp \paramOrInt]{\vrun_1,\vrun_2,\dots,\vrun_n}{\Boolean}{\Boolean}
  \\
  & \mid \HAUntilFml[\compOp \paramOrInt]{\vrun_1,\vrun_2,\dots,\vrun_n}{\Boolean}{\Boolean}
  \\
 \fullFml \Coloneq & \fml \mid \param \compOp \ltermNN \mid \neg \fullFml \mid \fullFml \lor \fullFml \mid \PExists \param\, \fullFml
 \end{align*}
\end{definition}
}

\end{document}